\documentclass[only,shortleftarrow,shortrightarrow]{LMCS}

\def\dOi{10(2:15)2014}
\lmcsheading%
{\dOi}
{1--51}
{}
{}
{Nov.~28, 2012}
{Jun.~25, 2014}
{}

\ACMCCS{[{\bf Theory of computation}]: Models of Computation---Concurrency---Process Calculi}

\usepackage{graphicx}
\usepackage{amsmath,amssymb
  ,amsthm}
\usepackage{prooftree}
\usepackage{display-figure}
\usepackage{lang, shorthand, uniquesyntax}
\usepackage{txfonts}
\usepackage{hyperref}
\usepackage{mathpartir}
\usepackage{tikz}
\usepackage{extpfeil}

\usepackage{proof}

\newcommand{\justifiedBy}[1]{\using{\textsc{#1}}\justifies}
\renewcommand{\struct}{\ensuremath{\mathrel{\prec}}}

\newtheorem*{lemm}{Sublemma}

\begin{document}

\title[Compositional Reasoning for Channel-Based Concurrent Resource Management]{Compositional Reasoning for Explicit Resource Management in Channel-Based Concurrency\rsuper*}

%

\author[A.~Francalanza]{Adrian Francalanza\rsuper a}	
\address{{\lsuper a}ICT, University of Malta}	
\email{adrian.francalanza@um.edu.mt}  

\author[E.~DeVries]{Edsko DeVries\rsuper b}	
\address{{\lsuper b}Well-Typed LLP, UK}	
\email{edsko@well-typed.com}  

\author[M.~Hennessy]{Matthew Hennessy\rsuper c}	
\address{{\lsuper c}Trinity College Dublin, Ireland}	
\email{matthew.hennessy@cs.tcd.ie}  
\thanks{{\lsuper c}Supported by SFI project SFI 06 IN.1 1898.}




\titlecomment{{\lsuper*}An extended abstract of a preliminary version of the paper has appeared in \cite{DevFraHen09}}

\maketitle


\begin{abstract}
We define a  \pic variant with a costed semantics where channels are treated as
resources that must explicitly be allocated before they are used and can be
deallocated when no longer required. We use a substructural type system
tracking permission transfer to construct coinductive proof techniques for
comparing behaviour and resource usage efficiency of concurrent processes. We establish full abstraction results between our coinductive definitions and a contextual behavioural preorder describing a notion of process efficiency \wrt its management of resources.  We also justify these definitions and respective proof techniques through numerous examples and a case study comparing two concurrent implementations of an extensible buffer.
\end{abstract}


\section{Introduction}
\label{sec:introduction}

We investigate the \emph{behaviour} and \emph{space efficiency} of concurrent programs with
\emph{explicit}  
\emph{resource-management}.  In particular, our study focuses on \emph{channel-passing concurrent programs}: we define a \pic variant, called \picr, where
the only resources available are channels; these channels must explicitly be allocated before they can be used, and can be deallocated when no longer required.  As part of the operational model of the language, channel allocation and deallocation have costs associated with them, reflecting the respective resource usage.  

Explicit resource management is typically desirable in settings where 
resources are \emph{scarce}.  Resource management programming constructs such as explicit deallocation provide fine-grained control over how these resources are used and recycled. By comparison, in automated mechanisms such as garbage collection, unused resources (in this case, memory) tend to remain longer in an unreclaimed state \cite{JonesGC96,GC2011}.    Explicit resource 
management constructs such as memory deallocation also carry advantages over automated mechanisms such as garbage collection techniques
when it comes to \emph{interactive} and \emph{real-time} programs \cite{Bulka:1999:ECP:320041,JonesGC96,GC2011}.  In particular, garbage collection techniques require additional computation to determine otherwise explicit information as to which parts of the memory to reclaim and at what stage of the computation; the associated overheads may lead to uneven performance and intolerable pause periods where the system becomes unresponsive \cite{Bulka:1999:ECP:320041}.  

In the case of channel-passing concurrency  with explicit memory-management, the analysis of the relative behaviour and efficiency of programs is non-trivial for a number of reasons. Explicit memory-management introduces the risk of either premature or multiple deallocation of resources along separate threads of execution; these are more difficult to detect than in single-threaded programs and potentially result in problems such as wild pointers or corrupted heaps which may, in turn, lead to unpredictable, even catastrophic, behaviour \cite{JonesGC96,GC2011}.  It also increases the possibility of memory leaks, which are often not noticeable in short-running, terminating programs but subtly eat up resources over the course of  long-running programs.  In a concurrent settings such as ours, complications relating to the assessment and comparison of resource consumption is further compounded by the fact that the runtime execution of channel-passing concurrent programs can have \emph{multiple interleavings}, is sometimes \emph{non-deterministic} and often \emph{non-terminating}.




\subsection{Scenario:}
\label{sec:scenario}

Consider a setting with two servers, $\ptit{S}_1$  and $\ptit{S}_2$, which repeatedly listen for service requests on channels $\ctit{srv}_1$ and $\ctit{srv}_2$, respectively. Requests send a \emph{return} channel on  $\ctit{srv}_1$ or $\ctit{srv}_2$ which is then used by the servers to service the requests and send back answers, $\textit{v}_1$ and $\textit{v}_2$. A possible implementation for these servers is given in \eqref{eq:3} below, where \piRecX{P} denotes a process $P$ recursing at $w$, $\piIn{\ctit{c}}{x}{P}$ denotes a process inputting on channel \ctit{c} some value that is bound to the variable $x$ in the continuation $P$, and $\piOut{\ctit{c}}{v}{P}$ outputs a value $v$ on channel \ctit{c} and continues as $P$:
\begin{equation}\label{eq:3}
  \ptit{S}_i  \deftri \piRecX{\;\piIn{\ctit{srv}_i}{x}{\; \piOut{x}{\textit{v}_i}}{\;w}} \qquad\qquad \text{for $i \in \sset{1,2}$} 
\end{equation}
 

Clients that need to request service from \emph{both} servers, so as to report back the outcome of both server interactions on some channel, \ctit{ret}, can be programmed in a variety of ways:
\begin{equation}\label{eq:clients}
\begin{split}
  \ptit{C}_0 &\deftri \piRecX{\;\piAll{x_1}{\piAll{x_2}\;\piOut{\ctit{srv}_1}{x_1}\,\piIn{x_1}{y}{\;\piOut{\ctit{srv}_2}{x_2}\,\piIn{x_2}{z}{\;\piOut{\ctit{ret}}{(y,z)}{\;w}}}}} \\
  \ptit{C}_1 &\deftri \piRecX{\,\piAll{x}{\;\piOut{\ctit{srv}_1}{x}\,\piIn{x}{y}{\;\piOut{\ctit{srv}_2}{x}\,\piIn{x}{z}{\piOut{\ctit{ret}}{(y,z)}{\;w}}}}} \\ 
  \ptit{C}_2 &\deftri \piRecX{\piAll{x}{\;\piOut{\ctit{srv}_1}{x}\,\piIn{x}{y}{\;\piOut{\ctit{srv}_2}{x}\,\piIn{x}{z}{\;\piFree{x}{\;\piOut{\ctit{ret}}{(y,z)}{\;w}}}}}} 
\end{split}
\end{equation}
$\ptit{C}_0$ corresponds to an idiomatic \pic client. In order to ensure that it is the sole recipient of the service requests, it creates \emph{two} new 
return channels to communicate with $\ptit{S}_1$  and $\ptit{S}_2$ on  $\ctit{srv}_1$ and $\ctit{srv}_2$, using the command $\piAll{x}{P}$; this command allocates a \emph{new} channel \ctit{c} and binds it to the variable $x$ in the continuation $P$. Allocating a new channel for each service request ensures that the return channel used between the client and server is \emph{private} for the duration of the service, preventing interferences from other parties executing in parallel. 

One important difference between the computational model considered in this paper and  that of
the standard \pic is that channel allocation is an expensive operation \ie it incurs an additional  \emph{(spatial)} cost compared to the other operations. Client $\ptit{C}_1$ attempts to address the inefficiencies of $\ptit{C}_0$ by allocating only \emph{one} additional new channel, and \emph{reusing} this channel for both interactions with the servers.   Intuitively, this channel reuse is valid, \ie it preserves the client-server behaviour  $\ptit{C}_0$ had with servers $\ptit{S}_1$  and $\ptit{S}_2$, because the server implementations above use the received return-channels  \emph{only once}.  This single channel usage guarantees that return channels remain private during the duration of the service,   despite  the reuse from client $\ptit{C}_1$.
 
Client $\ptit{C}_2$ attempts to be more efficient still. More precisely, since our computational model does not assume implicit resource reclamation,  the previous two clients can be deemed as having \emph{memory leaks}: at every iteration of the client-server interaction sequence, $\ptit{C}_0$ and $\ptit{C}_1$  allocate new channels that are not disposed of, even though these channels are never used again in subsequent iterations.  By contrast, $\ptit{C}_2$  deallocates unused channels at the end of each iteration using the construct $\piFree{\ctit{c}}{P}$.

In this work we  develop a formal framework for comparing the behaviour of concurrent processes that explicitly allocate and deallocate channels.  For instance, processes consisting of the servers  $\ptit{S}_1$ and $\ptit{S}_2$ together with any of the clients $\ptit{C}_0$, $\ptit{C}_1$ or $\ptit{C}_2$ should be \emph{related}, on the basis that they exhibit the same behaviour.   In addition, we would like to \emph{order} these systems, based on their relative efficiencies \wrt the (channel) resources used.  We note that there are various, at times contrasting, notions of efficiency that one may consider. For instance, one notion may consider acquiring memory for long periods to be less efficient than repeatedly allocating and deallocating memory; another notion of efficiency could instead focus on minimising the allocation and deallocation operations used, as these as considerably more expensive than other operations.  In this work, we mainly focus on a notion of efficiency that accounts for the relative memory allocations required to carry out the necessary computations. Thus, we would intuitively like to develop a framework yielding  the following preorder, where $\sqsubsetsim$ reads "more efficient than":
\begin{equation} \label{eq:1a}
  \ptit{S}_1 \piParal \ptit{S}_2 \piParal \ptit{C}_2  \quad\sqsubsetsim\quad \ptit{S}_1 \piParal \ptit{S}_2 \piParal \ptit{C}_1 \quad\sqsubsetsim\quad \ptit{S}_1 \piParal \ptit{S}_2 \piParal \ptit{C}_0
\end{equation}
A pleasing property of this preorder would be 
\emph{compositionality}, which  implies that orderings are preserved under larger contexts, \ie for all (valid) contexts $\mathcal{C}[-]$, $P \,\sqsubsetsim\, Q$ implies $\mathcal{C}[P] \,\sqsubsetsim\, \mathcal{C}[Q]$.  
Dually, compositionality would also improve the scalability of our formal framework  since, to show that  $\mathcal{C}[P] \,\sqsubsetsim\, \mathcal{C}[Q]$ (for some context $\mathcal{C}[-]$), it suffices to obtain $P \,\sqsubsetsim\, Q$.  For instance, in the case of \eqref{eq:1a}, compositionality would allow us to factor out the common code, \ie the servers $\ptit{S}_1$ and $\ptit{S}_2$ as the context $\ptit{S}_1 \piParal \ptit{S}_2 \piParal [-]$, and focus on showing that 
\begin{equation}\label{eq:1b}
\ptit{C}_2  \;\sqsubsetsim\;  \ptit{C}_1 \;\sqsubsetsim\; \ptit{C}_0
\end{equation}

\smallskip

\subsection{Main Challenges:}
\label{sec:main-challenges}

The details are however far from straightforward.    To begin with, we need to assess relative program cost over potentially infinite computations.  Thus, rudimentary aggregate measures such as adding up the total computation cost of  processes and comparing this total at the end of the computation is insufficient for system comparisons such as \eqref{eq:1a}. In such cases, a preliminary attempt at a solution 
would be to compare the \emph{relative cost} for \emph{every} server interaction (action): in the sense of \cite{Arun-Kumar:1992}, the preorder would then ensure that every \emph{costed} interaction by the inefficient clients must be matched by a corresponding \emph{cheaper} interaction by the more efficient client (and, dually, costed interactions by the efficient client must be matched by interactions from the inefficient client that are as costly or more). 
\begin{equation}\label{eq:clients-complicated}
  \begin{split}
    \!\!\!\ptit{C}_3 &\deftri \piRecX{\piAll{x_1}{\piAll{x_2}{\;\,\piOut{\ctit{srv}_1}{x_1}\,\piIn{x_1}{y}{\;\,\piOut{\ctit{srv}_2}{x_2}\,\piIn{x_2}{z}{\;\,\piFree{x_1}{\piFree{x_2}{\piOut{\ctit{ret}}{(y,z)}{w}}}}}}}}
  \end{split}
\end{equation}
There are however problems with this approach.  Consider, for instance, $\ptit{C}_3$ defined in \eqref{eq:clients-complicated}.  Even though this client allocates two channels for every  iteration of server interactions, it does not exhibit any memory leaks 
since it deallocates them both at the end of the iteration.  It may therefore be sensible for our preorder to equate $C_3$ with client $C_2$ of \eqref{eq:clients} by having  $\ptit{C}_2  \;\sqsubsetsim\;  \ptit{C}_3$ as well as $\ptit{C}_3  \;\sqsubsetsim\;  \ptit{C}_2$.  However showing $\ptit{C}_3  \;\sqsubsetsim\;  \ptit{C}_2$ would not be possible using the preliminary strategy discussed above, since, $\ptit{C}_3$ must engage in more expensive computation (allocating two channels as opposed to 1) by the time the interaction with the first server 
is carried out.


Worse still, an analysis strategy akin to \cite{Arun-Kumar:1992} would not be applicable for a comparison involving the clients $\ptit{C}_1$ and $\ptit{C}_3$. 
 In spite of the fact that over the course of its entire computation $\ptit{C}_3$ requires less resources than $\ptit{C}_1$, \ie it is more efficient, 
  client $\ptit{C}_3$ appears to be \emph{less efficient} than  $\ptit{C}_1$ after the interaction with the first server on channel  $\ctit{srv}_1$ since, at that stage, it has allocated two new channels as opposed to one.  
However,
$\ptit{C}_1$ becomes  less efficient 
for the remainder of the iteration  since it never deallocates the channel it allocates whereas $\ptit{C}_3$ deallocates both channels.
To summarise, for  comparisons $\ptit{C}_3  \;\sqsubsetsim\;  \ptit{C}_2$ and $\ptit{C}_3  \;\sqsubsetsim\;  \ptit{C}_1$, we need our analysis to allow a process to be \emph{temporarily inefficient} as long as it can  recover later on.




In this paper, we use a costed semantics to define an efficiency preorder to reason about the relative cost of processes over potentially infinite computation, based on earlier work by \cite{Kiehn05,LuttgenV06}.  In particular, we adapt the concept of \emph{cost amortisation} to our setting, 
used  by our preorders to compare  processes that are eventually more efficient than others over the course of their entire computation, but are temporarily less efficient at certain stages of the computation.    

\smallskip
Issues concerning cost assessment are however not the only obstacles tackled in this work; there are also complications associated with the compositionality aspects of our proposed framework.  More precisely, we want to limit our analysis to \emph{safe} contexts, \ie contexts that use resources in a sensible  way, \eg  not deallocating channels while they are still in use.   In addition, we also want to consider behaviour \wrt a subset of the possible safe contexts.  For instance, our clients from \eqref{eq:clients} only exhibit the same behaviour  \wrt servers that $(i)$ accept \emph{(any number of)} requests on channels  $\ctit{srv}_1$ and  $\ctit{srv}_2$ containing a return channel,  which then $(ii)$ use this channel at most \emph{once} to return the requested answer.  We can characterise the interface between the servers and the clients using  fairly standard channel type descriptions adapted from \cite{KobayashiPT:linearity} in \eqref{eq:2a}, where $\chantypW{\tV}$ describes a channel than can be used \emph{any} number of times (\ie the channel-type attribute $\unres$) to communicate values of type \tV, whereas $\chantypO{\tV}$ denotes an \emph{affine} channel (\ie a channel type with attribute $\affine$) that can be used \emph{at most} once to communicate values of type \tV:
\begin{equation}
  \label{eq:2a}
  \ctit{srv}_1 : \chantypW{\chantypO{\tV_1}}, \quad \ctit{srv}_2 : \chantypW{\chantypO{\tV_2}} 
\end{equation}
In the style of \cite{Yoshida07:linearity, hennessy04behavioural}, we could then use this interface to abstract away from the actual server implementations described in \eqref{eq:3} and state that,  \wrt contexts that observe the channel mappings of \eqref{eq:2a}, client $\ptit{C}_2$ is more efficient than $\ptit{C}_1$ which is, in turn, more efficient than $\ptit{C}_0$.  These can be expressed as:
\begin{align}
  \label{eq:4}
  \ctit{srv}_1 : \chantypW{\chantypO{\tV_1}}, \ctit{srv}_2 : \chantypW{\chantypO{\tV_2}} &\,\models\;  \ptit{C}_2  \,\sqsubsetsim\,  \ptit{C}_1 \hspace{3cm}\\
  \label{eq:5}
  \ctit{srv}_1 : \chantypW{\chantypO{\tV_1}}, \ctit{srv}_2 : \chantypW{\chantypO{\tV_2}} &\,\models\;  \ptit{C}_1  \,\sqsubsetsim\,  \ptit{C}_0
\end{align}

Unfortunately, the machinery of \cite{Yoshida07:linearity, hennessy04behavioural} cannot be easily extended  
to our costed analysis because of two main reasons.  First, in order to limit our analysis to safe computation, we would need to show that clients $\ptit{C}_0$, $\ptit{C}_1$ and $\ptit{C}_2$ adhere to the channel usage stipulated 
by the type associations in \eqref{eq:2a}.  However, the channel reuse in  $\ptit{C}_1$ and $\ptit{C}_2$ (an essential feature to attain space efficiency)  requires our analysis to associate potentially different types (\ie $\chantypO{\tV_1}$ and $\chantypO{\tV_2}$) to the same return channel; 
   this channel reuse at different types amounts to a form of \emph{strong update}, a degree of flexibility  not supported by \cite{Yoshida07:linearity, hennessy04behavioural}.   

 Second, the equivalence reasoning mechanisms used in \cite{Yoshida07:linearity, hennessy04behavioural} would be substantially limiting for processes with channel reuse.  More specifically, consider  the slightly tweaked client implementation of $\ptit{C}_2$ below:
\begin{align}\label{eq:6}
  \ptit{C}'_2 & \deftri \piRecX{\piAll{x}{\bigl(\piOutA{\ctit{srv}_1}{x}\,\piParal\,\piIn{x}{y}{(\piOutA{\ctit{srv}_2}{x}\,\piParal\,\piIn{x}{z}{\piFree{x}{\piOut{\ctit{c}}{(y,z)}{X}}})}\bigr)}}
\end{align}
The only difference between the client in \eqref{eq:6} and the original one in \eqref{eq:clients} is that  $\ptit{C}_2$ \emph{sequences} the service requests before the service inputs, \ie $\ldots\piOut{\ctit{srv}_1}{x}\,\piInA{x}{y}{\ldots}$ and $\ldots\piOut{\ctit{srv}_2}{x}\,\piInA{x}{z}{\ldots}$, whereas   $\ptit{C}'_2$ parallelises them, \ie $\ldots\piOutA{\ctit{srv}_1}{x}\,\piParal\,\piInA{x}{y}{\ldots}$ and $\ldots\piOutA{\ctit{srv}_2}{x}\,\piParal\,\piInA{x}{z}{\ldots}$.   Resource-centric type disciplines such as \cite{EFH:uniqueness:journal:12,AliasTypes:SmithWM00} preclude  name matching for a particular resource once all the permissions to use that resource have been used up; this feature is essential to statically reason about a number of basic design patterns for reuse.     For such type settings, it turns out that the client implementations $\ptit{C}_2$  and $\ptit{C}'_2$ exhibit the same behaviour because the return channel used by both clients for \emph{both} server interactions is private, \ie unknown to the respective servers; as a result, the servers cannot answer the service on that channel before it is receives it on  either $\ctit{srv}_1$ or $\ctit{srv}_2$.\footnote{Analogously, in the \pic,
  $\piRes{d}{(\piOutA{c}{d} \piParal \piIn{d}{x}{P})}$
is indistinguishable from
  $\piRes{d}{(\piOut{c}{d}{\piIn{d}{x}{P}})}$}  Through \emph{scope extrusion}, theories such as \cite{Yoshida07:linearity, hennessy04behavioural} can reason adequately about the first server interaction, and relate $\ldots\piOut{\ctit{srv}_1}{x}\,\piInA{x}{y}{\ldots}$ of  $\ptit{C}_2$ with   $\ldots\piOutA{\ctit{srv}_1}{x}\,\piParal\,\piIn{x}{y}{\ldots}$ of $\ptit{C}_2$.   However, they have no mechanism for tracking channel locality post scope extrusion, thereby recovering the information that the return channel \emph{becomes private again} to the client  after the first server interaction (since the servers use up the permission to use the return channel once they reply on it). This prohibits \cite{Yoshida07:linearity, hennessy04behavioural}  from determining that the second server interaction is just an instance of the first server interaction, 
thus failing to relate these two implementations.

In \cite{EFH:uniqueness:journal:12} we developed a substructural type system based around a type attribute describing channel \emph{uniqueness}, and this was used  to statically ensure safe computations for \picr.  In this work, we  weave  this type information into our framework, imbuing it with an operational permission-semantics  to reason compositionally about the costed behaviour of (safe) processes.  More specifically, in \eqref{eq:clients}, when $\ptit{C}_2$ allocates channel $x$, no other process knows about $x$: from a typing perspective, but also operationally, $x$ is \emph{unique} to $\ptit{C}_2$. Client $\ptit{C}_2$ then sends $x$ on $\ctit{srv}_1$ at an \emph{affine} type, which (by definition) limits the server to use $x$ at most once. At this point, from an operational perspective, $x$ is to $\ptit{C}_2$, the entity previously ``owning'' it,  \emph{unique-after-1} (communication) use.  This means that after one communication step on $x$, (the derivative of) $\ptit{C}_2$ recognises that all the other processes apart from it must have used up the single affine permission for $x$, and hence $x$  becomes once again \emph{unique}  to $\ptit{C}_2$. This also means that $\ptit{C}_2$ can safely \emph{reuse} $x$, possibly at a different object type (strong update), or else safely deallocate it. 
 
The concept of affinity is well-known in the process calculus
community.  By contrast, uniqueness (and its duality to affinity)  is used far less. 
In a compositional framework, uniqueness can be used to record 
the guarantee at one end
of a channel corresponding to the restriction 
associated with affine channel usage 
at
the other;  an operational semantics can be defined,  tracking the \emph{permission transfer} of  affine permissions back and forth between processes as a result of communication, addressing the aforementioned complications associated with idioms such as channel reuse.   We employ such an operational (costed) semantics  to define our efficiency preorders for concurrent processes with explicit resource management, based on the notion of amortised cost discussed above.


\subsection{Paper Structure:}
\label{sec:structure}

Section~\ref{sec:language} introduces our language with constructs for explicit memory management and defines a costed semantics for it.  We illustrate issues relating to resource usage in this language through a case study in Section~\ref{sec:case-study}, discussing different implementations for an unbounded buffer. Section~\ref{sec:cost-bisim} develops a labelled-transition system for our language that takes into consideration some representation of the observer and the permissions that are exchanged between the program and the observer; it is a typed transition system similar to \cite{PierceS96,hennessy04behavioural,Hennessy07}, nuanced to the resource-focussed type system of \cite{EFH:uniqueness:journal:12}. Based on this transition system, the section also defines a coinductive cost-based preorder and proves a number of properties about it.  Section~\ref{sec:characterisation} justifies the cost-based preorder by relating it with a behavioural contextual preorder defined in terms of  the reduction semantics of Section~\ref{sec:language}.  Section~\ref{sec:proofs-relat-effic} applies the theory of Section~\ref{sec:cost-bisim} to reason about the efficiency of the unbounded buffer implementations of Section~\ref{sec:case-study}.  Finally, Section~\ref{sec:RelatedWork} surveys related work and Section~\ref{sec:conclusion} concludes.





\section{The Language}
\label{sec:language}

\begin{display}{\picr Syntax}{fig:syntax-ext}
\begin{equation*}
\begin{array}{l@{\hspace{1ex}}r@{\hspace{1ex}}lllllllllll}
P, Q & \bnfdef & \piOut{u}{\vec{v}}{P} & \textsl{(output)}    & \bnfsep & \piIn{u}{\vec{x}}{P} & \textsl{(input)}    \\ 
     & \bnfsep & \piNil                & \textsl{(nil)}       & \bnfsep & \piIf{u=v}{P}{Q}   & \textsl{(match)}    \\       
     & \bnfsep & \piRecX{P}              & \textsl{(recursion)} & \bnfsep & x                    & \textsl{(process variable)} \\              
     & \bnfsep & P \piParal Q            & \textsl{(parallel)}  & \bnfsep & \piAll{x}{P}          & \textsl{(allocate)} \\
      & \bnfsep & \piFree{u}{P}          & \textsl{(deallocate)}\\ 
      
\end{array}
\end{equation*}
\end{display}

\figref{fig:syntax-ext} shows the syntax for our language, the resource \pic, or \picr for short. It has
the standard \pic constructs with the exception of scoping, which is replaced
with primitives for explicit channel allocation, \piAll{x}{P}, and deallocation, \piFree{x}{P}.  The syntax
assumes two separate denumerable sets of channel names $c,d \in \Chans$, and variables
$x, y, z, w \in \Vars$, and lets identifiers $u,\,v$ range over both sets, $\Chans \cup \Vars$.  The input construct,  \piIn{c}{x}{P}, recursion construct, \piRecX{P},  and channel
allocation construct, \piAll{x}{P}, are binders whereby free occurrences of the variables $x$  and $w$ in $P$ are bound.    As opposed to more standard versions of the \pic, we \emph{do not} use name scoping to bind and bookkeep the visibility of names; we shall however use alternative mechanisms to track name knowledge and usage in subsequent development.


\begin{display}{\picr  Reduction Semantics}{fig:reduction-semantics}
\textbf{Contexts}\\
\begin{mathpar}
  \begin{array}{rl}
 \ctxt &\bnfdef \quad \ctxtEmp{-} \quad | \quad \piCtxtPar{\ctxt}{P} \quad | \quad \piCtxtPar{P}{\ctxt}\\
\end{array}\\
\begin{array}{rll}
  \ctxtEmp{\sysSP} & \deftxt \sysSP \\
  \ctxtPar{\ctxtGenN{\context}{\sysSP}}{Q} & \deftxt \sys{\sV'}{(P'\piParal Q)} \quad& \text{if } \ctxtGenN{\context}{\sysSP}=\sys{\sV'}{P'}\\
  \ctxtPar{Q}{\ctxtGenN{\context}{\sysSP}} & \deftxt \sys{\sV'}{(Q\piParal P')} \quad& \text{if } \ctxtGenN{\context}{\sysSP}=\sys{\sV'}{P'}\\\\
\end{array}
\end{mathpar}
\textbf{Structural Equivalence}\\
\begin{equation*}
\begin{array}{l@{\hspace{2ex}}l@{\piStructS}l@{\hspace{4ex}}l@{\hspace{2ex}}r@{\piStructS}l@{\hspace{4ex}}l@{\hspace{2ex}}l@{\piStructS}l}
\rtit{sCom}  & P \piParalL Q       & Q\piParalL P & \rtit{sAss} & P \piParalL (Q \piParalL R) & (P \piParalL Q) \piParalL R & 
\rtit{sNil} & P \piParalL \piNil & P       \\\\ 
\end{array}
\end{equation*}
\textbf{Reduction Rules}\\
\begin{mathpar}
\begin{prooftree}
\justifiedBy{\rtit{rCom}}
\sys{\sV,c 
}{\piOut{c}{\vec{d}}{P} \piParal \piIn{c}{\vec{x}}{Q}}\piRedCost{0}  \sys{\sV,c 
}{P\piParal Q\subC{\,\vec{d}\,}{\,\vec{x}\,}}
\end{prooftree}
\\
\begin{prooftree}
\justifiedBy{\rtit{rThen}}
\sys{\sV,c 
}{\piIf{c=c}{P}{Q}}\piRedCost{0} \sys{\sV,c 
}{P}
\end{prooftree}
\\
\begin{prooftree}
\justifiedBy{\rtit{rElse}}
\sys{\sV,c,d 
}{\piIf{c=d}{P}{Q}}\piRedCost{0} \sys{\sV,c,d 
}{Q}
\end{prooftree}
 \\
\begin{prooftree}
\strut
\justifiedBy{\rtit{rRec}}
\sysS{\piRecX{P}}\piRedCost{0} \sysS{P\subC{\piRecX{P}}{w}}
\end{prooftree} \qquad
\begin{prooftree}
    P \piStruct P' \quad \sysS{P'} \piRedCost{k}  \sysS{Q'} \quad Q'\piStruct Q
    \justifiedBy{\rtit{rStr}}
    \sysSP\piRedCost{k} \sysS{Q} 
  \end{prooftree}
\\
  \begin{prooftree}
    \justifiedBy{\rtit{rAll}}
    \sys{\sV
    }{\piAll{x}{P}}\piRedCost{+1} \sys{\sV,c 
   }{P\subC{c}{x}}
  \end{prooftree}
\qquad
  \begin{prooftree}
    \strut
    \justifiedBy{\rtit{rFree}}
    \sys{\sV,c 
    }{\piFree{c}{P}}\piRedCost{-1} \sys{\sV
    }{P} 
  \end{prooftree}\\
\end{mathpar}
\textbf{Reflexive Transitive Closure}\\
\begin{equation*}
\begin{prooftree}
\strut
\justifies
\sysSP \piRedCost{0}^\ast \sysSP
\end{prooftree}
\qquad
\begin{prooftree}
\sysSP \piRedCost{k}^\ast \sys{\sV'}{P'} \qquad
\sys{\sV'}{P'} \piRedCost{l} \sys{\sV''}{P''}
\justifies
\sysSP \piRedCost{{k+l}}^\ast \sys{\sV''}{P''}
\end{prooftree}
\end{equation*}
\end{display}

\picR processes run in a resource environment, ranged over by $\sV, \sVV$, representing predicates over channel names stating whether a channel is allocated or not.  We find it convenient to denote such functions as a list  of channels representing the set channels that are allocated, \eg the list $c,d$ denotes the set $\sset{c,d}$, representing the resource environment returning \textit{true} for channels $c$ and $d$ and \textit{false} otherwise - in this representation, the order of the channels in the list is unimportant, but duplicate channels are disallowed; as shorthand, we also write $\sV,c$ to denote $\sV\cup\sset{c}$ whenever $c\not\in\sV$. 
 In this paper we consider only resource environments with an
  \emph{infinite} number of deallocated channels, \ie $\sV$ is a total
  function.  Models with finite resources can be easily accommodated
  by making $\sV$ partial; this also would entail a slight change in
  the semantics of the allocation construct, which could either block
  or fail whenever there are no deallocated resources left.  Although
  interesting in its own right, we focus on settings with infinite
  resources as it lends itself better to the analysis of resource
  efficiency that follows.

We refer to the pair \sysSP, consisting of a resource environment \sV\ and a \emph{closed} process\footnote{A closed process has no free variables.  Note that the absence of name binders \ie no name scoping, means that all names are free.} $P$ as a
\emph{system}; note that \emph{not all} free names in $P$ need to be allocated \ie present in \sV: intuitively, any name $c$ used by $P$ and $c \not\in \sV$ represents a \emph{dangling pointer}. Contexts consist of parallel composition of processes; they are however defined over systems, through the grammar and the respective definition at the top of  \figref{fig:reduction-semantics}. The reduction relation is
defined as the least \emph{contextual} relation over systems satisfying the rules in
\figref{fig:reduction-semantics}.   More specifically our reduction relation leaves the following  rule implicit:
\begin{mathpar}
  \begin{prooftree}
      \sysS{P} \;\piRedCost{k}\;  \sysS{Q} 
    \justifiedBy{\rtit{rCtx}}
    \ctxtGen{\sysSP}\;\piRedCost{k}\; \ctxtGen{\sysS{Q}} 
  \end{prooftree}
\end{mathpar}
 Rule (\rtit{rStr}) extends reductions to structurally equivalent processes, $P\piStruct Q$, \ie processes that are identified up to superfluous \piNil\ processes, and commutativity/associativity of parallel composition (see the structural equivalence rules \figref{fig:reduction-semantics}).  

Most rules follow those of the standard \pic, \eg  (\rtit{rRec}), with the exception of those involving resource handling. For instance, the rule for communication (\rtit{rCom}) requires the
communicating channel to be \emph{allocated}. Allocation (\rtit{rAll}) chooses a deallocated 
channel, allocates it, and substitutes it for the bound variable of the
allocation construct.\footnote{The expected side-condition $c \!\not\in\! \sV$ is implicit in the notation $(\sV,c)$ used in the system $\sys{\sV,c}{P\subC{c}{x}}$ to which it reduces, since $c$ cannot be present in \sV\ for $\sV,c$ to be valid.}  
Deallocation (\rtit{rFree}) changes the states of a channel from allocated to
deallocated, making it available for future allocations. The rules are
annotated with a cost reflecting resource usage; allocation has a cost of $+1$, 
deallocation has a (negative) cost of $-1$ while the other reductions carry no cost, \ie $0$. \figref{fig:reduction-semantics} also shows the natural definition of the reflexive
transitive closure of the costed reduction relation.  In what follows, we use $k, l\in \Ints$ as 
integer
metavariables to range over costs. 

\begin{exa}\label{ex:bad-behaviour} The following reduction sequence illustrates potential unwanted behaviour resulting from resource mismanagement:
\begin{align}
  \label{eq:54}
   &\sys{M, c 
  }{\piFree{c}{(\piOutA{c}{\num{1}} \piParal \piIn{c}{x}{P})} \;\piParal\; \piAll{y}{ (\piOutA{y}{\num{42}} \piParal \piIn{y}{z}{Q})}} & \piRedCost{-1}  \\
   \label{eq:1}
   &\sys{M\phantom{, c}
   }{\piOutA{c}{\num{1}} \piParal \piIn{c}{x}{P} \;\piParal\; \piAll{y}{ (\piOutA{x}{\num{42}} \piParal \piIn{x}{z}{Q})}} & \piRedCost{+1} \\
   \label{eq:2}
   &\sys{M, c 
   }{\piOutA{c}{\num{1}} \piParal \piIn{c}{x}{P} \;\piParal\; \piOutA{c}{\num{42}} \piParal \piIn{c}{z}{Q}}
\end{align}
Intuitively, allocation should yield ``fresh'' channels \ie channels that are not in use by any active process. This assumption is used by the right process in system \eqref{eq:54}, $\piAll{y}{ (\piOutA{y}{\num{42}} \piParal \piIn{y}{z}{Q})}$,  to carry out a \emph{local} communication, sending the value  \num{42} on some local channel $y$ that no other process is using.  However,  the premature deallocation of the channel $c$ by the left process in \eqref{eq:54}, $\piFree{c}{(\piOutA{c}{\num{1}} \piParal \piIn{c}{x}{P})}$,  allows channel $c$ to be reallocated by the right process in the subsequent reduction, \eqref{eq:1}. This may then lead to unintended behaviour since we may end up with interferences when communicating on $c$ in the residuals of the left and right processes, \eqref{eq:2}.\footnote{Operationally, we do not describe errors that may result from attempted communications on deallocated channels (we do not have error values).  This may occur after  reduction \eqref{eq:54}, if the residual of the left process communicate on channel $c$. Rather, communications on deallocated channels are blocked.}
\hfill$\Box$
\end{exa}

\begin{display}{Type Attributes and Types}{fig:types}
%
\begin{gather*}
\begin{array}{lrllllll}
\aV & \bnfdef & \unrestricted & \text{(unrestricted)} 
    &\bnfsepp   \affine       & \text{(affine)}
     & \bnfsepp \unique{i}    & \text{(unique after $i$ steps)}\\[1em]  
\tV & \bnfdef &  \tVV & \text{(channel type)} 
     &\bnfsepp \proctyp                 & \text{(process type)} \\
\tVV & \bnfdef & \chantyp{\vec{\tVV}}{\aV}  & \text{(channel)}   &\bnfsepp  \chanrec{X}{\tVV}  & \text{(recursion)}   &\bnfsepp  X  & \text{(variable)}  
\end{array}
\end{gather*}
\end{display}

In \cite{EFH:uniqueness:journal:12} we defined a type system that precludes unwanted behaviour
such as in \exref{ex:bad-behaviour}.  The type syntax is shown in \figref{fig:types}. 
 The main type entities are \emph{channel types}, denoted as
$\chantyp{\tVVlst}{a}$, where \emph{type attributes} $a$ range over
\begin{itemize}
\item \affine, for affine, imposing a restriction/obligation on
  usage;
\item $\unique{i}$, for unique-after-$i$ usages ($i \in
  \mathbb{N}$), providing guarantees on usage; 
\item $\unrestricted$, for unrestricted channel usage without restrictions or guarantees.
\end{itemize}
Uniqueness typing can be seen as dual to affine typing \cite{harrington:uniquenesslogic}, and in \cite{EFH:uniqueness:journal:12} we make use of this duality to keep track of uniqueness across channel-passing parallel processes: an attribute $\unique{i}$ typing an endpoint of a channel $c$ accounts for (at most) $i$ instances of affine attributes typing endpoints of that same channel.   

A channel type $\chantyp{\tVVlst}{a}$ also describes the type of the values that can be communicated on that channel, \tVVlst, which denotes a list of types $\tVV_1,\ldots,\tVV_n$ for $n\in\Nats$; when $n=0$, the type list is an empty list and we simply write $\chantyp{}{a}$.   Note
the difference between $\chantyp{\tVVlst}{\affine}$, \ie a channel with an affine usage restriction,  and
$\chantyp{\tVVlst}{\unique{1}}$,  \ie a channel with a unique-after-1 usage guarantee. We denote fully unique
channels as $\chantyp{\tVVlst}{\uniqueNow}$ in lieu of
$\chantyp{\tVVlst}{\unique{0}}$. 

The type syntax also assumes a denumerable set of
type variables $X,Y$, bound by the recursive type construct \chanrecX{\tVV}.  In what follows, we restrict our attention to \emph{closed, contractive} types, where every type variable is bound and appears within a channel constructor $\chantyp{-}{\aV}$; this ensures that channel types such as $\chanrec{X}{X}$ are avoided.  We assume an equi-recursive interpretation for our recursive types \cite{Pierce:2002:TPL} (see \rtit{tEq} in \figref{fig:typingrules}), characterised as the least type-congruence satisfying rule \rtit{eRec} in   \figref{fig:typingrules}.

\begin{display}{Typing processes}{fig:typingrules}
\textbf{Logical rules}\\[1em]
\begin{tabular}{llll}
\quad

&
\begin{prooftree}
\tprocP{\env, \envmap{u}{\chantyp{\tVlst}{\aV-1}}}  
\justifiedBy{tOut}
\tproc{\env, \envmap{u}{\chantyp{\tVlst}{\aV}}, \overrightarrow{\envmap{v}{\tV}}}{\piOut{u}{\vec{v}}{P}}
\end{prooftree}
&
\begin{prooftree}
\tprocP{\env, \envmap{u}{\chantyp{\tVlst}{\aV-1}}, \overrightarrow{\envmap{x}{\tV}}} 
\justifiedBy{tIn}
\tproc{\env, \envmap{u}{\chantypA{\tVlst}}}{\piIn{u}{\vec{x}}{P}}  
\end{prooftree}\;\quad
 &
\begin{prooftree}
\tprocP{\env_1} \quad
\tproc{\env_2}{Q} 
\justifiedBy{tPar}
\tproc{\env_1, \env_2\,}{\,P \piParal Q}
\end{prooftree}
\\[2em]
&
\begin{prooftree}
u, v \in \Gamma \quad
\tprocE{P} \quad
\tprocE{Q}
\justifiedBy{tIf}
\tprocE{\piIf{u=v}{P}{Q}}
\end{prooftree} \;\quad
 &
\begin{prooftree}
\tprocP{\env^\unrestricted, \envmap{x}{\proctyp}} 
\justifiedBy{tRec}
\tproc{\env^\unrestricted}{\piRecX{P}} 
\end{prooftree} 
 &
\begin{prooftree}
\strut
\justifiedBy{tVar}
\tproc{\envmap{x}{\proctyp}}{x}
\end{prooftree}
\\[2em]
&
\begin{prooftree}
\tprocEP
\justifiedBy{tFree}
\tproc{\env,\envmap{u}{\chantyp{\vec{\tV}}{\uniqueNow}}}{\piFree{u}{P}}
\end{prooftree}
&
\begin{prooftree}
\tprocP{\env,\envmap{x}{\chantyp{\vec{\tV}}{\uniqueNow}}}
\justifiedBy{tAll}
\tprocE{\piAll{x}{P}}
\end{prooftree}
&
\begin{prooftree}
\strut
\justifiedBy{tNil}
\tproc{\emptyset}{\piNil}
\end{prooftree}   
\\[2em]
&
\begin{prooftree}
\tproc{\env'}{P} \qquad
\env \struct \env'
\justifiedBy{tStr}
\tproc{\env}{P}
\end{prooftree}   
 \\[2em]
\end{tabular}

where $\env^\unrestricted$ can only contain unrestricted
assumptions and all bound variables are fresh.\\[1em]
\textbf{Structural rules}
$(\struct)$ is the least reflexive transitive relation satisfying \\
\begin{mathpar}
\begin{prooftree}
\tV = \tV_1 \circ \tV_2 
\justifiedBy{tCon}
  \env, \envmap{u}{\tV}
\struct
  \env, \envmap{u}{\tV_1}, \envmap{u}{\tV_2}
\end{prooftree} \qquad
\begin{prooftree}
\tV = \tV_1 \circ \tV_2 
\justifiedBy{tJoin}
  \env, \envmap{u}{\tV_1}, \envmap{u}{\tV_2} 
\struct 
  \env, \envmap{u}{\tV}
\end{prooftree}
\qquad
\begin{prooftree}
\tV_1 \sim \tV_2
\justifiedBy{tEq}
  \env,\envmap{u}{\tV_1}
\struct
  \env,\envmap{u}{\tV_2}
\end{prooftree}
\\ 
\begin{prooftree}
\strut
\justifiedBy{tWeak}
  \env, \envmap{u}{\tV}
\struct
  \env
\end{prooftree} 
\qquad
\begin{prooftree}
\tV_1 \subtype \tV_2
\justifiedBy{tSub}
  \env,\envmap{u}{\tV_1}
\struct
  \env,\envmap{u}{\tV_2}
\end{prooftree} \qquad
\begin{prooftree}
\strut
\justifiedBy{tRev}
  \env,\envmap{u}{\chantyp{\vec{\tV_1}}{\uniqueNow}}
\struct
  \env,\envmap{u}{\chantyp{\vec{\tV_2}}{\uniqueNow}}
\end{prooftree}
\end{mathpar}\\[1em]
$\begin{array}{cc}
\textbf{Equi-Recursion}
& \textbf{Counting channel usage}\\[0.5em]
\qquad\begin{prooftree}
\phantom{\aV_1 \subtype \aV_2}
\justifiedBy{eRec}
\chanrecX{\tVV} \sim \tVV\subC{\chanrecX{\tVV}}{X}
\end{prooftree} \qquad
&
  \qquad\envmap{c}{\chantyp{\vec{\tV}}{\aV-1}} \deftxt
  \begin{cases}
    \varepsilon \qquad \textit{ (empty list)}
    & \text{if }\,\aV=\affine\\
    \envmap{c}{\chantyp{\vec{\tV}}{\unrestricted}} & \text{if }\,\aV=\unrestricted\\
    \envmap{c}{\chantyp{\vec{\tV}}{\unique{i}}} &  \text{if }\,\aV=\unique{i + 1}
  \end{cases}
\end{array}$
\\[1em]
\textbf{Type splitting}
\begin{equation*}
\begin{prooftree}
\justifiedBy{pUnr}
\chantyp{\tVlst}{\unrestricted} = \chantyp{\tVlst}{\unrestricted} \circ \chantyp{\tVlst}{\unrestricted}
\end{prooftree} \qquad
\begin{prooftree}
\justifiedBy{pProc}
\proctyp = \proctyp \circ \proctyp
\end{prooftree} \qquad 
\begin{prooftree}
\justifiedBy{pUnq}
\chantyp{\tVlst}{\unique{i}} = \chantyp{\tVlst}{\affine} \circ \chantyp{\tVlst}{\unique{i+1}}
\end{prooftree}
\end{equation*}\\[1em]
\textbf{Subtyping}
\begin{equation*}
\begin{prooftree}
\phantom{\aV_1 \subtype \aV_2}
\justifiedBy{sIndx}
\unique{i} \subtype \unique{i+1}
\end{prooftree} \qquad 
\begin{prooftree}
\phantom{\aV_1 \subtype \aV_2}
\justifiedBy{sUnq}
\unique{i} \subtype \unrestricted
\end{prooftree} \qquad 
\begin{prooftree}
\phantom{\aV_1 \subtype \aV_2}
\justifiedBy{sAff}
\unrestricted \subtype \affine 
\end{prooftree} \qquad 
\begin{prooftree}
\aV_1 \subtype \aV_2
\justifiedBy{sTyp}
\chantyp{\tVlst}{\aV_1} \subtype \chantyp{\tVlst}{\aV_2}
\end{prooftree} 
\end{equation*}
\end{display}

\begin{gather*}
\begin{prooftree}
\tprocEP \qquad
\dom(\env) \subseteq \sV \qquad
\Gamma \text{ is consistent}
\justifiedBy{tSys}
\tproc{\env}{\sysSP}
\end{prooftree}  
\end{gather*}

  The rules for typing processes  are given in \figref{fig:typingrules} and take the usual shape $\tproc{\env}{P}$ stating that
process $P$ is well-typed with respect to the environment $\env$, a list of pairs of identiﬁers and
types.  Systems are typed according to (\rtit{tSys}) above: a system $\sys{M}{P}$ is
well-typed under $\env$ if $P$ is well-typed \wrt $\env$, $\env \vdash P$, and $\env$ only contains assumptions for
channels that have been allocated, $\dom(\env) \subseteq \sV$. This restricts channel usage in $P$ to allocated channels and is key for ensuring safety. 

\smallskip

In \cite{EFH:uniqueness:journal:12}, typing environments are multisets of pairs
of identifiers and types; we do not require them to be partial functions.  However, 
the (top-level) typing rule for systems (\rtit{tSys}) requires that the typing
environment is \emph{consistent}.  A typing environment is consistent if
whenever it contains multiple assumptions about a channel, then these
assumptions can be derived from a \emph{single assumption} using the structural rules
of the type system (see the structural rule \rtit{tCon} and the splitting rule \rtit{pUnq} in \figref{fig:typingrules}).

\begin{defi}[Consistency]\label{def:consistency}
A typing environment $\env$ is \emph{consistent} if there is a partial map
$\env'$ such that $\env' \struct \env$. 
\end{defi}

The environment structural rules, $\env_1\struct\env_2$, defined in \figref{fig:typingrules},  govern the way type environments are syntactically manipulated.  For instance, rules \rtit{tCon} and \rtit{tJoin} state that type assumptions for the same identifier can be split or joined according to the type splitting relation $\tV = \tV_1 \circ \tV_2$, also defined in \figref{fig:typingrules}: apart from standard splitting of unrestricted channels, \rtit{pUnr}, and process types, \rtit{pProc}, we note that a unique-after-$i$ channel may be split into a unique-after-$(i+1)$ channel and an affine channel; we also note that affine channels are \emph{never} split. The environment structural rules also allow for weakening, \rtit{tWeak}, equi-recursive manipulation of types, \rtit{tEq} and \rtit{eRec}, and subtyping, \rtit{tSub}; the latter rule is defined in terms of the subtyping relation also stated in \figref{fig:typingrules} (bottom) where, for instance, an unrestricted channel can be used instead of an affine channel (that can be used at most once).  The key novel structural rule is however \rtit{tRev}, which allows us to change (revise) the object type of a channel whenever we are guaranteed that the type assumption for that identifier is unique.  These rules are recalled from \cite{EFH:uniqueness:journal:12} and the reader is encouraged to consult that document for more details.   

The consistency condition of \defref{def:consistency} ensures that there is no mismatch in the duality between the guarantees of unique types and the restrictions of affine types, which allows sound compositional type-checking by our type system.  For instance, consistency rules out environments such as
\begin{equation}
\envmap{c}{\chantyp{\tVV}{\uniqueNow}}, \envmap{c}{\chantypO{\tVV}}\label{eq:1:lang}
\end{equation}
where a process typed under the guarantee that a channel $c$ is unique now, \envmap{c}{\chantyp{\tVV}{\uniqueNow}},   contradicts the  fact that  some other process may be typed under the affine usage  allowed by the assumption  $\envmap{c}{\chantypO{\tVV}}$. For similar reasons, consistency also rules out environments such as 
\begin{equation}
\envmap{c}{\chantyp{\tVV}{\uniqueNow}}, \envmap{c}{\chantypW{\tVV}}\label{eq:2:lang}
\end{equation}
However, it does not rule out environments such as  \eqref{eq:3:lang}
even though the guarantee provided by \envmap{c}{\chantypUU{\tVV}{2}} is too conservative: it states that channel $c$ will become unique after \emph{two} uses but, in actual fact,  it becomes unique after one use since  the (top-level) environment contains only \emph{one}  other affine type assumption, \envmap{c}{\chantypO{\tVV}}, that other processes can be typed at.
\begin{equation}
\envmap{c}{\chantypUU{\tVV}{2}}, \envmap{c}{\chantypO{\tVV}}\label{eq:3:lang}
\end{equation}
A less conservative uniqueness typing guarantee would therefore be \envmap{c}{\chantypUU{\tVV}{1}} as shown in \eqref{eq:3a:lang} below; this environment constitutes another case of a consistent environment allowed by Definition~\ref{def:consistency}.  
\begin{equation}
\envmap{c}{\chantypUU{\tVV}{1}}, \envmap{c}{\chantypO{\tVV}}\label{eq:3a:lang}
\end{equation}

\medskip

The type system is \emph{substructural}, implying that typing assumptions can be used \emph{only once} during typechecking \cite{Pierce:2004:ATT}.
This is clearly manifested in the output and input rules, \rtit{tOut} and \rtit{tIn} in  \figref{fig:typingrules}. In fact, using the operation $\envmap{c}{\chantyp{\vec{\tV}}{\aV-1}}$ (see\footnote{This operation on type assumptions, $\envmap{c}{\chantyp{\vec{\tV}}{\aV-1}}$, defined in \figref{fig:typingrules}, describes the cases where, when using an affine type assumption to typecheck a process, the continuation of the process in the rule premise is typed without that assumption (the operation returns no type assumption), whereas when using an unrestricted or unique-after-$i$ assumptions, the  premise judgement use \wrt (new) unrestricted and unique-after-$(i-1)$ assumptions, respectively. Note that the operation $\envmap{c}{\chantyp{\vec{\tV}}{\aV-1}}$ is not defined for $\aV=\uniqueNow$. See \cite{EFH:uniqueness:journal:12} for more detail.} \figref{fig:typingrules}), rule \rtit{tOut} collapses three different possibilities for typing output processes, which could alternatively have been expressed as the three separate typing rules in \eqref{eq:7:lang}.  
\begin{equation}\label{eq:7:lang}
  \begin{split}
    &\begin{prooftree}
      \tprocEP \justifiedBy{tOutA} \tproc{\env,\,
        \envmap{u}{\chantypO{\tVlst}},\,
        \overrightarrow{\envmap{v}{\tV}}\;}{\;\piOut{u}{\vec{v}}{P}}
    \end{prooftree}
    \qquad\qquad
    \begin{prooftree}
      \tprocP{\env,\, \envmap{u}{\chantypW{\tVlst}}}
      \justifiedBy{tOutW} \tproc{\env,\,
        \envmap{u}{\chantypW{\tVlst}},\,
        \overrightarrow{\envmap{v}{\tV}}\;}{\;\piOut{u}{\vec{v}}{P}}
    \end{prooftree}
    \\[0.5em]
    &\hspace{3cm}\begin{prooftree}
      \tprocP{\env,\, \envmap{u}{\chantyp{\tVlst}{\unique{i}}}}
      \justifiedBy{tOutU} \tproc{\env,\,
        \envmap{u}{\chantyp{\tVlst}{\unique{i+1}}},\,
        \overrightarrow{\envmap{v}{\tV}}\;}{\;\piOut{u}{\vec{v}}{P}}
    \end{prooftree}
  \end{split}
\end{equation}
Rule \rtit{tOutA} states that an output of values $\vec{v}$ on channel $u$ is allowed if  the type environment has an \emph{affine} channel-type assumption for that channel, \envmap{u}{\chantypO{\tVlst}}, and  the corresponding type assumptions for the values communicated, $\overrightarrow{\envmap{v}{\tV}}$, match the object type of the affine channel-type assumption, \tVlst;  in the rule premise, the continuation $P$ must also be typed \wrt the \emph{remaining} assumptions in the environment,  \emph{without} the assumptions consumed by the conclusion.  Rule \rtit{tOutW} is similar, but permits outputs on $u$ for  environments with an \emph{unrestricted} channel-type assumption for that channel, \envmap{u}{\chantypW{\tVlst}}. The continuation  $P$ is typechecked \wrt the remaining assumptions and a \emph{new} assumption, $\envmap{u}{\chantypW{\tVlst}}$; this assumption is identical to the one consumed in the conclusion, so as to model the fact that uses of channel $u$ are unrestricted. Rule \rtit{tOutU} is again similar,  but it allows outputs on channel $u$ for a  \emph{``unique after $i\!+\!1$''} channel-type assumption;  in the premise of the rule,  $P$ is typechecked \wrt the remaining assumptions and a \emph{new} assumption $\envmap{u}{\chantyp{\tVlst}{\unique{i}}}$, where $u$ is now \emph{unique after $i$} uses. Analogously, the input rule,  \rtit{tIn}, also encodes three input cases (listed below):
\begin{equation} \label{eq:8:lang}
\begin{split}
  &\begin{prooftree}
    \tprocP{\env, 
      \overrightarrow{\envmap{x}{\tV}}} \justifiedBy{tInO} \tproc{\env,
      \envmap{u}{\chantypO{\tVlst}}}{\piIn{u}{\vec{x}}{P}}
    \end{prooftree}
    \qquad\;
    \begin{prooftree}
    \tprocP{\env, \envmap{u}{\chantypW{\tVlst}},
      \overrightarrow{\envmap{x}{\tV}}} 
    \justifiedBy{tInW} 
    \tproc{\env,
      \envmap{u}{\chantypW{\tVlst}}}{\piIn{u}{\vec{x}}{P}}
  \end{prooftree}
 \qquad
  \begin{prooftree}
    \tprocP{\env, \envmap{u}{\chantyp{\tVlst}{\unique{i}}},
      \overrightarrow{\envmap{x}{\tV}}} 
    \justifiedBy{tInU} 
    \tproc{\env,
      \envmap{u}{\chantyp{\tVlst}{\unique{i+1}}}}{\piIn{u}{\vec{x}}{P}}
  \end{prooftree}\quad
\end{split}
\end{equation}
Parallel composition (\rtit{tPar}) enforces the substructural treatment of type assumptions, by ensuring that type assumptions are used by either the left
process or the right, but not by both. However,  some type assumption can be \emph{split} using contraction, \ie rules (\rtit{tStr}) and (\rtit{tCon}). For example, an
assumption $c : \chantyp{\tVlst}{\unique{i}}$ can be split as $c
: \chantyp{\tVlst}{\affine}$ and $c :
\chantyp{\tVlst}{\unique{i+1}}$---see (\rtit{pUnq}).

The rest of the rules in Figure~\ref{fig:typingrules}
are fairly straightforward. Even though these typing rules do not require \env to be consistent,  the consistency requirement at the top level typing judgement (\rtit{tSys})  ensures that whenever a process is typed \wrt a unique assumption for a
channel, $\chantypU{\tVlst}$, no other process has access to that channel. It
can therefore safely deallocate it (\rtit{tFree}), or change the object type of
the channel (\rtit{tRev}). Dually, when a channel is newly allocated it
is assumed unique (\rtit{tAll}).  Note also that name matching is only
permitted when channel permissions are owned, $u, v \in \Gamma $ in
(\rtit{tIf}).  Uniqueness can therefore also be thought of as ``freshness'', a
claim we substantiate further in Section~\ref{sec:Bisimulation}.

In \cite{EFH:uniqueness:journal:12} we prove the usual subject reduction and progress lemmas
for this type system, given an (obvious) error relation. 

\begin{exa}  All client implementations discussed in Section~\ref{sec:introduction} typecheck \wrt the type environment $$\env=\emap{\ctit{srv}_1}{\chantypW{\chantypO{\tV_1}}}, \emap{\ctit{srv}_2}{\chantypW{\chantypO{\tV_2}}}, \emap{\ctit{ret}}{\chantypW{\tV_1,\tV_2}}.$$  For instance, to typecheck $\ptit{C}_2$ from \eqref{eq:clients}, we can apply the typing rules \rtit{tRec} and \rtit{tAll} from Figure~\ref{fig:typingrules} to obtain the typing sequent:
  \begin{equation}\label{eq:7:ts}
    \tproc{\env,\,\emap{w}{\proctyp},\, \emap{x}{\chantypU{\tV_1}}\;}{\;\piOut{\ctit{srv}_1}{x}\,\piIn{x}{y}{\;\piOut{\ctit{srv}_2}{x}\,\piIn{x}{z}{\;\piFree{x}{\;\piOut{\ctit{ret}}{(y,z)}{\;w}}}}}
  \end{equation}
  Using the environment structural rules (\ie \rtit{tCon}) we can split the type assumption for $x$:
  \begin{equation*}
     \env,\,
      \emap{w}{\proctyp},\, \emap{x}{\chantypU{\tV_1}} \quad\struct\quad
     \env,\,
      \emap{w}{\proctyp},\, \emap{x}{\chantypO{\tV_1}},\,
      \emap{x}{\chantypUU{\tV_1}{1}}
  \end{equation*}
  Using \rtit{tStr} and \rtit{tOut} we can type \eqref{eq:7:ts} to obtain
  \begin{equation*}
    \tproc{\env,\,\emap{w}{\proctyp},\, \emap{x}{\chantypUU{\tV_1}{1}}\;}{\;\piIn{x}{y}{\;\piOut{\ctit{srv}_2}{x}\,\piIn{x}{z}{\;\piFree{x}{\;\piOut{\ctit{ret}}{(y,z)}{\;w}}}}}
  \end{equation*}\,
  After applying \rtit{tIn} to typecheck the input, we are left with the sequent
  \begin{equation*}
    \tproc{\env,\, \emap{w}{\proctyp},\, \emap{x}{\chantypU{\tV_1}},\, \emap{y}{\tV_1} \;}{\piOut{\ctit{srv}_2}{x}\,\piIn{x}{z}{\;\piFree{x}{\;\piOut{\ctit{ret}}{(y,z)}{\;w}}}}
  \end{equation*}
  In particular, we note that the input typing rule stipulates that the input continuation process needs to type \wrt the following type assumption for $\emap{x}{\chantyp{\tV_1}{{\unique{1} - 1}}}$ which is equal to $\emap{x}{\chantypU{\tV_1}}$.   Since $x$ is unique now, we can change the object type from $\tV_1$ to $\tV_2$ using \rtit{tRev}, which allows us to type the interactions with  $\ctit{srv}_2$ in analogous fashion. This leaves us with    
\begin{equation*}
    \tproc{\env,\, \emap{w}{\proctyp},\, \emap{x}{\chantypU{\tV_2}},\, \emap{y}{\tV_1},\, \emap{z}{\tV_2} \;}{\piFree{x}{\;\piOut{\ctit{ret}}{(y,z)}{\;w}}}
  \end{equation*}
which we can discharge using rules \rtit{tFree}, \rtit{tOut} and \rtit{tVar}.
\end{exa}

\newcommand{\tVR}{\ensuremath{\tV_\text{\rm rec}}\xspace}
\newcommand{\pBuf}{\ensuremath{\text{\rm Buff}}\xspace}
\newcommand{\pBufE}{\ensuremath{\text{\rm eBuff}}\xspace}
\newcommand{\pF}{\ensuremath{\text{\rm Frn}}\xspace}
\newcommand{\pFf}{\ensuremath{\text{\rm Frn'}}\xspace}
\newcommand{\pFff}[1]{\ensuremath{\text{\rm Frn''}(#1)}\xspace}
\newcommand{\pFfff}[2]{\ensuremath{\text{\rm Frn'''}(#1,#2)}\xspace}
\newcommand{\pB}{\ensuremath{\text{\rm Bck}}\xspace}
\newcommand{\pBb}{\ensuremath{\text{\rm Bck'}}\xspace}
\newcommand{\pBbb}[1]{\ensuremath{\text{\rm Bck''}(#1)}\xspace}
\newcommand{\pBbbb}[2]{\ensuremath{\text{\rm Bck'''}(#1,#2)}\xspace}
\newcommand{\pBE}{\ensuremath{\text{\rm eBk}}\xspace}
\newcommand{\pBEe}{\ensuremath{\text{\rm eBk'}}\xspace}
\newcommand{\pBEee}[1]{\ensuremath{\text{\rm eBk''}(#1)}\xspace}
\newcommand{\pBEeee}[3]{\ensuremath{\text{\rm eBk'''}(#1,#2,#3)}\xspace}
\newcommand{\pBEeeee}[2]{\ensuremath{\text{\rm eBk''''}(#1,#2)}\xspace}
\newcommand{\envExt}{\ensuremath{\env_\text{ext}}\xspace}

\section{A Case Study}
\label{sec:case-study}

Resource management is particularly relevant to programs manipulating (unbounded) 
regular structures. We consider the concurrent implementation of an unbounded buffer, \pBuf, receiving values to queue on channel \ctit{in} and dequeuing values by outputting on channel \ctit{out}.  
\begin{align*}
  \pBuf & \;\deftxt\;  \piIn{\ctit{in}}{y}{\;\piAll{z}{\;\bigl(\pF\piParal \piOutA{b}{z} \piParal \piOutA{c_1}{(y,z)}\bigr)}} \;\;\piParalS\;\; \piIn{c_1}{(y,z)}{\;{\piOut{\ctit{out}}{y}{\;\bigl(\pB \piParal \piOutA{d}{z}\bigr)}}}\\
  \pF & \;\deftxt\; \piRec{w}{\;\piIn{b}{x}{\;\piIn{\ctit{in}}{y}{\piAll{z}{\;\bigl(w\piParal \piOutA{b}{z} \piParal \piOutA{x}{(y,z)}\bigr)}}}}\\
  \pB & \;\deftxt\; \piRec{w}{\;\piIn{d}{x}{\;\piIn{x}{(y,z)}{\;\piOut{\ctit{out}}{y}{\;\bigl(w \piParal \piOutA{d}{z}\bigr)}}}}
\end{align*}
In order to  decouple input requests from output requests while still preserving the order of inputted values, the process handling inputs in \pBuf, $\piIn{\ctit{in}}{y}{\piAll{z}{\bigl(\pF\piParal \piOutA{b}{z} \piParal \piOutA{c_1}{(y,z)}\bigr)}}$, stores inputted values $v_1,\ldots,v_n$ as a queue of interconnected outputs
\begin{equation}
\piOutA{c_1}{(v_1,c_2)}  \piParalS  \ldots \piParalS \piOutA{c_n}{(v_n,c_{n+1})}\label{eq:4cs}
\end{equation}
on the internal\footnote{Subsequent allocated channels are referred to as $c_2,c_3,$ \etc.} channels $c_1,\ldots,c_{n+1}$.  The process handling the outputs, $\piIn{c_1}{(y,z)}{{\piOut{\ctit{out}}{y}{\bigl(\pB \piParal \piOutA{d}{z}\bigr)}}}$, then reads from the head of this queue, \ie the output on channel $c_1$, so as to obtain the first value inputted, $v_1$, and the next head of the queue, $c_2$. The input and output processes are defined in terms of the recursive processes, \pF  and \pB  \resp, which are parameterised by the channel to output (\resp input) on next through the channels $b$ and $d$.\footnote{This models parametrisable process definitions \pF(x) and \pB(x) within our language.}  

Since the buffer is \emph{unbounded}, the number of internal channels used for the queue of interconnected outputs, \eqref{eq:4cs}, is not fixed and these channels cannot therefore be created up front. Instead, they are created on demand by the input process  for every value inputted, using the \picr construct \piAll{z}{P}.  The newly allocated channel $z$ is then passed on the next iteration of \pF through channel $b$, \piOutA{b}{z},  and communicated as the next head of the queue when adding the subsequent queue item; this is received by the output process when it inputs the value at the head of the chain and  passed on the next iteration of \pB through channel $d$, \piOutA{d}{z}.

\subsection{Typeability and behaviour of the Buffer}
\label{sec:typab-behav-pbuf}

Our unbounded buffer implementation, \pBuf,  can be typed \wrt the type environment
\begin{equation} \label{eq:8cs}
\env_\text{int} \;\deftxt\; \emap{\ctit{in}}{\chantypW{\tV}},\, \emap{\ctit{out}}{\chantypW{\tV}},\, \emap{b}{\chantypW{\tVR}}, \,\emap{d}{\chantypW{\tVR}}, \,\emap{c_1}{\chantypU{\tV,\tVR}}
\end{equation}
where \tV\ is the type of the values stored in the buffer and \tVR\ is a recursive type defined as
$$\tVR  \;\deftxt\; \chanrecX{\chantypUU{\tV, X}{1}}.$$
This recursive type is used to type the internal channels $c_1,\ldots,c_{n+1}$ --- recall that in \eqref{eq:4cs} these channels carry channels of the same kind in order to link to one another as a chain of outputs.  In particular, using the typing rules of \secref{sec:language} we can prove the following typing judgements:
\begin{align}
  \label{eq:9cs}
  \tproc{\emap{\ctit{in}}{\chantypW{\tV}},\,  \emap{b}{\chantypW{\tVR}},  \,\emap{c_1}{\chantypO{\tV,\tVR}}\,}{}&\;\piIn{\ctit{in}}{y}{\;\piAll{z}{\;\bigl(\pF\piParal \piOutA{b}{z} \piParal \piOutA{c_1}{(y,z)}\bigr)}}\\
  \label{eq:10cs}
   \tproc{\emap{\ctit{out}}{\chantypW{\tV}},\, \emap{d}{\chantypW{\tVR}}, \,\emap{c_1}{\chantypUU{\tV,\tVR}{1}}\,}{}&\;\piIn{c_1}{(y,z)}{\;{\piOut{\ctit{out}}{y}{\;\bigl(\pB \piParal \piOutA{d}{z}\bigr)}}}
\end{align}
From the perspective of a user of the unbounded buffer, \pBuf implements the interface defined by the environment
$$\envExt\;\deftxt\; \emap{\ctit{in}}{\chantypW{\tV}}, \emap{\ctit{out}}{\chantypW{\tV}}$$
abstracting away from the implementation channels $b, d$ and $c_1$.

\subsection{A resource-conscious Implementation of the Buffer}
\label{sec:more-effic-impl}


When the buffer implementation of \pBuf\ retrieves values from the head of the internal queue, \eg (\ref{eq:4cs}), the channel holding the initial value, \ie $c_1$ in (\ref{eq:4cs}), is never reused again even though it is left allocated in memory.    
This fact will repeat itself for every value that is stored and retrieved from the buffer and amounts to the equivalent of a \emph{``memory leak''}.   A more resource-conscious implementation of the unbounded buffer is  \pBufE, defined in terms of the previous input process  used for \pBuf, and a modified output process,  $\piIn{c_1}{(y,z)}{\piFree{c_1}{\piOut{\ctit{out}}{y}{\bigl(\pBE \piParal \piOutA{d}{z}\bigr)}}}$, which uses the tweaked recursive process, \pBE.   
\begin{align*}
    \pBufE & \deftxt  \piIn{\ctit{in}}{y}{\piAll{z}{\bigl(\pF\piParal \piOutA{b}{z} \piParal \piOutA{c_1}{(y,z)}\bigr)}} \piParal \piIn{c_1}{(y,z)}{\piFree{c_1}{\piOut{\ctit{out}}{y}{\bigl(\pBE \piParal \piOutA{d}{z}\bigr)}}}\\
  \pBE & \deftxt \piRec{w}{\;\piIn{d}{x}{\;\piIn{x}{(y,z)}{\;\piFree{x}{\;\piOut{\ctit{out}}{y}{\bigl(w \piParal \piOutA{d}{z}\bigr)}}}}}
\end{align*}
The main difference between \pBuf and \pBufE is that the latter deallocates the channel at the head of the internal chain once it is consumed.   We can typecheck \pBufE as safe since no other process uses the internal channels making up the chain after deallocation.  More specifically, the typeability of  \pBufE  \wrt $\env_\text{int}$ of \eqref{eq:8cs} follows from  \eqref{eq:9cs} and the type judgement below:
\begin{equation*}
  \tproc{\emap{\ctit{out}}{\chantypW{\tV}},\, \emap{d}{\chantypW{\tVR}}, \,\emap{c_1}{\chantypUU{\tV,\tVR}{1}}\,}{\;\piIn{c_1}{(y,z)}{\;\piFree{c_1}{\;{\piOut{\ctit{out}}{y}{\;\bigl(\pB \piParal \piOutA{d}{z}\bigr)}}}}}
\end{equation*}
Note that by the typing rule \rtit{tIn} of  \figref{fig:typingrules}, we need to typecheck the continuation of the input process, \piFree{c_1}{\;{\piOut{\ctit{out}}{y}{\;\bigl(\pB \piParal \piOutA{d}{z}\bigr)}}} \wrt the type environment 
\begin{equation*}
  \emap{\ctit{out}}{\chantypW{\tV}},\, \emap{d}{\chantypW{\tVR}}, \,\emap{c_1}{\chantypU{\tV,\tVR}},\, \emap{y}{\tV},\, \emap{z}{\tVR}
\end{equation*}
where, in particular, $c_1$ is now assigned a \emph{unique} channel type.  According to the typing rule \rtit{tFree}, this suffices to safely type the respective deallocation of $c_1$.



\section{A Cost-Based Preorder}
\label{sec:cost-bisim}

We define our cost-based preorder as a \emph{bisimulation relation} that relates two
systems $\sys{\sV}{P}$ and $\sys{\sVV}{Q}$ whenever they have equivalent behaviour and when, in addition, 
$\sys{\sV}{P}$ is more efficient than $\sys{\sVV}{Q}$. We are interested in reasoning about \emph{safe} computations, aided by the type system described in Section~\ref{sec:language}.  For this reason, we  limit our analysis to instances of $\sys{\sV}{P}$ and $\sys{\sVV}{Q}$ that are \emph{well-typed}, \ie that there exist (consistent) environments $\envv,\envv'$ such that $\tproc{\envv}{\sys{\sV}{P}}$ and $\tproc{\envv'}{\sys{\sVV}{Q}}$.    In order to preserve safety, we also need to reason under the assumption of safe contexts. Again, we employ the type system described in Section~\ref{sec:language} and characterise the (safe) context through a type environment that typechecks it, $\env_\mathit{obs}$.  Thus our bisimulation relations take the form of a typed relation, indexed by type environments \cite{hennessy04behavioural}:
\begin{align}\label{eq:typed-relation}
  \env_\mathit{obs} & \vDash (\sys{\sV}{P}) \;\relR\; (\sys{\sVV}{Q})
\end{align}
Behavioural reasoning  for safe systems is achieved by ensuring that the overall type environment $(\env_\mathit{sys}, \env_\mathit{obs})$, consisting of the environment typing $\sys{\sV}{P}$ and $\sys{\sVV}{Q}$, say $\env_\mathit{sys}$, and the observer environment $\env_\mathit{obs}$,  is \emph{consistent} according to  Definition~\ref{def:consistency}. This means that there exists a global environment,
$\env_\mathit{global}$, which can be decomposed  into  $\env_\mathit{obs}$ and  $\env_\mathit{sys}$;  it also means that the observer process, which is universally quantified by our semantic interpretation \eqref{eq:typed-relation}, typechecks when composed in parallel with $P$, \resp $Q$ (see \rtit{tPar} of \figref{fig:typingrules}). 

There is one other complication worth highlighting regarding \eqref{eq:typed-relation}:  although  both systems $\sys{\sV}{P}$ and $\sys{\sVV}{Q}$ are related \wrt the
same \emph{observer}, $\env_\mathit{obs}$, they can each be typed under
\emph{different} typing environments.  For instance, consider the two clients
$\ptit{C}_0$ and $\ptit{C}_1$ we would like to relate from the introduction:
\begin{equation} \label{eq:1:bisim}
\begin{split}
  \ptit{C}_0 &\deftri \piRecX{\;\piAll{x_1}{\piAll{x_2}\;\piOut{\ctit{srv}_1}{x_1}\,\piIn{x_1}{y}{\piOut{\ctit{srv}_2}{x_2}\,\piIn{x_2}{z}{\piOut{\ctit{c}}{(y,z)}{w}}}}} \\
  \ptit{C}_1 &\deftri \piRecX{\;\piAll{x}{\;\piOut{\ctit{srv}_1}{x}\,\piIn{x}{y}{\piOut{\ctit{srv}_2}{x}\,\piIn{x}{z}{\piOut{\ctit{c}}{(y,z)}{w}}}}} 
\end{split}
\end{equation}

Even though, initially, they may be  typed by the same type environment, after a few steps, the derivatives of $\ptit{C}_0$ and $\ptit{C}_1$ must be typed under
different typing environments, because $\ptit{C}_0$ allocates two channels,
while $\ptit{C}_1$ only allocates a single channel. Our typed relations  allows for this by \emph{existentially quantifying} over the type environments typing the respective systems.   All this is achieved indirectly through the use of \emph{configurations}.  

\begin{defi}[Configuration]\label{def:configuration}
The triple \confESP\ is a configuration if and only if $\dom(\env) \subseteq \sV$ and there exist
some \envv such that  $(\env, \envv)$ is consistent and $\tproc{\envv}{\sysSP}$.
\end{defi}

\noindent Note that, in a configuration \confESP\ (where \env types some implicit observer): 
\begin{itemize}
\item $c\in(\dom(\env)\cup \names(P))$ implies $c\in\sV$ \ie \sV\ is a global resource environment accounting for both $P$ and $\env$.  
\item $c \in \sV$ and $c\not\in(\dom(\env)\cup \names(P))$ denotes a
  resource leak for channel $c$.
\item $c\not\in\dom(\env)$ implies that channel $c$ is not known to the observer; in some sense, this mimics name scoping in more standard \pic settings.
\end{itemize}

\begin{defi}[Typed Relation]\label{def:typed-relation}
A type-indexed relation $\relR$ relates systems under a observer
characterized by a context $\env$; we write
\begin{equation*}
\env \vDash \sys{M}{P} \; \relR \; \sys{N}{Q}
\end{equation*}
if $\relR$ relates $\confE{\sys{M}{P}}$ and $\confE{\sys{N}{Q}}$, and both
$\confE{\sys{M}{P}}$ and $\confE{\sys{N}{Q}}$ are configurations.
\end{defi}

\subsection{Labelled Transition System}
\label{sec:LTS}

In order to be able to reason coinductively over our typed relations, we define a labelled transition system (LTS) over configurations.  Apart from describing the behaviour of the system \sysSP in a configuration \confESP,  the LTS also models interactions between the system and an observer typed under \env.  Our LTS is also \emph{costed}, assigning a cost to each form of transition.

The costed LTS, whose actions take the form $\piRedDecCost{\;\mu\;}{k}$, is defined in \figref{fig:LTS}, in terms of a top-level rule, \rtit{lRen}, and a pre-LTS, denoted as $\piRedDecCostPre{\;\mu\;}{k}$. The rule \rtit{lRen} allows us to rename channels for transitions derived in the pre-LTS, as long as this renaming is invisible to the observer, and is  comparable to alpha-renaming of scoped bound names in the standard
\pic.  It relies on the renaming-modulo (observer) type environments given in Definition~\ref{def:renaming}.

\begin{defi}[Renaming  Modulo \env]\label{def:renaming}
Let  $\sigma_\env :\Names \mapsto \Names$ range over bijective name substitutions 
satisfying the constraint that
\begin{math}
  c\in\dom(\env) \text{ implies } c\sigma_\env = c\sigma_\env^{-1} = c
\end{math}.
 \end{defi}

The renaming introduced by \rtit{lRen}  allows us to relate the clients  $\ptit{C}_0$ and $\ptit{C}_1$ of \eqref{eq:1:bisim} \wrt an observer environment such as $\ctit{srv}_1 : \chantypW{\chantypO{\tV_1}},  \ctit{srv}_2 : \chantypW{\chantypO{\tV_2}}$ of \eqref{eq:2a} and some appropriate common set of resources \sV\ even when, after the initial channel allocations, the two clients communicate potentially different (newly allocated) channels on $\ctit{srv}_1$.  The rule is particularly useful when, later on, we need to also match the output of a new allocated channel on $\ctit{srv}_2$  from $\ptit{C}_0$ with the output on the previously allocated channel from $\ptit{C}_1$ on $\ctit{srv}_2$. The renaming-modulo observer environments function can be used for $\ptit{C}_1$ at that stage --- even though the client reuses a channel previously communicated to the observer --- because the respective observer information relating to that channel is lost,  \ie it is not in the domain of the observer environment; see discussion for \rtit{lOut} and \rtit{lIn} below for an explanation of how observers lose information.   This mechanism differs from standard scope-extrusion techniques for \pic which assume that,  once a name has been extruded, it remains forever known to the observer.  As a result, there are more opportunities for renaming in our calculus than there are in the standard \pic.

\begin{display}{LTS Process Moves}{fig:LTS}
 \small
\textbf{Costed Transitions and pre-Transitions}\\
\begin{mathpar}
\begin{prooftree}
  \confE{\bigl(\sysSP\bigr)\sigma_\env} \;\piRedDecCostPre{\;\mu\;}{k}\; \conf{\env'}{\sys{\sV'}{P'}}
  \justifiedBy{lRen}
  \confESP \;\piRedDecCost{\;\mu\;}{k}\; \conf{\env'}{\sys{\sV'}{P'}}
\end{prooftree}
 \\
\begin{prooftree}
\justifiedBy{lOut}
  \conf{\env,\emap{c}{\chantypA{\tVlst}}\;}{\;\sys{\sV
      \;}{\;\piOut{c}{\vec{d}}{P}}} 
\;\piRedDecCostPre{\actout{c}{\vec{d}}}{0}\;
  \conf{\env,\emap{c}{\chantyp{\tVlst}{\aV-1}}, \envmap{\vec{d}}{\vec{\tV}}\;}{\;\sys{\sV
      \;}{\;P}}
\end{prooftree}
\\
\begin{prooftree}
\justifiedBy{lIn}
  \conf{\env,\emap{c}{\chantypA{\tVlst}}, \emap{\vec{d}}{\vec{\tV}}\;}{\;\sys{\sV
      \;}{\;\piIn{c}{\vec{x}} {P}}} 
\;\piRedDecCostPre{\actin{c}{\vec{d}}}{0}\;
  \conf{\env,\emap{c}{\chantyp{\tVlst}{\aV-1}}\;}{\;\sys{\sV
      \;}{\;P \subC{\vec{d}}{\vec{x}}}}
\end{prooftree}
\\
\begin{prooftree}
  \conf{\env_1}{\sysS{P}} 
\;\piRedDecCostPre{\actout{c}{\vec{d}}}{0}\;
  \conf{\env'_1}{\sysS{P'}} 
\qquad 
  \conf{\env_2}{\sysS{Q}} 
\;\piRedDecCostPre{\actin{c}{\vec{d}}}{0}\;
  \conf{\env'_2}{\sysS{Q'}} 
\justifiedBy{lCom-L}
  \conf{\env}{\sysS{P \piParal Q}} 
\;\piTauCostPre{0}\;
  \conf{\env} {\sysS{P' \piParal Q'}} 
\end{prooftree}
\\
\begin{prooftree}
  \confESP 
\;\piRedDecCostPre{\;\mu\;}{k}\;
  \conf{\env'}{\sys{\sV'}{P'}} 
\justifiedBy{lPar-L}
  \confE{\sysS{P \piParal Q}} 
\;\piRedDecCostPre{\;\mu\;}{k}\;
  \conf{\env'}{\sys{\sV'}{P' \piParal Q}} 
\end{prooftree}
\\
\begin{prooftree}
  \env \struct \env'
\justifiedBy{lStr}
  \conf{\env}{\sysSP}
\;\piRedDecCostPre{\;\actenv\;}{0}\;
  \conf{\env'}{\sysSP}
\end{prooftree} 
\qquad\qquad
\begin{prooftree}
\strut
\justifiedBy{lRec}
  \confE{\sysS{\piRecX{P}}}
\;\piTauCostPre{0}\;
  \confE{\sysS{P\subC{\piRecX{P}}{w}}}
\end{prooftree}
\\
\begin{prooftree}
\strut
\justifiedBy{lThen}
  \confE{\sys{\sV,c}{\piIf{c=c}{P}{Q}}}
 \;\piTauCostPre{0}\;
  \confEP{\sV,c
  }
\end{prooftree}
\\
\begin{prooftree}
\strut
\justifiedBy{lElse}
  \confE{\sys{\sV,c,d}{\piIf{c=d}{P}{Q}}} 
\;\piTauCostPre{0}\;
  \confE{\sys{\sV,c,d
    }{Q}}
\end{prooftree}\\
  \begin{prooftree}
    \strut  
    \justifiedBy{lAll}
    \confE{\sys{\sV}{\piAll{x}{P}}}
\;\piTauCostPre{+1}\;
\confE{\sys{\sV,c}{P\subC{c}{x}}}
  \end{prooftree}
\qquad\qquad
\begin{prooftree}
\strut
\justifiedBy{lAllE}
  \conf{\env}{\sys{\sV}{P}}
\;\piRedDecCostPre{\actall}{+1}\;
  \conf{\env, \envmap{c}{\chantypU{\tVlst}}}{\sys{\sV,c}{P}}
\end{prooftree}
\\
  \begin{prooftree}
    \strut
    \justifiedBy{lFree}
    \confE{\sys{\sV,c}{\piFree{c}{P}}}
\;\piTauCostPre{-1}\;
\confE{\sys{\sV}{P}} 
  \end{prooftree}
\qquad\qquad
\begin{prooftree}
\strut
\justifiedBy{lFreeE}
  \conf{\env, \envmap{c}{\chantypU{\tV}}}{\sys{\sV,c}{P}}
\;\piRedDecCostPre{\actfree{c}}{-1}\;
  \conf{\env}{\sys{\sV}{P}}
\end{prooftree}
\\
\end{mathpar}
\textbf{Weak (Cost-Accumulating) Transitions}\\
\begin{mathpar}
 \begin{prooftree}
\confESP \piRedDecCostPad{\mu}{k} \conf{\envv}{\sys{\sVV}{Q}}
\justifiedBy{wTra}
\confESP \piRedWDecCostPad{\mu}{k} \conf{\envv}{\sys{\sVV}{Q}}
\end{prooftree}
\qquad\qquad
\begin{prooftree}
  \confESP 
\piRedDecCostPad{\tau}{l} 
  \conf{\env'}{\sys{\sV'}{P}} 
\piRedWDecCostPad{\mu}{k} 
  \conf{\env''}{\sys{\sVV}{Q}}
\justifiedBy{wLeft}
\confESP \piRedWDecCostPad{\mu}{(l+k)} \conf{\env''}{\sys{\sVV}{Q}}
\end{prooftree}
\\
\begin{prooftree}
  \confESP 
\piRedWDecCostPad{\mu}{l} 
  \conf{\env'}{\sys{\sV'}{P}} 
\piRedDecCostPad{\tau}{k} 
  \conf{\env''}{\sys{\sVV}{Q}}
\justifiedBy{wRight}
\confESP \piRedWDecCostPad{\mu}{(l+k)} \conf{\env''}{\sys{\sVV}{Q}}
\end{prooftree}
\end{mathpar}
\end{display}

To ensure that only safe interactions are specified, the (pre-)LTS must 
be able to reason compositionally  about resource usage between the  process, $P$, and the observer, \env.   We therefore imbue our type assumptions from \secref{sec:language} with a \emph{permission semantics}, in the style of \cite{TerauchiAiken08,FraRatSas11}. Under this interpretation, type assumptions constitute  \emph{permissions}  describing the respective usage of resources.  Permissions are woven into the behaviour of configurations giving them an \emph{operational} role: they may either restrict usage or privilege processes to use resources in special ways.  In a configuration, the observer and the process each \emph{own} a set of permissions and 
may \emph{transfer} them
to one another  during communication.  The consistency requirement of a configuration ensures that the guarantees given by permissions owned by the observer are not in conflict with those given by permissions owned by the configuration process, and viceversa.

To understand how the pre-LTS deals with permission transfer and compositional resource usage, consider the rule for output, (\rtit{lOut}).  Since we employ the type system of \secref{sec:language} to ensure safety, this rule models the typing rule for output (\rtit{tOut}) on the part of the process, and  the typing rule for input (\rtit{tIn}) on the part of the observer.  Thus, apart from describing the communication of values $\vec{d}$ from the configuration process to the observer on channel $c$, it also captures permission transfer between the two parties, mirroring the type assumption usage in \rtit{tOut} and \rtit{tIn}. More specifically, rule (\rtit{lOut}) employs  the operation $\envmap{c}{\chantyp{\vec{\tV}}{\aV-1}}$ of \figref{fig:typingrules} so as to concisely describe the three variants of the output rule:
\begin{equation}\label{eq:7:bisim}
  \begin{split}
    &\begin{prooftree}
      \justifiedBy{lOutU}
      \conf{\env,\emap{c}{\chantypUU{\tVlst}{i+1}}\;}{\;\sys{\sV
          \;}{\;\piOut{c}{\vec{d}}{P}}}
      \quad\piRedDecCostPre{\;\actout{c}{\vec{d}}\;}{0}\quad
      \conf{\env,\emap{c}{\chantypUU{\tVlst}{i}},
        \envmap{\vec{d}}{\vec{\tV}}\;}{\;\sys{\sV
          \;}{\;P}}
    \end{prooftree}\\[0.5em]
    &\begin{prooftree}
      \justifiedBy{lOutA}
      \conf{\env,\emap{c}{\chantypO{\tVlst}}\qquad}{\;\sys{\sV
          \;}{\;\piOut{c}{\vec{d}}{P}}}
      \quad\piRedDecCostPre{\;\actout{c}{\vec{d}}\;}{0}\quad
      \conf{\env,\phantom{\emap{c}{\chantypUU{\tVlst}{i}},}
        \envmap{\vec{d}}{\vec{\tV}}\;}{\;\sys{\sV
          \;}{\;P}}
    \end{prooftree}\\[0.5em]
    &\begin{prooftree}
      \justifiedBy{lOutW}
      \conf{\env,\emap{c}{\chantypW{\tVlst}}\quad\;\;}{\;\sys{\sV
          \;}{\;\piOut{c}{\vec{d}}{P}}}
      \quad\piRedDecCostPre{\;\actout{c}{\vec{d}}\;}{0}\quad
      \conf{\env,\emap{c}{\chantypW{\tVlst}},\;\;
        \envmap{\vec{d}}{\vec{\tV}}\;}{\;\sys{\sV
          \;}{\;P}}
    \end{prooftree}\\[0.5em]
  \end{split}
\end{equation}
The first output rule variant, \rtit{lOutU}, deals with the case where the observer owns a unique-after-$(i\!+\!1)$ permission for channel $c$.  \defref{def:configuration} implies that the process in the configuration is well-typed (\wrt some environment) and, since the process is in a position to output on channel $c$,  rule \rtit{tOut} must have been used to type it.  This typing rule, in turn, states that the type assumptions relating to the values communicated,  $\envmap{\vec{d}}{\vec{\tV}}$, must have been owned by the process and consumed by the output operation. Dually, since the observer is capable of inputting on $c$, rule \rtit{tIn} must have been used to type it,\footnote{More specifically, \rtit{tInU} of \eqref{eq:8:lang}.} 
which states that the continuation (after the input) assumes the use the assumptions  $\envmap{\vec{d}}{\vec{\tV}}$.  Rule  \rtit{lOutU} models these two usages  operationally as the \emph{explicit transfer} of the permissions $\envmap{\vec{d}}{\vec{\tV}}$ from the process to the observer.   

The rule also models the \emph{implicit transfer} of permissions between the observer and the output 
process.  More precisely,  \defref{def:configuration} requires that the process is typed \wrt an environment that \emph{does not conflict with} the observer environment, which  implies that the process environment must have (necessarily) used an affine permission, $\emap{c}{\chantypO{\tVlst}}$, for outputting on channel $c$.\footnote{This implies that \rtit{tOutA} of \eqref{eq:7:lang} was used when typing the process} In fact,  any other type of permission 
 would conflict with the unique-after-$(i\!+\!1)$ permission for channel $c$ owned by the observer.  Moreover, through the guarantee given by the permission used, \emap{c}{\chantypUU{\tVlst}{i+1}},  the observer knows that, after the communication, it is one step closer towards gaining exclusive permission for  channel $c$. Rule \rtit{lOutU} models all this as the (implicit)  transfer of the affine permission  $\emap{c}{\chantypO{\tVlst}}$ from the process to the observer, updating the observer's permission for $c$ to $\chantypUU{\tVlst}{i}$ --- note that two permissions $\emap{c}{\chantypUU{\tVlst}{i+1}},\emap{c}{\chantypO{\tVlst}}$ can be consolidated as $\emap{c}{\chantypUU{\tVlst}{i}}$ using the structural rules \rtit{tJoin} and \rtit{pUnq} of \figref{fig:typingrules}. 

The second output rule variant of \eqref{eq:7:bisim}, \rtit{lOutA}, is similar to the first when  modelling the explicit transfer of permissions $\envmap{\vec{d}}{\vec{\tV}}$ from the process to the observer. However, it describes a different implicit transfer of permissions, since the observer uses an affine permission to input from the configuration process on channel $c$.  The rule caters for two possible subcases. In the first case, the process could have used  a unique-after-$(i\!+\!1)$
permission when typed using \rtit{tOut}: this constitutes a  dual case to that of rule \rtit{lOutU}, and the rule models the implicit transfer of the affine permission $\emap{c}{\chantypO{\tVlst}}$ in the \emph{opposite} direction, \ie from the observer to the process. In the second case, the process could have used an affine or an unrestricted permission instead,  which does not result in any implicit permission transfer, but merely the consumption of  affine permissions. 
Since the environment on the process side is existentially quantified in a configuration, this difference is abstracted away and the two subcases are handled by the same rule variant. Note that, in the
extreme case where the observer affine permission is the only one relating to channel $c$, 
the observer loses all knowledge of channel $c$.

The explicit permission transfer for \rtit{lOutW} of \eqref{eq:7:bisim},  is identical to the other two rule variants. The use of an unrestricted permission for $c$ from the part of the observer,  $\emap{c}{\chantypW{\tVlst}}$, implies that the output process could have either used an affine or an unrestricted permission---see \eqref{eq:2:lang}.  In either case, there is no implicit permission transfer involved.  Moreover, the observer permission is not consumed since it is unrestricted.

The pre-LTS rule \rtit{lIn} can also be expanded into three rule variants, and models analogous permission transfer between the observer and the input process. Importantly, however, the \emph{explicit} permission transfer described is \emph{in the opposite direction} to that of \rtit{lOut}, namely from the observer to the input process.  As in the case of \rtit{lOutA} of \eqref{eq:7:bisim}, the permission transfer from the observer to the input process may result in the  observer losing all knowledge relating  to the channels communicated, $\vec{d}$.

In order to allow an internal communication step  through either \rtit{lCom-L}, or its dual \rtit{lCom-R} (elided), the left
process should be considered to be part of the ``observer'' of the right
process, and vice versa.  However, it is not necessary to be quite so precise;
we can follow \cite{Hennessy07} and consider an arbitrary observer instead.  More explicitly, the rule states that if we can find observer environments ($\Gamma_1$ and $\Gamma_2$) to induce the respective input and output actions from separate constituent processes making up the system,  we can then express these separate interactions as a single synchronous interaction; since this interaction is internal, it is independent of the environment representing the observer in the conclusion,  $\Gamma$.  See \cite{Hennessy07} for more justification.


In our LTS, both the process (\rtit{lAll, \rtit{lFree}}) and the observer (\rtit{lAllE},
\rtit{lFreeE}) can allocate and deallocate memory. Finally, since the observer
is modelled exclusively by the permissions it owns, we must allow the observer
to split these permissions when necessary (\rtit{lStr}). The only rules that may alter the observer environment are those corresponding to external actions \ie \rtit{\rtit{lIn}}, \rtit{\rtit{lOut}}, \rtit{lAllE}, \rtit{lFreeE} and \rtit{lStr}.  The remaining axioms in the pre-LTS model reduction rules from Figure~\ref{fig:reduction-semantics} and should be self-explanatory; note that, as in the reduction semantics, the only actions carrying a cost are those describing allocation and
deallocation, where the respective costs associated are inherited directly from the reduction semantics of \secref{sec:language}.

In \figref{fig:LTS} we also specify weak costed transitions for configurations, based on the transitions of our LTS (rule \rtit{wTra}).  As is standard, the relation denotes actions padded by $\tau$-transitions to the left and right.  However, it also \emph{accumulates} the costs of the respective transitions into one aggregate cost for the entire weak action (rules \rtit{wLeft} and \rtit{wRight}).

\medskip

Technically, the pre-LTS is defined over triples $\env, M, P$ rather than
configurations $\confESP$, but we can prove that the pre-LTS rules preserve
the requirements for such triples to be configurations; see Lemma~\ref{lem:subject-reduction}.  

\begin{lem}[Transition and Structure]\label{lem:transition-structure}
  \begin{math}
  \confESP \piRedDecCostPre{\;\;\mu\;\;}{k} \conf{\env'}{\sys{\sV'}{P'}} 
  \text{ and for } \envv \text{ consistent } \\\tproc{\envv}{\sysSP} \;\text{ implies the cases:}
\end{math}

\begin{description}
\item[If $\mu = \actout{c}{\vec{d}}$]  $\sV=\sV'$,  $k=0$, $P\piStruct \piOut{c}{\vec{d}}P_1 \piParal P_2$,  $P'\piStruct P_1 \piParal P_2$\; and\;  
\begin{math}
\env= (\env'', \emap{c}{\chantypA{\tVlst}}), \\ \env' = (\env'', \emap{c}{\chantyp{\tVlst}{\aV-1}}, \lst{\envmap{d}{\tV}}) 
\end{math}   \;and\; 
\begin{math} 
\envv  \struct (\envv', \envmap{c}{\chantyp{\tVlst}{\aVV}}, \lst{\envmap{d}{\tV}}),\;  
\end{math} 
  \begin{math}  
{\tproc{(\envv', \envmap{c}{\chantyp{\tVlst}{\aVV-1}})}{P'
  }}
\end{math}\\  for some $P_1, P_2, \env'', b, \tVlst$ and  $\envv'$.

\item[If $\mu = \actin{c}{\vec{d}}$]  $\sV=\sV'$, $k=0$, $P\piStruct \piIn{c}{\vec{x}}P_1 \piParal P_2$, $P'\piStruct P_1\subC{\vec{d}}{\vec{x}} \piParal P_2$ \; and\\
\begin{math}
\env  = (\env'', \emap{c}{\chantypA{\tVlst}}, \lst{\emap{d}{\tV}}), \env' = (\env'',\emap{c}{\chantyp{\tVlst}{\aV-1}})
\end{math} \; and\;
\begin{math}
\envv  \struct (\envv', \envmap{c}{\chantyp{\tVlst}{\aVV}}),
\end{math}\\   
\begin{math}   
 {\tproc{(\envv', \envmap{c}{\chantyp{\tVlst}{\aVV-1}}, \lst{\envmap{d}{\tV}})}{P'}}
\end{math}
\quad  for some $P_1, P_2, \env'', b, \tVlst$ and  $\envv'$.

\item[If $\mu = \tau$] Either of three cases hold :
  \begin{itemize}
    \item  $\sV=\sV'$, $k=0$  \;and\; $\env=\env'$ \;and\; $\tproc{\envv}{P'}$ or;
    \item $\sV=(\sV',c)$,\; $k=-1$ and $P \piStruct \piFree{c}{P_1}\piParal P_2$, $P'\piStruct P_1\piParal P_2$, $\env=\env'$ and ${\envv\struct \envv',\emap{c}{\chantypU{\tVlst}}}$ \; where\;$\tproc{\envv'}{P'}$ (for some $P_1, P_2, \tVlst$ and  $\envv'$)  or;
     \item$\sV'=(\sV,c)$,\; $k=+1$\; and\;$P \piStruct \piAll{x}{P_1}\piParal P_2$, $P'\piStruct P_1\subC{c}{x}\piParal P_2$ \; and\;$\env=\env'$ and  $\envv\struct \envv'$ and $\tproc{\envv',\emap{c}{\chantypU{\tVlst}}}{P'}$ (for some $P_1, P_2, \tVlst$ and  $\envv'$)
  \end{itemize}
\item[If $\mu = \actfree{c}$]  $\sV=(\sV',c)$, $k=-1$  \; and\;$\env = \env',\emap{c}{\chantypU{\tVlst}}$ \;and\;   $P=P'$ for some \tVlst.
\item[If $\mu = \actall$]  $\sV'=(\sV,c)$, $k=+1$  \; and\; $\env,\emap{c}{\chantypU{\tVlst}} = \env'$ \;and\;   $P=Q$  for some \tVlst.
\item[If $\mu=\actenv$] $\env \struct \env'$, $\sV =\sV'$, $k=0$ and $P=P'$
\end{description}
\end{lem}

\begin{proof} By rule induction on $\confESP \piRedDecCostPre{\;\mu\;}{k} \conf{\env'}{\sys{\sV'}{P'}}$
\end{proof}

\begin{lem}[Subject reduction]  
\label{lem:subject-reduction}
If $\confESP$ is a configuration and $\confESP \;\piRedDecCostPre{\;\mu\;}{k}\; \conf{\envv}{\sys{\sVV}{Q}}$ then $\conf{\envv}{\sys{\sVV}{Q}}$ is also a configuration.
\end{lem}

\begin{proof}
   We assume that $\dom(\env)\subseteq \sV$ and that there exists $\envv$ such that  $\env,\envv$ is consistent and that $\tproc{\envv}{\sysSP}$.  The rest of the proof  follows from Lemma~\ref{lem:transition-structure} (Transition and Structure), by case analysis of $\mu$. 
\end{proof}

As a consistency  check, we can also show that our LTS semantics is in accordance with the reduction semantics presented in \ref{sec:language}.  In particular, $\tau$-transitions correspond to reductions modulo renaming and process structural equivalence. 

\begin{lem}[Reduction and Silent Transitions] \label{lem:reduc-and-transitions}\quad
  \begin{enumerate}
  \item $\sys{\sV}{P}\piRed_k \sys{\sV'}{P'}$ implies $\confE{\sys{\sV}{P}} \piRedDecCost{\tau}{k} \confE{\sys{\sV'}{P''}}$ for arbitrary $\env$ where $P''\piStruct P'$.
  \item  $\confE{\sys{\sV}{P}} \piRedDecCost{\tau}{k} \conf{\envv}{\sys{\sV'}{P'}}$ implies  $(\sys{\sV}{P})\sigma_\env \piRed_k \sys{\sV'}{P'}$ for some $\sigma_\env$.
  \end{enumerate}
\end{lem}
\begin{proof}
  By rule induction on $\sys{\sV}{P}\piRed_k \sys{\sV'}{P'}$ and $\confE{\sys{\sV}{P}} \piRedDecCostPad{\tau}{k} \conf{\envv}{\sys{\sV'}{P'}}$.   
\end{proof}


\begin{exa}\label{ex:buff-trans}  Recall the buffer implementation \pBuf from \secref{sec:case-study} and the respective external environment \envExt defined in \secref{sec:typab-behav-pbuf}.
  The transition rules of \figref{fig:LTS} allow us to derive the
  following behaviour for the configuration
  \conf{\envExt}{\sys{\sV,c_1}{\pBuf}} (where $\ctit{in}, \ctit{out},
  b, d \in \sV$):
  \begin{align}
    \label{eq:1cs}
    \conf{\envExt}{\sys{\sV,c_1}{\pBuf}} &
    \;\piDRedDecCost{\actin{\ctit{in}}{v_1}}{0}\;
    \conf{\envExt}{\sys{\sV,c_1}{ \left(\!\begin{array}{l}
            \piAll{z}{\bigl(\pF\piParal \piOutA{b}{z} \piParal  \piOutA{c_1}{(v_1,z)}\bigr)} \\
            \piParalS\;
            \piIn{c_1}{(y,z)}{{\piOut{\ctit{out}}{y}{\bigl(\pB
                  \piParal \piOutA{d}{z}\bigr)}}}
          \end{array}\!\right)
      }}\\
    \label{eq:2cs}
    &\;\piTauTauCost{+1}\; \conf{\envExt}{\sys{\sV,c_1,c_2}{
        \left(\!\begin{array}{l}
            \bigl(\pF  \piParal \piOutA{b}{c_2} \piParal  \piOutA{c_1}{(v_1,c_2)}\bigr) \\
            \piParalS\;
            \piIn{c_1}{(y,z)}{{\piOut{\ctit{out}}{y}{\bigl(\pB
                  \piParal \piOutA{d}{z}\bigr)}}}
          \end{array}\!\right)
      }}\\
    \nonumber
    &\;\;=\;\; \conf{\envExt}{\sys{\sV,c_1,c_2}{
        \left(\!\begin{array}{l}
            \piRec{w}{\;\piIn{b}{x}{\;\piIn{\ctit{in}}{y}{\piAll{z}{\;\bigl(w\piParal \piOutA{b}{z} \piParal \piOutA{x}{(y,z)}\bigr)}}}}  \\
            \piParalS \piOutA{b}{c_2} \piParalS  \piOutA{c_1}{(v_1,c_2)} \\
            \piParalS\; \piIn{c_1}{(y,z)}{{\piOut{\ctit{out}}{y}{\bigl(\pB \piParal \piOutA{d}{z}\bigr)}}}
          \end{array}\!\right)
      }}\\
    \label{eq:3cs}
    &\;\piTauTauCost{0}\; \conf{\envExt}{\sys{\sV,c_1,c_2}{
        \left(\!\begin{array}{l}
            \piIn{b}{x}{\;\piIn{\ctit{in}}{y}{\piAll{z}{\;\bigl(\pF\piParal \piOutA{b}{z} \piParal \piOutA{x}{(y,z)}\bigr)}}}  \\
            \piParalS \piOutA{b}{c_2} \piParalS  \piOutA{c_1}{(v_1,c_2)} \\
            \piParalS\;
            \piIn{c_1}{(y,z)}{{\piOut{\ctit{out}}{y}{\bigl(\pB
                  \piParal \piOutA{d}{z}\bigr)}}}
          \end{array}\!\right)
      }}\\
    \label{eq:5cs}
    &\;\piTauTauCost{0}\; \conf{\envExt}{\sys{\sV,c_1,c_2}{
        \left(\!\begin{array}{l}
            \piIn{\ctit{in}}{y}{\piAll{z}{\bigl(\pF\piParal \piOutA{b}{z} \piParal  \piOutA{c_2}{(y,z)}\bigr)}} \\
            \piParalS\; \piOutA{c_1}{(v_1,c_2)} \\
            \piParalS\;
            \piIn{c_1}{(y,z)}{{\piOut{\ctit{out}}{y}{\bigl(\pB
                  \piParal \piOutA{d}{z}\bigr)}}}
          \end{array}\!\right)
      }}\\
    \label{eq:6cs}
    &\;\piRedWDecCost{\actin{\ctit{in}}{v_2}}{+1}\;
    \conf{\envExt}{\sys{\sV,c_1,c_2,c_3}{ \left(\!\begin{array}{l}
            \piIn{\ctit{in}}{y}{\piAll{z}{\bigl(\pF\piParal \piOutA{b}{z} \piParal  \piOutA{c_3}{(y,z)}\bigr)}} \\
            \piParalS\; \piOutA{c_1}{(v_1,c_2)} \piParalS
            \piOutA{c_2}{(v_2,c_3)} 
            \\
            \piParalS\;
            \piIn{c_1}{(y,z)}{{\piOut{\ctit{out}}{y}{\bigl(\pB
                  \piParal \piOutA{d}{z}\bigr)}}}
          \end{array}\!\right)
      }}\\
    \label{eq:7cs}
    & \;\piRedWDecCost{\actout{\ctit{out}}{v_1}}{0}\;
    \conf{\envExt}{\sys{\sV,c_1,c_2,c_3}{ \left(\!\begin{array}{l}
            \piIn{\ctit{in}}{y}{\piAll{z}{\bigl(\pF\piParal \piOutA{b}{z} \piParal  \piOutA{c_3}{(y,z)}\bigr)}} \\
            \piParalS\;
            \piOutA{c_2}{(v_2,c_3)} 
            \\
            \piParalS\;
            \piIn{c_2}{(y,z)}{{\piOut{\ctit{out}}{y}{\bigl(\pB
                  \piParal \piOutA{d}{z}\bigr)}}}
          \end{array}\!\right)
      }}
  \end{align}
  Transition \eqref{eq:1cs} describes an input from the user whereas
  \eqref{eq:2cs} allocates a new internal channel, $c_2$, followed by
  a recursive process unfolding,~\eqref{eq:3cs}, and the instantiation
  of the unfolded process with the newly allocated channel $c_2$,
  \eqref{eq:5cs}, through a communication on channel $b$. The weak
  transition \eqref{eq:6cs} is an aggregation of 4 analogous
  transitions to the ones just presented, this time relating to a
  second input of value $v_2$. This yields an internal output chain of
  length 2, \ie $ \piOutA{c_1}{(v_1,c_2)} \piParalS
  \piOutA{c_2}{(v_2,c_3)} $. Finally, \eqref{eq:7cs} is an aggregation
  of 4 transitions relating to the consumption of the first item in
  the chain, $\piOutA{c_1}{(v_1,c_2)}$, the subsequent output of $v_1$
  on channel \ctit{out}, and the unfolding and instantiation of the
  recursive process \pB with $c_2$ --- see definition for \pB.
\end{exa}

\subsection{Costed Bisimulation}
\label{sec:Bisimulation}
We define a cost-based preorder over systems as a \emph{typed relation}, \cf Definition~\ref{def:typed-relation}, ordering systems that exhibit the same external behaviour at a less-than-or-equal-to cost. We require the preorder to consider client $\ptit{C}_1$ as more efficient than $\ptit{C}_0$ \wrt an appropriate resource environment \sV\ and observers characterised by the type environment stated in \eqref{eq:2a}  but also that, \wrt the same resource and observer environments, client   $\ptit{C}_3$ of \eqref{eq:clients-complicated} is more efficient than $\ptit{C}_1$.  This latter ordering is harder to establish since client  $\ptit{C}_1$ is at times \emph{temporarily} more efficient than $\ptit{C}_3$. 

In order to handle this aspect we define our preorder as an \emph{amortized} bisimulation \cite{Kiehn05}.  Amortized bisimulation uses a \emph{credit} $n$ to compare  a system $\sys{M}{P}$ with a less efficient 
 system $\sys{N}{Q}$ while allowing $\sys{M}{P}$ to do a more expensive action than
$\sys{N}{Q}$, as long as the credit can make up for the difference.  Conversely,
whenever $\sys{M}{P}$ does a cheaper action than $\sys{N}{Q}$, then the difference
gets \emph{added} to the credit.\footnote{Stated otherwise, $\sys{M}{P}$ can do a more
expensive action than $\sys{N}{Q}$ now, as long as it makes up for it later.} Crucially, however, the amortisation credit is \emph{never allowed} to become \emph{negative} \ie $n \in \Nats$. In
general, we refine Definition~\ref{def:typed-relation} to amortized typed relations with the following structure:

\begin{defi}[Amortised Typed Relation]\label{def:typed-relation-amm}
An amortized type-indexed relation $\relR$ relates systems under an observer
characterized by a context $\env$, with credit $n$ ($n \in \Nats$); we write
\begin{equation*}
\env \vDash \sys{M}{P} \; \relR^n \; \sys{N}{Q}
\end{equation*}
if $\relR^n$ relates $\confE{\sys{M}{P}}$ and $\confE{\sys{N}{Q}}$, and both
$\confE{\sys{M}{P}}$ and $\confE{\sys{N}{Q}}$ are configurations.
\end{defi}

\begin{defi}[Amortised Typed Bisimulation]
\label{def:amortized-typed-bisim} 
An amortized type-indexed relation over processes \relR\ is a bisimulation at
$\env$ with credit $n$ if, whenever  ${\env \vDash (\sys{M}{P}) \,\relR^n\, (\sys{N}{Q})}$,
  \begin{itemize}
  \item If $\confE{\sysS{P}} \piRedDecCostPad{\mu}{k} \conf{\env'}{\sys{\sV'}{P'}}$ then there exist $\sVV'$ and $Q'$ such that \\
    $\confE{\sys{\sVV}{Q}} \piRedWDecCostPad{\hat{\mu}}{l} \conf{\env'}{\sys{N'}{Q'}}$ where $\env' \vDash (\sys{\sV'}{P'}) \,\relR^{n+l-k}\, (\sys{\sVV'}{Q'})$ 
  \item If $\confE{\sys{\sVV}{Q}} \piRedDecCostPad{\mu}{l} \conf{\env'}{\sys{\sVV'}{Q'}}$ then there exist $\sV'$ and $P'$ such that \\
    $\confE{\sys{\sV}{P}} \piRedWDecCostPad{\hat{\mu}}{k} \conf{\env'}{\sys{\sV'}{P'}}$ where $\env' \vDash (\sys{\sV'}{P'}) \,\relR^{n+l-k}\, (\sys{\sVV'}{Q'})$
\end{itemize}
where $\hat{\mu}$ is the empty string if $\mu = \tau$ and $\mu$ otherwise.

Bisimilarity at $\env$ with credit $n$, denoted 
  $\env \vDash \sys{M}{P} \,\piCostBisAmm{n}  \sys{N}{Q}$, 
is the largest amortized typed bisimulation at $\env$ with credit $n$. We sometimes existentially quantify over the credit and write 
  ${\env \vDash \sys{M}{P}  \,\piCostBis\,  \sys{N}{Q}}$.
We write $\env \vDash \sys{M}{P} \piCostBisEq  \sys{N}{Q}$ to denote the kernel of the preorder (\ie whenever we have both
$\env \vDash \sys{M}{P}  \,\piCostBis\,  \sys{N}{Q}$ and $\env \vDash  \sys{N}{Q} \,\piCostBis\,  \sys{M}{P}$), and write $\env \vDash \sys{M}{P}  \piCostBiss \sys{N}{Q}$ whenever $\env \vDash \sys{M}{P}  \,\piCostBis\,  \sys{N}{Q}$ but $\env \vDash  \sys{N}{Q} \,\not\!\piCostBis\,  \sys{M}{P}$.
\end{defi}


\begin{exa}[Assessing Client Efficiency]
  \label{eg:Clients-Bisim}
For the (observer) type  environment
\begin{equation}\label{eq:env-sys-example:bis}
\env_1 \deftxt \envmap{\ctit{srv}_1}{\chantypW{\chantypO{\tV_1}}}, \; \envmap{\ctit{srv}_2}{\chantypW{\chantypO{\tV_2}}}, \;\envmap{c}{\chantypW{\tV_1,\tV_2}} 
\end{equation}
and clients $\ptit{C}_0$ and $\ptit{C}_1$ defined earlier in \eqref{eq:clients}, 
we can show that $\env_1 \vDash (\sysS{\ptit{C}_1}) \; \piCostBisAmm \; (\sysS{\ptit{C}_0})$ by constructing the witness bisimulation (family of) relation(s) \relR\ for $\env_1 \vDash (\sysS{\ptit{C}_1}) \; \piCostBisAmm{0} \; (\sysS{\ptit{C}_0})$ stated below:\footnote{In families of relations ranging over systems indexed by type environments and amortisation credits, such as \relR, we represent $\env \vDash (\sysS{P}) \; \piCostBisAmm{n} \; (\sys{\envv}{Q})$ as the quadruple $\langle \env, n, (\sysS{P}), (\sys{\envv}{Q})\rangle$.}
\begin{equation*}
    \relR \!\deftxt\!\left\{
      \begin{array}{l|@{\;}l}
        \langle\env,\;  n,\;  \sys{\sV'}{\ptit{C}_1}, \;\sys{\sVV'}{\ptit{C}_0}\rangle & n \geq 0  \\[0.1em]
        \left\langle\!\!
          \begin{array}{l}
            \env, n,\sys{\sV'}{\piAll{x}{\;\piOut{\ctit{srv}_1}{x}\,\piIn{x}{y}{\piOut{\ctit{srv}_2}{x}\,\piIn{x}{z}{\piOut{\ctit{c}}{(y,z)}{\ptit{C}_1}}}}}\\
             \qquad\quad , \sys{\sVV'}{\piAll{x_1}{\piAll{x_2}\;\piOut{\ctit{srv}_1}{x_1}\,\piIn{x_1}{y}{\piOut{\ctit{srv}_2}{x_2}\,\piIn{x_2}{z}{\piOut{\ctit{c}}{(y,z)}{\ptit{C}_0}}}}}
          \end{array}
           \!\!\right\rangle  & d \!\not\in\! \dom(\env) \\[0.8em]
         \left\langle\!\!
          \begin{array}{l}
            \env,\; n, \;\sys{(\sV',d)}{\piOut{\ctit{srv}_1}{d}\,\piIn{d}{y}{\piOut{\ctit{srv}_2}{d}\,\piIn{d}{z}{\piOut{\ctit{c}}{(y,z)}{\ptit{C}_1}}}}\\
             \qquad\quad , \sys{(\sVV',d')}{\piAll{x_2}\;\piOut{\ctit{srv}_1}{d'}\,\piIn{d'}{y}{\piOut{\ctit{srv}_2}{x_2}\,\piIn{x_2}{z}{\piOut{\ctit{c}}{(y,z)}{\ptit{C}_0}}}}
          \end{array}
           \!\!\right\rangle  &d' \!\not\in\! \dom(\env) \\[0.8em]
           \left\langle\!\!
           \begin{array}{l}
            \env,\; n+1, \;\sys{(\sV',d)}{\piOut{\ctit{srv}_1}{d}\,\piIn{d}{y}{\piOut{\ctit{srv}_2}{d}\,\piIn{d}{z}{\piOut{\ctit{c}}{(y,z)}{\ptit{C}_1}}}}\\
             \qquad\quad , \sys{(\sVV',d',d'')}{\piOut{\ctit{srv}_1}{d'}\,\piIn{d'}{y}{\piOut{\ctit{srv}_2}{d''}\,\piIn{d''}{z}{\piOut{\ctit{c}}{(y,z)}{\ptit{C}_0}}}}
          \end{array}
           \right\rangle &d'' \!\not\in\! \dom(\env) \\[0.8em]
           \left\langle\!\!
           \begin{array}{l}
            (\env,\emap{d}{\chantypO{\tV_1}}),\; n+1, \;\sys{(\sV',d)}{\piIn{d}{y}{\piOut{\ctit{srv}_2}{d}\,\piIn{d}{z}{\piOut{\ctit{c}}{(y,z)}{\ptit{C}_1}}}}\\
             \qquad\qquad\qquad , \sys{(\sVV',d,d'')}{\piIn{d}{y}{\piOut{\ctit{srv}_2}{d''}\,\piIn{d''}{z}{\piOut{\ctit{c}}{(y,z)}{\ptit{C}_0}}}}
          \end{array}
           \right\rangle & \sV' \subseteq \sVV' \\[0.8em]
           \left\langle\!\!
           \begin{array}{l}
            \env,\; n+1, \;\sys{(\sV',d)}{\piOut{\ctit{srv}_2}{d}\,\piIn{d}{z}{\piOut{\ctit{c}}{(v,z)}{\ptit{C}_1}}}\\
             \qquad\quad , \sys{(\sVV',d,d'')}{\piOut{\ctit{srv}_2}{d''}\,\piIn{d''}{z}{\piOut{\ctit{c}}{(v,z)}{\ptit{C}_0}}}
          \end{array}
           \right\rangle & \dom(\env) \subseteq \sV'\\[0.4em]
           \langle
            (\env,\emap{d}{\chantypO{\tV_2}}),\; n+1, \;\sys{(\sV',d)}{\piIn{d}{z}{\piOut{\ctit{c}}{(v,z)}{\ptit{C}_1}}},\;
            \sys{(\sVV',d',d)}{\piIn{d}{z}{\piOut{\ctit{c}}{(v,z)}{\ptit{C}_0}}}
          \rangle & \\
          \langle
            \env,\; n+1, \;\sys{(\sV',d)}{\piOut{\ctit{c}}{(v,v')}{\ptit{C}_1}},\;
            \sys{(\sVV',d',d)}{\piOut{\ctit{c}}{(v,v')}{\ptit{C}_0}}
          \rangle\\
      \end{array}
    \right\}
  \end{equation*}  
It is not hard to see that \relR\ contains the quadruple $\langle \env_1, 0, \sysS{\ptit{C}_1}, \sysS{\ptit{C}_0}\rangle$.  One can also show that it is closed \wrt the transfer property of Definition~\ref{def:amortized-typed-bisim}.  The key moves are:
\begin{itemize}
\item a single channel allocation by $\ptit{C}_1$ is matched by two channel allocations by  $\ptit{C}_0$ --- from the second up to the fourth quadruple in the definition of \relR.  Since channel allocations carry a positive cost, the amortisation credit increases from $n$ to $n+2-1$, \ie  $n+1$, but this still yields a quadruple that is in the relation.  One thing to note is that the first channel allocated by both systems is allowed to be different, \eg $d$ and $d'$, as long as it is not allocated already.  
\item Even though the internal channels allocated may be different, rule \rtit{lRen} allows us to rename the \resp names of the allocated channels (not known to the observer) so as match the channels communicated on $\ctit{srv}_1$  by the other system (fourth and fifth quadruples). Since these channels are not known to the observer, \ie they are not in $\dom(\env)$, they all amount to \emph{fresh} names, akin to scope extrusion \cite{Milner99,Hennessy07}.
\item Communicating on the previously communicated channel on $\ctit{srv}_1$ consumes all of the observer's permissions for that channel (fifth quadruple), which allows rule \rtit{lRen} to be applied again so as to match  the channels communicated on $\ctit{srv}_2$ (sixth quadruple). 
\end{itemize}

We cannot however prove  that
  $\env_1 \vDash (\sysS{\ptit{C}_0}) \; \piCostBisAmm{n} \; (\sysS{\ptit{C}_1})$
for any $n$ because  we  would need an \emph{infinite} amortisation credit to account for additional cost incurred by  $\ptit{C}_0$ when it performs the channel extra allocation at every iteration; recall that this credit cannot become negative, and thus no finite credit is large enough to cater for all the additional cost incurred by  $\ptit{C}_0$ over sufficiently large transition sequences.

Similarly, from \eqref{eq:clients}, we can show  that $\env_1 \vDash (\sysS{\ptit{C}_2}) \; \piCostBiss \; (\sysS{\ptit{C}_1})$ but also, from \eqref{eq:clients-complicated}, that  ${\env_1 \vDash (\sysS{\ptit{C}_3}) \; \piCostBiss \; (\sysS{\ptit{C}_1})}$.  In particular, we can show $\env_1 \vDash (\sysS{\ptit{C}_3}) \; \piCostBis \; (\sysS{\ptit{C}_1})$ even though $\sysS{\ptit{C}_1}$ is temporarily more efficient than $\sysS{\ptit{C}_3}$, \ie during the course of the first iteration.  Our framework handles this through the use of the amortisation credit whereby, in this case, it suffices to use a credit of value $1$ and show  $\env_1 \vDash (\sysS{\ptit{C}_3}) \; \piCostBisAmm{1} \; (\sysS{\ptit{C}_1})$; we leave the details to the interested reader.  Using an amortisation credit of $1$ we can also show  $\env_1 \vDash (\sysS{\ptit{C}_3}) \; \piCostBisAmm{1} \; (\sysS{\ptit{C}_2})$ through the bisimulation family-of-relations $\relR'$ below --- it is easy to check that  it observes the transfer property of Definition~\ref{def:amortized-typed-bisim}; by constructing a similar relation, one can also show that $\env_1 \vDash (\sysS{\ptit{C}_2}) \; \piCostBisAmm{0} \; (\sysS{\ptit{C}_3})$ which implies that $\env_1 \vDash (\sysS{\ptit{C}_2}) \; \piCostBisEq \; (\sysS{\ptit{C}_3})$.  We just note that in $\relR'$, the amortisation credit $n$ can be capped $0 \leq n \leq 1$ and  revisit this point again  in Section~\ref{sec:bisim-results}. 

\begin{equation*}
    \relR' \deftxt\left\{
      \begin{array}{l|@{\;}l}
        \langle\env,\;  1,\;  \sysS{\ptit{C}_3}, \;\sys{\sV}{\ptit{C}_2}\rangle \\[0.1em]
        \left\langle\!\!
          \begin{array}{l}
            \env,\; 1, \;
              \sysS{
               \left(\begin{array}{l}
\piAll{x_1}{\piAll{x_2}\;\piOut{\ctit{srv}_1}{x_1}\,\piIn{x_1}{y}{}}\\
\quad\piOut{\ctit{srv}_2}{x_2}\,\piIn{x_2}{z}{\piFree{x_1}{\piFree{x_2}{\piOut{\ctit{c}}{(y,z)}{\ptit{C}_3}}}}
\end{array}\right)
},\\
             \qquad\quad \sys{\sV}{\piAll{x}{\;\piOut{\ctit{srv}_1}{x}\,\piIn{x}{y}{\piOut{\ctit{srv}_2}{x}\,\piIn{x}{z}{\piFree{x}{\piOut{\ctit{c}}{(y,z)}{\ptit{C}_2}}}}}}
          \end{array}
           \right\rangle  &  \\[0.8em]
        \left\langle\!\!
          \begin{array}{l}
            \env,\; 1, \;
              \sys{(\sV,d)}{
               \left(\begin{array}{l}
\piAll{x_2}\;\piOut{\ctit{srv}_1}{d}\,\piIn{d}{y}{}\\
\quad\piOut{\ctit{srv}_2}{x_2}\,\piIn{x_2}{z}{\piFree{d}{\piFree{x_2}{\piOut{\ctit{c}}{(y,z)}{\ptit{C}_3}}}}
\end{array}\right)
},\\
             \qquad\quad \sys{(\sV,d')}{{\;\piOut{\ctit{srv}_1}{d'}\,\piIn{d'}{y}{\piOut{\ctit{srv}_2}{d'}\,\piIn{d'}{z}{\piFree{d'}{\piOut{\ctit{c}}{(y,z)}{\ptit{C}_2}}}}}}
          \end{array}
           \right\rangle  &d \!\not\in\! \dom(\env) \\[0.8em]
           \left\langle\!\!
          \begin{array}{l}
            \env,\; 0, \;
              \sys{(\sV,d,d'')}{
               \left(\begin{array}{l}
\piOut{\ctit{srv}_1}{d}\,\piIn{d}{y}{\piOut{\ctit{srv}_2}{d''}\,\piIn{d''}{z}{}}\\
\quad\piFree{d}{\piFree{d''}{\piOut{\ctit{c}}{(y,z)}{\ptit{C}_3}}}
\end{array}\right)
},\\
             \qquad\quad \sys{(\sV,d')}{{\piOut{\ctit{srv}_1}{d'}\,\piIn{d'}{y}{\piOut{\ctit{srv}_2}{d'}\,\piIn{d'}{z}{\piFree{d'}{\piOut{\ctit{c}}{(y,z)}{\ptit{C}_2}}}}}}
          \end{array}
           \right\rangle
         &d' \!\not\in\! \dom(\env)\\[0.8em]
         \left\langle\!\!
          \begin{array}{l}
            (\env,\emap{d}{\chantypO{\tV_1}}),\; 0, \;
              \sys{(\sV,d,d'')}{
               \left(\begin{array}{l}
\piIn{d}{y}{\piOut{\ctit{srv}_2}{d''}\,\piIn{d''}{z}{}}\\
\quad\piFree{d}{\piFree{d''}{\piOut{\ctit{c}}{(y,z)}{\ptit{C}_3}}}
\end{array}\right)
},\\
             \qquad\quad \sys{(\sV,d)}{{\piIn{d}{y}{\piOut{\ctit{srv}_2}{d}\,\piIn{d}{z}{\piFree{d}{\piOut{\ctit{c}}{(y,z)}{\ptit{C}_2}}}}}}
          \end{array}
           \right\rangle
            &d'' \!\not\in\! \dom(\env) \\[0.8em]
         \left\langle\!\!
          \begin{array}{l}
            \env,\; 0, \;
              \sys{(\sV,d,d'')}{
\piOut{\ctit{srv}_2}{d''}\,\piIn{d''}{z}{}\piFree{d}{\piFree{d''}{\piOut{\ctit{c}}{(v,z)}{\ptit{C}_3}}}
},\\
             \qquad\quad \sys{(\sV,d)}{\piOut{\ctit{srv}_2}{d}\,\piIn{d}{z}{\piFree{d}{\piOut{\ctit{c}}{(v,z)}{\ptit{C}_2}}}}
          \end{array}
           \right\rangle & \\[0.8em]
                     \left\langle\!\!
          \begin{array}{l}
             (\env,\emap{d'}{\chantypO{\tV_2}}),\; 0, \;
              \sys{(\sV,d,d')}{
\piIn{d'}{z}{}\piFree{d}{\piFree{d'}{\piOut{\ctit{c}}{(v,z)}{\ptit{C}_3}}}
},\\
             \qquad\quad \sys{(\sV,d')}{\piIn{d'}{z}{\piFree{d'}{\piOut{\ctit{c}}{(v,z)}{\ptit{C}_2}}}}
          \end{array}
           \right\rangle
 & \dom(\env) \subseteq \sV\\[0.8em]
             \left\langle\!\!
          \begin{array}{l}
             \env,\; 0, \;
              \sys{(\sV,d,d')}{
\piFree{d}{\piFree{d'}{\piOut{\ctit{c}}{(v,z)}{\ptit{C}_3}}}
},\\
             \qquad\qquad
             \sys{(\sV,d')}{{\piFree{d'}{\piOut{\ctit{c}}{(v,v')}{\ptit{C}_2}}}}
          \end{array}
           \right\rangle  & \\
             \left\langle\!\!
          \begin{array}{l}
             \env,\; 0, \;
              \sys{(\sV,d')}{
\piFree{d'}{\piOut{\ctit{c}}{(v,z)}{\ptit{C}_3}}
},\;
             \sys{\sV}{{\piOut{\ctit{c}}{(v,v')}{\ptit{C}_2}}}
          \end{array}
           \right\rangle  & \\
           \left\langle\!\!
          \begin{array}{l}
             \env,\; 1, \;
              \sys{\sV}{
\piOut{\ctit{c}}{(v,z)}{\ptit{C}_3}
},\;
             \sys{\sV}{{\piOut{\ctit{c}}{(v,v')}{\ptit{C}_2}}}
          \end{array}
           \right\rangle  
      \end{array}
    \right\}
  \end{equation*}
\hfill$\Box$
\end{exa}

\subsection{Alternatives}
\label{sec:alternatives}

The cost model we adhere to in Section~\ref{sec:cost-bisim} is not the only plausible one, but is intended to  follow that described by costed reductions of Section~\ref{sec:language}. There may however be other valid alternatives, some of which can be easily accommodated through minor tweaking to our existing framework.  

For instance, an alternative cost model may focus on assessing the runtime execution of programs, whereby operations that access memory such as \piAll{x}{P} and \piFree{c}{P} have a runtime cost that far exceeds that of other operations.  We can model this by considering an LTS that assigns a cost of $1$ to both of these operations, which can be attained as a derived LTS from our existing LTS of Section~\ref{sec:LTS} through the rule 
\begin{equation*}
  \begin{prooftree}
    \confESP \;\piRedDecCost{\;\mu\;}{k}\; \conf{\env'}{\sys{\sV'}{P'}}
  \justifiedBy{lDer1}
  \confESP \;\piRedDecCostPost{\;\mu\;}{|k|}\; \conf{\env'}{\sys{\sV'}{P'}}
\end{prooftree}
\end{equation*}
where $|k|$ returns the absolute value of an integer. \defref{def:amortized-typed-bisim} extends in straightforward fashion to work with the derived costed LTS $\piRedDecCostPost{\;\mu\;}{k}$. This new preorder would allow us to conclude  ${\env_1 \vDash (\sysS{\ptit{C}_1}) \; \piCostBis \; (\sysS{\ptit{C}_2})}$ because, according to the new cost model, for every server-interaction iteration, client  $\ptit{C}_1$ uses less expensive memory operations than $\ptit{C}_2$.

Another cost model may require us to refine our existing preorder. For instance, consider another client $\ptit{C}_4$, defined below, that creates a single channel and keeps on reusing it for all iterations:
\begin{equation*}
  \ptit{C}_4 \deftri \piAll{x}{\;\piRecX{\,\piOut{\ctit{srv}_1}{x}\,\piIn{x}{y}{\;\piOut{\ctit{srv}_2}{x}\,\piIn{x}{z}{\piOut{\ctit{ret}}{(y,z)}{\;w}}}}}
\end{equation*} 
At present, we are able to equate this client with $\ptit{C}_2$ and $\ptit{C}_3$ from \eqref{eq:clients} and \eqref{eq:clients-complicated} \resp, on the  basis that neither client carries any memory leaks. 
\begin{equation*}
  \env_1 \vDash (\sysS{\ptit{C}_4}) \; \piCostBisEq \; (\sysS{\ptit{C}_3}) \; \piCostBisEq \; (\sysS{\ptit{C}_2})
\end{equation*}

However, we may want a finer preorder where $\ptit{C}_4$ is considered to be (strictly) more efficient than $\ptit{C}_2$, which is in turn more efficient than $\ptit{C}_3$. The underlying reasoning for this would be that $\ptit{C}_4$ uses the least amount of expensive operations; by contrast $\ptit{C}_2$ keeps on allocating (and deallocating) new channels for each iteration, and   $\ptit{C}_3$ allocates (and deallocates) two new channels for every iteration.  We can characterise this preorder as follows. First we generate the derived costed LTS using the rule \rtit{lDer2} below --- $\lfloor k \rfloor$ maps all negative integers to $0$, leaving positive integers unaltered. 
\begin{equation*}
  \begin{prooftree}
    \confESP \;\piRedDecCost{\;\mu\;}{k}\; \conf{\env'}{\sys{\sV'}{P'}}
  \justifiedBy{lDer2}
  \confESP \;\piRedDecCostPost{\;\mu\;}{\lfloor k \rfloor }\; \conf{\env'}{\sys{\sV'}{P'}}
\end{prooftree}
\end{equation*}
Then, after adapting \defref{def:amortized-typed-bisim} to this derived LTS, denoting such a bisimulation relation as \piCostBisTwo, we can define the refined preorder, denoted as \piCostBisTre, as follows:
\begin{equation*}
  \env \vDash \sysSP \;\piCostBisTre\; \sys{\sVV}{Q} \quad\deftxt\quad
  \begin{cases}
    \env \vDash \sysSP\, \piCostBis\, \sys{\sVV}{Q}
    \text{
      and } \\
  \env \vDash \sys{\sVV}{Q} \,\piCostBis\, \sysSP
    \text{ implies }\env \vDash \sysSP\, \piCostBisTwo\,
    \sys{\sVV}{Q}
  \end{cases}
\end{equation*}
The new refined preorder $\piCostBisTre$ above requires that \sysSP is at least as efficient as \sys{\sVV}{Q} (possibly more) when it comes to memory leaks, \ie \piCostBis, and moreover, whenever they are equally efficient \wrt these leaks, \sysSP must also be as efficient (possibly more) \wrt memory allocations, \ie \piCostBisTwo.

\subsection{Properties of \,\texorpdfstring{\piCostBis}{piCostBis}}
\label{sec:bisim-results}

We show that our bisimulation relation of Definition~\ref{def:amortized-typed-bisim} observes a number of properties that are useful when reasoning about resource efficiency; see Example~\ref{eg:Clients-Bisim-continued} below.  Lemmas~\ref{lem:reflexivity} and \ref{lem:transitivity} prove that the relation is in fact a preorder, whereas Lemma~\ref{lem:symmetry-bound} outlines conditions where symmetry can be recovered.  Finally, Theorem~\ref{thm:costed-bisim-compositionality} shows that this preorder is preserved under (valid) context; this is the main result of the section. 

First off, we show that $\piCostBis$ is a preorder following Lemma~\ref{lem:reflexivity} (where $\sigma_\env$ would be the identity) and Lemma~\ref{lem:transitivity}.

\begin{lem}[Reflexivity upto Renaming]\label{lem:reflexivity} Whenever the triple \confESP\ is a configuration, then
  ${\env \vDash (\sysSP)\sigma_\env \piCostBisEq \sysSP}$
\end{lem}

\begin{proof}
  By coinduction, by showing that the family of relations
  \begin{displaymath}
    \sset{\langle\env,0,(\sysSP)\sigma_\env, \sysSP \rangle \;|\; \confESP \text{ is a configuration} }
  \end{displaymath}
 is a bisimulation.
\end{proof}

\begin{lem}[Transitivity]\label{lem:transitivity}
Whenever   $\env \vDash \sysSP \,\piCostBis\, \sys{\sV'}{P'}$ and  $\env \vDash \sys{\sV'}{P'} \,\piCostBis\, \sys{\sV''}{P''}$ then  ${\env \vDash \sysSP \,\piCostBis\, \sys{\sV''}{P''}}$
\end{lem}

\begin{proof}
  $\env \vDash \sysSP \,\piCostBis\, \sys{\sV'}{P'}$ implies that there exists some $n\geq 0$ and corresponding bisimulation relation justifying $\env \vDash \sysSP \,\piCostBisAmm{n}\, \sys{\sV'}{P'}$.  The same applies for $\env \vDash \sys{\sV'}{P'} \,\piCostBis\, \sys{\sV''}{P''}$ and some $m\geq 0$.  From these two relations, one can construct a corresponding bisimulation  justifying  $\env \vDash \sysSP \,\piCostBisAmm{{n+m}}\, \sys{\sV''}{P''}$.
\end{proof}

\begin{cor}[Preorder]\label{cor:preorder}
  $\piCostBis$ is a preorder.
\end{cor}
\begin{proof}
  Follows from Lemma~\ref{lem:reflexivity} (for the special case where $\sigma_\env$ is the identity) and Lemma~\ref{lem:transitivity}.
\end{proof}

\medskip

We can define a restricted form of amortised typed bisimulation, in analogous fashion to Definition~\ref{def:amortized-typed-bisim}, whereby the credit is \emph{capped at some upper bound}, \ie some natural number $m$. We refer to such relations as \emph{Bounded Amortised Typed-Bisimulations} and write $$\env \vDash^m \sysSP \,\piCostBisAmm{n}\, \sys{\sVV}{Q}$$ to denote that \confESP\ and \confE{\sys{\sVV}{Q}} are related by some amortised typed-indexed bisimulation at index \env and credit $n$, and where every credit in this relation is less than or equal to $m$; whenever the precise credit $n$ is not important we elide it and simply write $\env \vDash^m \sysSP \,\piCostBis\, \sys{\sVV}{Q}$.  We can show that bounded amortised typed-bisimulations are symmetric.

\begin{lem}[Symmetry]\label{lem:symmetry-bound}
  $\env \vDash^m \sysSP \,\piCostBis\, \sys{\sVV}{Q}$ implies $\env \vDash^m \sys{\sVV}{Q} \,\piCostBis\, \sysSP$
\end{lem}

\begin{proof}
  If \relR\ is the bounded amortised typed relation justifying $\env \vDash^m \sysSP \,\piCostBis\, \sys{\sVV}{Q}$, we define the amortised typed relation
  \begin{displaymath}
    \relR_\text{sym} = \sset{
      \langle \env, (m-n), \sys{\sVV}{Q}, \sysSP \rangle \;|\; \langle \env, n, \sysSP, \sys{\sVV}{Q} \rangle \in \relR
     }
  \end{displaymath} 
  and show that it is a bounded amortised typed bisimulation as well.  Consider an arbitrary pair of configurations $\env \vDash \sys{\sVV}{Q} \,\relR_\text{sym}^{\,m-n}\, \sysSP$:
  \begin{itemize}
  \item Assume $\confE{\sys{\sVV}{Q}} \piRedDecCost{\mu}{l} \conf{\env'}{\sys{\sVV'}{Q'}}$.  From the definition of  $\relR_\text{sym}$, it must be the case that  $ \langle \env, n, \sysSP, \sys{\sVV}{Q} \rangle \in \relR$. Since \relR\ is a bounded amortised typed bisimulation, we know that $\confE{\sys{\sV}{P}} \piRedWDecCost{\hat{\mu}}{l} \conf{\env'}{\sys{\sV'}{P'}}$ where  $ \langle \env', n+l-k, \sys{\sV'}{P'}, \sys{\sVV'}{Q'} \rangle \in \relR$.  We however need to show that $ \langle \env', ((m - n) +k-l), \sys{\sVV'}{Q'}, \sys{\sV'}{P'} \rangle \in \relR_\text{sym}$, which follows from the definition of $\relR_\text{sym}$  and the fact that $\bigl(m - (n + l - k)\bigr) =  (m-n) + k - l$.  

What is left  to show is that $\relR_\text{sym}$  is an amortised typed bisimulation bounded by $m$, \ie we need to show that $0 \leq  (m-n) + k - l \leq m$.  Since \relR\ is an $m$-bounded amortised typed bisimulation, we know that $ 0 \leq (n + l -k) \leq m$ from which we can drive $- m \leq -(n + l -k) \leq 0$ and, by adding $m$ throughout we obtain $0 \leq \bigl( m -(n + l -k) =  (m-n) + k - l\bigr) \leq m $ as required.
\item The dual case for $\confE{\sys{\sV}{P}} \piRedDecCost{\mu}{l} \conf{\env'}{\sys{\sV'}{P'}}$ is analogous. \qedhere
  \end{itemize}
\end{proof}

\medskip

\emph{Contextuality} is an important property for any behavioural relation.  In our case, this means that two systems \sysSP and \sys{\sVV}{Q} related by \piCostBis under \env, remain related when extended with an additional process, $R$, whenever this process runs safely over the respective resource environments \sV\ and \sVV, and observes the type restrictions and guarantees assumed by \env (and dually, those of the respective existentially-quantified type environments for \sysSP and \sys{\sVV}{Q}). Following Definition~\ref{def:configuration}, for these conditions to hold, contextuality requires $R$ to typecheck \wrt a sub-environment of \env, say $\env_1$ where $\env = \env_1,\env_2$, and correspondingly strengthens the relation of  \sysS{P\piParal R} and \sys{\sVV}{Q\piParal R} in \piCostBis under the remaining sub-environment, $\env_2$.  Stated otherwise, contextuality requires the transfer of the respective permissions associated with the observer sub-process $R$ from the observer environment \env; this is crucial in order to preserve consistency, thus safety, in the respective configurations.  The formulation of Theorem~\ref{thm:costed-bisim-compositionality}, proving contextuality for \piCostBis, follows this reasoning.  It relies on a list of lemmas outlined below.

\begin{lem}[Weakening]\label{lem:weakening} 
If
  $\confESP \;\piRedDecCostPre{\;\mu\;}{k}\; \conf{\env'}{\sys{\sV'}{P'}}$
then 
  $\conf{(\env,\envv)}{\sysSP} \;\piRedDecCostPre{\;\mu\;}{k}\; \conf{(\env',\envv)}{\sys{\sV'}{P'}}$.
(These may or may not be configurations.)
\end{lem}

\begin{proof}  
  By rule induction on $ \confESP \piRedDecCostPre{\;\mu\;}{k} \conf{\env'}{\sys{\sV'}{P'}}$.   Note that, in the case of $\actall$, the action can still be performed. 
\end{proof}

\begin{lem}[Strengthening]\label{lem:strenghtening}
If
  $\conf{(\env,\envv)}{\sysSP}\piRedDecCostPre{\;\mu\;}{k}\conf{(\env',\envv)}{\sys{\sV'}{P'}}$
then
  $\conf{\env}{\sysSP}\piRedDecCostPre{\;\mu\;}{k}\conf{\env'}{\sys{\sV'}{P'}}$.
\end{lem}

\begin{proof}
  By rule induction on $ \conf{\env,\envv}{P}\piRedDecCostPre{\;\mu\;}{k}\conf{\env',\envv}{P'}$.    Note that strengthening is restricted to the part of the environment that remains unchanged ($\envv$ is the same on the left and right hand side) --- otherwise the property does not hold for actions \actout{c}{\vec{d}} and \actin{c}{\vec{d}}.
\end{proof}

\begin{lem}\label{prop:consistency-env-struct}
  If $\env,\envv$ is consistent and $\envv\struct \envv'$ then $\env,\envv'$ is consistent and $\env,\envv\struct \env,\envv'$
\end{lem}

\begin{proof}
  As in \cite{EFH:uniqueness:journal:12}.
\end{proof}

\begin{lem}[Typing Preserved by \piStruct]\label{prop:steq-typing}
   $\env\vdash P$  and $P \piStruct Q$ implies $\env\vdash Q$
\end{lem}

\begin{proof}
 As in \cite{EFH:uniqueness:journal:12}.
\end{proof}

\begin{lem}[Environment Structural Manipulation Preserves Bisimulation]\label{prop:env-struct-rules-bisim}
  \begin{displaymath}
  \env \vDash S\;\piCostBisAmm{n}\; T \text{ and }\env \struct \env' \text{ implies } \env' \vDash S\;\piCostBisAmm{n}\; T 
\end{displaymath}
\end{lem}
\begin{proof}
  By coinduction.  We define the quarternary relation $$\sset{\langle \env',n,S,T \rangle | \; \env \vDash S\;\piCostBisAmm{n}\; T \text{ and }\env \struct \env' }$$ and show that it observes the transfer property of \defref{def:amortized-typed-bisim}.
\end{proof}

\begin{lem}[Bisimulation and Structural Equivalence]\label{lem:bisim-struct-equiv}
  \begin{displaymath}
    P\piStruct Q \text{ and }\confE{\sys{\sV}{P}} \piRedDecCostPad{\mu}{k} \conf{\envv}{\sys{\sV'}{P'}}  \text{ implies } \confE{\sys{\sV}{Q}} \piRedDecCostPad{\mu}{k} \piCfg{\envv}{\sys{\sV'}{Q'}} \text{ and } P'\piStruct Q' 
  \end{displaymath}
\end{lem}

\begin{proof}
  By rule induction on $P\piStruct Q$ and then a case analysis of the  rules permitting $\confE{\sys{\sV}{P}} \piRedDecCostPad{\mu}{k} \conf{\envv}{\sys{\sV'}{P'}}$.
\end{proof}

\begin{cor}[Structural Equivalence and Bisimilarity]\label{cor:Struct-eq-implies-bisim}
  \begin{math}
    P\piStruct Q \text{ implies } \env \vDash \sys{\sV}{P} \,\piCostBisAmm{n}\, \sys{\sV}{Q}   
  \end{math}
for arbitrary $n$ and \env where \confE{\sys{\sV}{P}}  and \confE{\sys{\sV}{Q}} are configurations.
\end{cor}

\begin{proof}
  By coinduction and Lemma~\ref{lem:bisim-struct-equiv}.
\end{proof}

\begin{lem}[Renaming]\label{lem:costed-bisim-renaming} 
  If
\begin{math}
\env, \envv \vDash (\sys{M}{P}) \; \piCostBisAmm{n} \; (\sys{N}{Q})
\end{math}
 then
\begin{math}
\env, (\envv\sigma_\env) \vDash (\sys{M}{P})\sigma_\env \; \piCostBisAmm{n} \; (\sys{N}{Q})\sigma_\env
\end{math}
\end{lem}

\begin{proof}
  By coinduction.
\end{proof}

\smallskip
\begin{thm}[Contextuality]\label{thm:costed-bisim-compositionality}
  If
\begin{math}
\env, \envv \vDash (\sys{M}{P}) \; \piCostBisAmm{n} \; (\sys{N}{Q})
\end{math}
and $\Delta \vdash R$ then
\begin{displaymath}
  \env \vDash (\sys{M}{P \piParal R}) \; \piCostBisAmm{n} \; (\sys{N}{Q \piParal R})
\quad\text{ and }\quad
  \env \vDash (\sys{M}{R \piParal P}) \; \piCostBisAmm{n} \; (\sys{N}{R \piParal Q})
\end{displaymath}
\end{thm}

\begin{proof}  We define the family of relations $\relR^{\env,n}$ to be the least one satisfying the rules
{\small
\begin{equation*}
\begin{prooftree}
\env \vDash (\sys{M}{P}) \piCostBisAmm{n} (\sys{N}{Q})
\justifies
\env \vDash (\sys{M}{P}) \; \relR^n \; (\sys{N}{Q})
\end{prooftree}
\qquad
\begin{prooftree}
\env,\envv \vDash (\sys{M}{P}) \; \relR^n \; (\sys{N}{Q}) \quad
\envv\vdash R
\justifies
\env \vDash (\sys{M}{P \piParal R}) \; \relR^n \; (\sys{N}{Q \piParal R})
\end{prooftree}
\qquad
\begin{prooftree}
\env,\envv \vDash (\sys{M}{P}) \; \relR^n \; (\sys{N}{Q}) \quad
\envv\vdash R
\justifies
\env \vDash (\sys{M}{R \piParal P}) \; \relR^n \; (\sys{N}{R \piParal Q})
\end{prooftree}
\end{equation*}
}
and then show that $\relR^{\env,n}$ is a costed typed bisimulation at \env and $n$  (up to $\piStruct$).  Note that the premise of the first rule implies that both 
  $\conf{\env,\envv}{\sys{M}{P}}$ 
and 
  $\conf{\env,\envv}{\sys{N}{Q}}$ 
are configurations.  We consider only the transitions of the left hand configurations for second case of the relation; the first
is trivial and the third is analogous to the second.  Although  the relation is not
symmetric, the transition of the right hand configurations are analogous to those of the left hand configurations.  There are three cases  to consider.
\begin{enumerate}

\item Case the action was instigated by $P$, \ie we have:  
\begin{equation}
\begin{prooftree}
  \[\confE{(\sys{\sV}{P})\sigma_\env} \;\piRedDecCostPre{\;\mu\;}{l}\; \conf{\env'}{\sys{\sV'}{P'}}
\justifiedBy{lPar-L}
  \confE{(\sys{\sV}{P\piParal R})\sigma_\env} \;\piRedDecCostPre{\;\mu\;}{l}\; \conf{\env'}{\sys{\sV'}{P'\piParal R\sigma_\env}}\]
\justifiedBy{lRen}
    \confE{\sys{\sV}{P\piParal R}} \;\piRedDecCost{\;\mu\;}{l}\; \conf{\env'}{\sys{\sV'}{P'\piParal R\sigma_\env}}
\end{prooftree}
\label{eq:27:bis}
\end{equation}
By Lemma~\ref{lem:weakening} (Weakening), \rtit{lRen} and \eqref{eq:27:bis} we obtain
\begin{align} \label{eq:27:biss}
\conf{\env,(\envv\sigma_\env)}{(\sysSP)\sigma_\env} \piRedDecCostPad{\mu}{l} \conf{\env',(\envv\sigma_\env)}{\sys{\sV'}{P'}}
\end{align}
Lemma~\ref{lem:costed-bisim-renaming} can be extended to $\relR^n$ is straightforward fashion, and from the case assumption $\env,\envv \vDash \,\sysSP\, \relR^n\, \sys{\sVV}{Q}$ (defining $\relR^{\env,n}$)  and the extension of Lemma~\ref{lem:costed-bisim-renaming} to $\relR^n$ we obtain:
\begin{equation}\label{eq:27a:bis}
\env,(\envv\sigma_\env) \vDash (\sysSP)\sigma_\env \relR^n (\sys{\sVV}{Q})\sigma_\env
\end{equation}
Hence by  \eqref{eq:27a:bis}, \eqref{eq:27:biss} and I.H. 
there exists a 
$\sys{N'}{Q'}$ such that 
\begin{align}
\label{eq:28:bis}
&\conf{\env,(\envv\sigma_\env)}{(\sys{\sVV}{Q})\sigma_{\env}} \;\piRedWDecCost{\hat{\mu}}{k}\; \conf{\env',(\envv\sigma_\env)}{\sys{\sVV'}{Q'}} \\
\label{eq:28b:bis}
\quad\text{where}& \quad \env',(\envv\sigma_\env)\vDash (\sys{\sV'}{P'})\;\relR^{n+k-l}\; (\sys{\sVV'}{Q'})
\end{align}
By \eqref{eq:28:bis} and \rtit{lRen}, were $\env_1 = \env,(\envv\sigma_\env)$, we obtain
\begin{equation}\label{eq:28c:bis}
  \conf{\env,(\envv\sigma_\env)}{\bigl((\sys{\sVV}{Q})\sigma_{\env})\bigr)\sigma'_{\env_1}} \quad\bigl(\piRedTauDecCostPre{k_1}\bigr)\piRedDecCostPre{\hat{\mu}}{k_2}\bigl(\piRedTauDecCostPre{k_3}\bigr)\quad \conf{\env',(\envv\sigma_\env)}{\sys{\sVV'}{Q'}}
\end{equation}
where $k=k_1+k_2+k_3$. By \rtit{lPar} Lemma~\ref{lem:strenghtening} (Strengthening) and \eqref{eq:28c:bis} we deduce 
\begin{equation}
\label{eq:41:bis}
\conf{\env}{\bigl((\sys{\sVV}{Q})\sigma_{\env})\bigr)\sigma'_{\env_1} \piParal R\sigma_\env} \quad\bigl(\piRedTauDecCostPre{k_1}\bigr)\piRedDecCostPre{\hat{\mu}}{k_2}\bigl(\piRedTauDecCostPre{k_3}\bigr)\quad \conf{\env'}{\sys{\sVV'}{Q'\piParal R\sigma_\env}}
\end{equation}
From \tproc{\envv}{R} we know
\begin{equation}\label{eq:41a:bis}
\tproc{\envv\sigma_\env}{R\sigma_\env}
\end{equation}
and, from $\env_1 = \env,(\envv\sigma_\env)$ and \defref{def:renaming} (Renaming Modulo Environments), we know that $(R\sigma_\env)\sigma'_{\env_1} = R\sigma_\env$ since the renaming does not modify any of the names in the domain of $\env_1$, hence of  $\envv\sigma_\env$.  Also, from \defref{def:renaming},  $\sigma'_{\env_1}$ is also a substitution modulo $\env$ and can therefore refer to it as $\sigma'_\env$, thereby rewriting \eqref{eq:41:bis} as
\begin{align}
\label{eq:42:bis}
\conf{\env}{\bigl(\sys{\sVV}{Q}\piParal R \bigr)\sigma_{\env}\sigma'_{\env} } \quad\bigl(\piRedTauDecCostPre{k_1}\bigr)\piRedDecCostPre{\hat{\mu}}{k_2}\bigl(\piRedTauDecCostPre{k_3}\bigr)\quad \conf{\env'}{\sys{\sVV'}{Q'\piParal R\sigma_\env}}
\end{align}
From \eqref{eq:42:bis} and \rtit{lRen} we thus obtain
\begin{align*}
\conf{\env}{\sys{\sVV}{Q \piParal R}} \quad\piRedWDecCost{\;\hat{\mu}\;}{k}\quad \conf{\env'}{\sys{\sVV'}{(Q'\piParal R\sigma_\env)}}
\end{align*}
This is our matching move since and by \eqref{eq:28b:bis}, \eqref{eq:41a:bis} and the definition of \relR\ we obtain 
  $\env'\vDash (\sys{\sV'}{P'\piParal R\sigma_\env})\;\relR^{n+l-k}\; (\sys{\sVV'}{Q'\piParal R\sigma_\env})$.

\item Case the action was instigated by $R$, \ie we have:
\begin{equation}
\begin{prooftree}
  \[\confE{(\sys{\sV}{R})\sigma_\env} \quad\piRedDecCostPre{\;\mu\;}{l}\quad \conf{\env'}{\sys{\sV'}{R'}}
\justifiedBy{lPar-R}
  \confE{\bigl(\sys{\sV}{P\piParal R}\bigr)\sigma_\env} \quad\piRedDecCostPre{\;\mu\;}{l}\quad \conf{\env'}{\sys{\sV'}{P\piParal R'}}\]
\justifiedBy{lRen}
  \confE{\sys{\sV}{P\piParal R}} \quad\piRedDecCost{\;\mu\;}{l}\quad \conf{\env'}{\sys{\sV'}{P\piParal R'}}
\end{prooftree}
\label{eq:62:bis}
\end{equation}
The proof proceeds by case analysis of $\mu$ whereby the most interesting cases are when $l=+1$ or $l=-1$.  We here show the case for when $l=-1$ (the other case is analogous).  By Lemma~\ref{lem:transition-structure} we know that either $\mu = \actfree{c} $  and
\begin{align*}
  &\sV\sigma_\env=\sV',c 
  &  & R' \piStruct R\sigma_\env
  &  & \env = \env',\emap{c}{\chantypUnq{\tVlst}}
\end{align*}
or else that $\mu = \tau$ and
    \begin{align} 
      \label{eq:63:bis}
      &\sV\sigma_\env=\sV',c \\
      \label{eq:76:bis}
      &R\sigma_\env \piStruct \piFree{c}{R_1}\piParal R_2 \;\text{ and }\; R'\piStruct R_1\piParal R_2 \\
      \label{eq:77:bis}
      & \env=\env' \\
      \label{eq:78:bis}
      &\envv\sigma_\env\struct \envv',\emap{c}{\chantypU{\tVlst}} \; \text{ and }\;\tproc{\envv'}{R'}.
    \end{align}
   We here focus on the latter case, \ie when $\mu = \tau$.  The main complication in finding a matching move for this subcase is that of inferring a pair of resultant systems (one of which is \conf{\env}{\sys{\sV'}{P\piParal R'}}) that are related by \relR\ by using the inductive nature of the relation definition. To be able to do so, we need to mimic the effect of $R$'s deallocation transition on \sV\ in the corresponding system \sys{\sVV}{Q}; we do this with the help of an appropriate external deallocation transition \actfree{c}.        

 By the extension of Lemma~\ref{lem:costed-bisim-renaming} to $\relR^n$ we know ${\env,\envv\sigma_\env \vDash  (\sysS{P})\sigma_\env\;\relR^n\; (\sys{\sVV}{Q})\sigma_\env}$, and by \eqref{eq:78:bis} and a straightforward extension of Lemma~\ref{prop:env-struct-rules-bisim} to $\relR$  we obtain
    \begin{align}
      \label{eq:79:bis}
      &\env,\envv',\emap{c}{\chantypU{\tVlst}} \vDash (\sysS{P})\sigma_\env\;\relR^n\; (\sys{\sVV}{Q})\sigma_\env
    \end{align}
    and by \eqref{eq:63:bis} and \rtit{lFreeE} we deduce
    \begin{align*}
      &\conf{\env,\envv',\emap{c}{\chantypU{\tVlst}}}{\bigl(\sys{\sV}{P}\bigr)\sigma_\env}  \piRedDecCostPad{{\actfree{c}}}{-1} \conf{\env,\envv'}{\bigl(\sys{\sV'}{P}\bigr)\sigma_\env}
    \end{align*}
     and by \eqref{eq:79:bis} and I.H. there exists a matching move
     \begin{align}
       \label{eq:81:bis}
        &\conf{\env,\envv',\emap{c}{\chantypU{\tVlst}}}{(\sys{\sVV}{Q})\sigma_\env } \piRedWDecCostPad{\actfree{c}}{k} \conf{\env,\envv'}{\sys{\sVV'}{Q'}} \\
        \label{eq:82a:bis}
        \text{ and }& \env,\envv' \vDash \sys{\sV'}{P'} \;\relR^{n+k-(-1)}\; \sys{\sVV'}{Q'}
     \end{align}
     By \eqref{eq:81:bis} and \rtit{lRen}, for $k = k_{1}-1+k_{2}$, we know
     \begin{align} \label{eq:82b:bis}
       &\conf{\env,\envv',\emap{c}{\chantypU{\tVlst}}}{\bigl((\sys{\sVV}{Q})\sigma_\env \bigr)\sigma'_{\env_2}} \quad\piRedTauDecCostPre{k_1}\quad \conf{\env,\envv',\emap{c}{\chantypU{\tVlst}}}{\sys{\sVV''}{Q''}}\\
       \label{eq:82:bis}
        \text{where }&  \env_2=\env,\envv',\emap{c}{\chantypU{\tVlst}} \text{ (used in $\sigma'_{\env_2}$ above)} \\
      \label{eq:82c:bis}
      & \conf{\env,\envv',\emap{c}{\chantypU{\tVlst}}}{\sys{\sVV''}{Q''}}\quad\piRedDecCostPre{\actfree{c}}{-1}\quad\conf{\env,\envv'}{\sys{\sVV'''}{Q''}}\\
      \label{eq:82d:bis}
     &\conf{\env,\envv'}{\sys{\sVV'''}{Q''}}\quad\piRedTauDecCostPre{k_2} \quad\conf{\env,\envv'}{\sys{\sVV'}{Q'}}
     \end{align}
From \eqref{eq:82b:bis},~\eqref{eq:82d:bis}, \rtit{lPar-L} and Lemma~\ref{lem:strenghtening} (Strengthening) we obtain:
\begin{align}
  \label{eq:83a:bis}
       &\conf{\env}{\bigl((\sys{\sVV}{Q})\sigma_\env \bigr)\sigma'_{\env_2}\piParal R\sigma_\env} \quad\piRedTauDecCostPre{k_1}\quad \conf{\env}{\sys{\sVV''}{Q'' \piParal (R\sigma_\env)}}\\
   \label{eq:83b:bis}
      &\conf{\env}{\sys{\sVV'''}{Q''\piParal R'}}\quad\piRedTauDecCostPre{k_2} \quad\conf{\env}{\sys{\sVV'}{Q'\piParal R'}}
\end{align}
Also, from \eqref{eq:82c:bis} and Lemma~\ref{lem:transition-structure} (Transition and Structure) we deduce that $\sVV''=\sVV''',c$ and thus, from \eqref{eq:76:bis}, \rtit{lFree}, \rtit{lPar-R} we obtain:
\begin{equation}
  \label{eq:83c:bis}
  \conf{\env}{\sys{\sVV''}{Q''\piParal R\sigma_\env}}\quad\piRedDecCostPre{\;\tau\;}{-1}\quad\conf{\env}{\sys{\sVV'''}{Q''\piParal R'}}
\end{equation}
By \eqref{eq:78:bis} and \eqref{eq:82:bis}, we know that we can find an alternative renaming function $\sigma''_{\env_3}$, where $\env_3 = \env,(\envv\sigma_\env)$, in a way that, from \eqref{eq:83a:bis}, we can obtain
\begin{equation}
  \label{eq:83d:bis}
  \conf{\env}{\bigl((\sys{\sVV}{Q})\sigma_\env \bigr)\sigma''_{\env_3}\piParal R\sigma_\env} \quad\piRedTauDecCostPre{k_1}\quad \conf{\env}{\sys{\sVV''}{Q'' \piParal (R\sigma_\env)}}
\end{equation}
Now, by $\tproc{\envv}{R}$ we know $\tproc{\envv\sigma_\env}{R\sigma_\env}$ and subsequently, by \defref{def:renaming} and \eqref{eq:82:bis} we know $(R\sigma_\env)\sigma''_{\env_3} = R\sigma_\env$. Thus, we can rewrite  $\bigl((\sys{\sVV}{Q})\sigma_\env \bigr)\sigma''_{\env_3}\piParal R\sigma_\env$ in \eqref{eq:83d:bis} as ${\bigl((\sys{\sVV}{Q}\piParal R)\sigma_\env \bigr)\sigma''_{\env_3}}$.  
Merging \eqref{eq:83d:bis},~\eqref{eq:83c:bis} and ~\eqref{eq:83b:bis} we obtain:
\begin{equation*}
  \conf{\env}{\bigl((\sys{\sVV}{Q}\piParal R)\sigma_\env \bigr)\sigma''_{\env_3}} \quad\piRedTauDecCostPre{k_1}\piRedDecCostPre{\;\tau\;}{-1}\piRedTauDecCostPre{k_2} \quad\conf{\env}{\sys{\sVV'}{Q'\piParal R'}}
\end{equation*}
By \defref{def:renaming} we know that $\sigma''_{\env_3}$ can be rewritten as $\sigma''_{\env}$ and thus by \rtit{lRen} we obtain the matching move
\begin{equation*}
  \conf{\env}{\sys{\sVV}{Q}\piParal R} \piRedWDecCostPad{\tau}{k} \conf{\env}{\sys{\sVV'}{Q'\piParal R'}}
\end{equation*}
because by \eqref{eq:82a:bis}, \eqref{eq:78:bis} and the definition of $\relR$ we know that $$\env \vDash \sys{\sV'}{P'\piParal R'} \;\relR^{n+k-(-1)}\; \sys{\sVV'}{Q'\piParal R'}.$$

\item Case the action resulted from an interaction between $P$ and $R$, \ie we have:
\begin{equation}
\begin{prooftree}
\[
\conf{\env_1}{(\sysSP)\sigma_\env} \piRedDecCostPre{\actout{c}{\vec{d}}}{0} \conf{\env'_1}{\sys{\sV'}{P'}} \qquad 
\conf{\env_2}{(\sysS{R})\sigma_\env} \piRedDecCostPre{\actin{c}{\vec{d}}}{0} \conf{\env'_2}{\sys{\sV'}{R'}}
\justifiedBy{lCom-L}
\confE{(\sysS{P\piParal R})\sigma_\env} \piRedDecCostPre{\;\;\tau\;\;}{0} \confE{\sys{\sV'}{P'\piParal R'}}
\]
\justifiedBy{lRen}
\confE{\sysS{P\piParal R}} \piRedDecCostPad{\;\;\tau\;\;}{0} \confE{\sys{\sV'}{P'\piParal R'}} 
\end{prooftree}
\label{eq:29:bis}
\end{equation}
By the two top premises of \eqref{eq:29:bis} and Lemma~\ref{lem:transition-structure}
we know
\begin{align}
\label{eq:31a:bis}
\sV\sigma_\env &= \sV'\\
\label{eq:31:bis}
  P\sigma_\env  & \piStruct \piOut{c}{\vec{d}}P_1 \piParal P_2 & 
  P' & \piStruct P_1 \piParal P_2 \\
\label{eq:32:bis}
  R\sigma_\env  & \piStruct \piIn{c}{\vec{x}}R_1 \piParal R_2 & 
  R' & \piStruct R_1 \subC{\vec{d}}{\vec{x}} \piParal R_2 
\end{align}
From $\envv\vdash R$ we obtain $\envv\sigma_\env\vdash R\sigma_\env$, and by \eqref{eq:32:bis},  $\envv\vdash R$ and Inversion we obtain
\begin{align}
  \label{eq:33:bis}
  &\envv\sigma_\env \struct \envv_1,\envv_2,\emap{c}{\chantypA{\tVVlst}} \\
  \label{eq:34:bis}
  &\envv_1,\emap{c}{\chantyp{\tVVlst}{{\aV-1}}},\emap{\vec{x}}{\tVVlst} \vdash R_1 \\
  \label{eq:35:bis}
  &\envv_2 \vdash R_2
\end{align}
Note that through~\eqref{eq:34:bis}  we know that
\begin{equation}\label{eq:30:bis}
\envmap{c}{\chantyp{\tVVlst}{{\aV-1}}}\text{ is
defined.}
\end{equation}
By \eqref{eq:34:bis}, the Substitution Lemma (Lemma 4.4 from
\cite{EFH:uniqueness:journal:12})  and \eqref{eq:35:bis} we obtain 
\begin{align} 
  \label{eq:36}
  &\envv_1,\envv_2, \emap{c}{\chantyp{\tVVlst}{{\aV-1}}},\emap{\vec{d}}{\tVVlst} \vdash  R_1\subC{\vec{d}}{\vec{x}} \piParal R_2 
\end{align}
From the assumption defining $\relR$, and Lemma~\ref{lem:costed-bisim-renaming} we obtain
\begin{equation}\label{eq:37a:bis}
\env,(\envv\sigma_\env) \;\vDash\; (\sysSP)\sigma_\env \;\relR^n\; (\sys{\sVV}{Q})\sigma_\env,
\end{equation}
and
by \eqref{eq:33:bis} and Proposition~\ref{prop:consistency-env-struct} we know that
  $\env,(\envv\sigma_\env)\struct \env,\envv_1,\envv_2, \emap{c}{\chantypA{\tVVlst}}$ 
and also that 
  $\env,\envv_1,\envv_2, \emap{c}{\chantypA{\tVVlst}}$ 
is consistent.  Thus by  \eqref{eq:37a:bis}
and Lemma~\ref{prop:env-struct-rules-bisim} we deduce
\begin{align}
  \label{eq:37b:bis}
  &\env,\envv_1,\envv_2, \emap{c}{\chantypA{\tVVlst}} \;\vDash\;  (\sysSP)\sigma_\env \;\relR^n\; (\sys{\sVV}{Q})\sigma_\env
\end{align}
Now by  \eqref{eq:30:bis}, \eqref{eq:31:bis}, \eqref{eq:31a:bis}, \rtit{lOut}, \rtit{lPar-L}, \rtit{lRen} and
Lemma~\ref{lem:bisim-struct-equiv} we deduce
\begin{align}
  \label{eq:38:bis}
  &\conf{\env,\envv_1,\envv_2, \emap{c}{\chantypA{\tVVlst}}}{(\sysSP)\sigma_\env} \piRedDecCostPad{\actout{c}{\vec{d}}}{0} \conf{\env,\envv'_1,\envv'_2, \emap{c}{\chantyp{\tVVlst}{{\aV-1}}},\emap{\vec{d}}{\tVVlst}}{\sys{\sV'}{P'}}
\end{align}
and hence by \eqref{eq:37b:bis} and I.H. we obtain
\begin{align}
  \label{eq:39:bis}
  &\conf{\env,\envv_1,\envv_2, \emap{c}{\chantypA{\tVVlst}}}{(\sys{\sVV}{Q})\sigma_{ \env}} \piRedWDecCostPad{\;\actout{c}{\vec{d}}\;}{k} \conf{\env,\envv'_1,\envv'_2, \emap{c}{\chantyp{\tVVlst}{{\aV-1}}},\emap{\vec{d}}{\tVVlst}}{\sys{\sVV'}{Q'}} \\
  \label{eq:40:bis}
  &\text{such that } \env,\envv_1,\envv_2, \emap{c}{\chantyp{\tVVlst}{{\aV-1}}},\emap{\vec{d}}{\tVVlst} \vDash (\sys{\sV'}{P'})\;\relR^{n+k-0}\; (\sys{\sVV'}{Q'})
\end{align}
From \eqref{eq:39:bis} and \rtit{lRen} we know
\begin{align}
  \label{eq:92:bis}
  &\conf{\env,\envv_1,\envv_2, \emap{c}{\chantypA{\tVVlst}}}{\bigl((\sys{\sVV}{Q})\sigma_{ \env}\bigr)\sigma'_{\env_4}} \piRedTauDecCostPre{k_1} \conf{\env,\envv_1,\envv_2, \emap{c}{\chantypA{\tVVlst}}}{\sys{\sVV''}{Q''}} \\
  \label{eq:93:bis}
   &\conf{\env,\envv_1,\envv_2, \emap{c}{\chantypA{\tVVlst}}}{\sys{\sVV''}{Q''}} \piRedDecCostPad{\actout{c}{\vec{d}}}{0} \conf{\env,\envv_1,\envv_2, \emap{c}{\chantyp{\tVVlst}{{\aV-1}}},\emap{\vec{d}}{\tVVlst}}{\sys{\sVV''}{Q'''}} \\
   \label{eq:94:bis}
   &\conf{\env,\envv_1,\envv_2, \emap{c}{\chantyp{\tVVlst}{{\aV-1}}},\emap{\vec{d}}{\tVVlst}}{\sys{\sVV''}{Q'''}} \piRedTauDecCostPre{k_2} \conf{\env,\envv_1,\envv_2, \emap{c}{\chantyp{\tVVlst}{{\aV-1}}},\emap{\vec{d}}{\tVVlst}}{\sys{\sVV'}{Q'}} \\
   \label{eq:95:bis}
   &\text{ where } k=k_1 + k_2 \text{ and }\env_4 = \env,\envv_1,\envv_2, \emap{c}{\chantypA{\tVVlst}} \;\text{ ($\env_4$ is used in \eqref{eq:92:bis})}
\end{align}
From \eqref{eq:92:bis},~\eqref{eq:94:bis}, \rtit{lPar-L} and Lemma~\ref{lem:strenghtening} (Strengthening) we obtain:
\begin{align}
  \label{eq:92a:bis}
  &\conf{\env}{\bigl((\sys{\sVV}{Q})\sigma_{ \env}\bigr)\sigma'_{\env_4}\piParal R\sigma_\env} \piRedTauDecCostPre{k_1} \conf{\env}{\sys{\sVV''}{Q''\piParal R\sigma_\env}}\\
  \label{eq:94a:bis}
  &\conf{\env}{\sys{\sVV''}{Q'''\piParal R'}} \piRedTauDecCostPre{k_2} \conf{\env}{\sys{\sVV'}{Q'\piParal R'}}
\end{align}
By \eqref{eq:32:bis}, \rtit{lIn} and \rtit{lPar-L} we can construct (for some $\env_6,\env_7$)
   \begin{align}
     \label{eq:98:bis}
     & \conf{\env_6}{\sys{\sVV''}{R\sigma_\env}} \piRedDecCostPre{\actin{c}{\vec{d}}}{0} \conf{\env_7}{\sys{\sVV''}{R'}}
   \end{align}
and by \eqref{eq:93:bis}, \eqref{eq:98:bis} and \rtit{lCom-L} we obtain
\begin{align}
     \label{eq:99:bis}
     \conf{\env}{\sys{\sVV''}{Q'' \piParal R\sigma_\env}} \piRedDecCostPre{\;\;\tau\;\;}{0} \conf{\env}{\sys{\sVV''}{Q''' \piParal R'}}
\end{align}
By \eqref{eq:33:bis} and \eqref{eq:95:bis}, we know that we can find an alternative renaming function $\sigma''_{\env_5}$, where $\env_5 = \env,(\envv\sigma_\env)$, in a way that, from \eqref{eq:92a:bis}, we can obtain
\begin{equation}
   \label{eq:92b:bis}
  \conf{\env}{\bigl((\sys{\sVV}{Q})\sigma_{ \env}\bigr)\sigma''_{\env_5}\piParal R\sigma_\env} \piRedTauDecCostPre{k_1} \conf{\env}{\sys{\sVV''}{Q''\piParal R\sigma_\env}}
\end{equation}
By \defref{def:renaming},  $\envv\sigma_\env\vdash R\sigma_\env$, \eqref{eq:33:bis}, \eqref{eq:95:bis} we know that $(R\sigma_\env)\sigma''_{\env_5} = R\sigma_\env$, and also that $\sigma''_{\env_5}$ is also a renaming modulo $\env$, so we can denote it as $\sigma''_\env$ and  rewrite $\bigl((\sys{\sVV}{Q})\sigma_{ \env}\bigr)\sigma''_{\env_5}\piParal R\sigma_\env$ as $\bigl((\sys{\sVV}{Q\piParal R})\sigma_{ \env}\bigr)\sigma''_\env$ in \eqref{eq:92b:bis}.  Thus, by  \eqref{eq:92b:bis},~\eqref{eq:99:bis},~\eqref{eq:94a:bis},~\eqref{eq:95:bis} and \rtit{lRen} we obtain the matching move
\begin{equation*}
  \conf{\env}{\sys{\sVV}{Q\piParal R}} \;\piRedWDecCostPad{\;\tau\;}{k}\; \conf{\env}{\sys{\sVV'}{Q'\piParal R'}}
\end{equation*}
since by \eqref{eq:40:bis},~\eqref{eq:36},~\eqref{eq:32:bis} and the definition of \relR\ we  obtain 
\[
  \env \vDash (\sys{\sV'}{P'\piParal R'})\;\relR^{n+k-0}\; (\sys{\sVV'}{Q'\piParal R'}) 
\]
as required. \qedhere
\end{enumerate}
\end{proof}

\begin{exa}[Properties of \piCostBis] \label{eg:Clients-Bisim-continued}
From the proved statements $\env_1 \vDash (\sysS{\ptit{C}_1}) \; \piCostBis \; (\sysS{\ptit{C}_0})$ and   $\env_1 \vDash (\sysS{\ptit{C}_2}) \; \piCostBis \; (\sysS{\ptit{C}_1})$  of Example~\ref{eg:Clients-Bisim}, and by Corollary~\ref{cor:preorder} (Preorder), we may conclude that
\begin{equation}
\label{eq:trans-statement:bis}
\env_1 \vDash (\sysS{\ptit{C}_2}) \; \piCostBis \; (\sysS{\ptit{C}_0})
\end{equation}
 without the need to provide a bisimulation relation justifying \eqref{eq:trans-statement:bis}.  We also note that $\relR'$ of  Example~\ref{eg:Clients-Bisim}, justifying $\env_1 \vDash (\sysS{\ptit{C}_3}) \; \piCostBis \; (\sysS{\ptit{C}_2})$ is a \emph{bounded} amortised typed-bisimulation, and by Lemma~\ref{lem:symmetry-bound} we can also conclude
\begin{equation*}
  \env_1 \vDash (\sysS{\ptit{C}_2}) \; \piCostBis \; (\sysS{\ptit{C}_3})
\end{equation*}
and thus $\env_1 \vDash (\sysS{\ptit{C}_3}) \; \piCostBisEq \; (\sysS{\ptit{C}_2})$.  Finally, by Theorem~\ref{thm:costed-bisim-compositionality}, in order to show that
\begin{equation*}
  \envmap{c}{\chantypW{\tV_1,\tV_2}} \vDash (\sysS{\ptit{S}_1\piParal\ptit{S}_2\piParal\ptit{C}_1}) \; \piCostBiss \; (\sysS{\ptit{S}_1\piParal\ptit{S}_2\piParal\ptit{C}_0})
\end{equation*}
it suffices to abstract away from the common code, $\ptit{S}_1\piParal\ptit{S}_2$, and show  $\env_1 \vDash (\sysS{\ptit{C}_1}) \; \piCostBiss \; (\sysS{\ptit{C}_0}),$ as proved already in  Example~\ref{eg:Clients-Bisim}.
\end{exa}

\section{Characterisation}
\label{sec:characterisation}

In this section we give a sound and complete characterization of bisimilarity
in terms of the reduction semantics of Section~\ref{sec:language}, justifying the bisimulation relation and
the respective LTS as a proof technique for reasoning about the behaviour of $\picr$ processes. Our touchstone behavioural preorder  is based on a costed version of families of reduction-closed barbed
congruences along similar lines to \cite{hennessy:buysell}. In order to limit behaviour to safe computations, these congruences are defined as typed relations (Definition~\ref{def:typed-relation}), where systems are subject to common observers typed by environments.

The observer type-environment delineates the observations that can be made: the observer can only make distinctions for channels that it has a permission for, \ie at least an affine
typing assumption.  
The observations that can be made in our touchstone behavioural preorder are described as \emph{barbs} \cite{HondaTokoro92:AsyncSemantics}  that take into account the permissions owned by the observer.  We require systems related by our behavioral preorder to exhibit the same barbs \wrt a common observer. 

\begin{defi}[Barb]\label{def:barb}
  $(\confE{\sys{M}{P}}) \piBarb{c} \;\deftxt\; (\sys{M}{P})  \piRedAst_{k}\piStruct
(\sys{\sV'}{P'\piParal{\piOut{c}{\vec{d}}{P''}}})  \text{ and } c \in \dom(\env).$ 
\end{defi}

\begin{defi}[Barb Preservation]\label{def:barb-preserving}
   A typed relation $\relR$ is barb preserving if and only if $$\env \vDash \sys{\sV}{P} \; \relR \; \sys{\sVV}{Q}\text{ implies } \left(\confE{\sys{\sV}{P}} \piBarb{c} \text{ iff }\confE{\sys{\sVV}{Q}} \piBarb{c}\right).$$ 
\end{defi}

Our behavioural preorder takes cost into consideration; it is defined in terms of families of amortised typed relations that are closed under costed reductions. 

\begin{defi}[Cost Improving]\label{def:cost-improving}
An amortized type-indexed relation $\relR$ is \emph{cost improving at credit} $n$
iff whenever $\env \vDash (\sys{\sV}{P}) \;\relR^n\; (\sys{\sVV}{Q})$ and 
  \begin{enumerate}
  \item if $\sys{\sV}{P} \piRedCost{k} \sys{\sV'}{P'}$ then $\sys{\sVV}{Q} \piRedCost{l}^\ast \sys{\sVV'}{Q'}$ such that $\env \vDash (\sys{\sV'}{P'}) \; \relR^{n+l-k} \; (\sys{\sVV'}{Q'})$;
  \item if $\sys{\sVV}{Q} \piRedCost{l} \sys{\sVV'}{Q'}$ then $\sys{\sV}{P} \piRedCost{k}^\ast \sys{\sV'}{P'}$  such that $\env \vDash (\sys{\sV'}{P'}) \; \relR^{n+l-k} \; (\sys{\sVV'}{Q'})$.
  \end{enumerate}
\end{defi}

Related processes must be related under arbitrary (parallel) contexts;
moreover, these contexts must be allowed to allocate new channels. 
We note that the second clause of our contextuality definition, Definition~\ref{def:contextuality}, is similar to that discussed earlier in Section~\ref{sec:bisim-results}, where we \emph{transfer} the respective permissions
held by the observer along with the test $R$  placed in parallel with the
processes. This is essential in order to preserve consistency (see \defref{def:consistency}) thus limiting our analysis to safe computations. Definition~\ref{def:contextuality} also requires an additional condition, when compared to the contextuality definition discussed in Section~\ref{sec:bisim-results}, namely that of \emph{resource extensions} where we consider systems in larger resource contexts (owned exclusively by the observer). This is described by the first clause in the definition; we recall the implicit condition for  resource environment representations from \secref{sec:language}, requiring the channel $c$ not to be present (thus allocated) in \sV\ (\resp \sVV) for the resource environment to be well-formed --- $c$ is therefore fresh.        In order to disambiguate between the different contextuality definitions, we refer to Definition~\ref{def:contextuality} as \emph{full contextuality}.

\begin{defi}[Full Contextuality]\label{def:contextuality}
An amortized type-indexed relation $\relR$  is contextual at environment $\env$ and credit $n$ iff whenever
\begin{math}
\env \vDash (\sysS{P}) \; \relR^n \; (\sys{\sVV}{Q}) 
\end{math}:

\begin{enumerate}
\item $\env,\emap{c}{\chantypU{\tVlst}} \vDash (\sys{\sV,c}{P}) \; \relR^n \; (\sys{\sVV,c}{Q}) $
\item If $\env\struct \env_1,\env_2$ where $\env_2 \vdash R$ then 
  \begin{itemize}
  \item $\env_1 \vDash (\sysS{P \piParal R}) \; \relR^n \;
    (\sys{\sVV}{Q \piParal R})$ and
  \item $\env_1 \vDash (\sysS{R \piParal
      P}) \; \relR^n \; (\sys{\sVV}{R \piParal Q})$
  \end{itemize}

\end{enumerate}
\end{defi}


We can now define the preorder defining our notion of observational system efficiency: 

\begin{defi}[Behavioral Contextual Preorder]\label{def:cost-preorder}
  $\piCost^{\env,n} $ is the largest family of amortized typed relations that is: 
  \begin{itemize}
  \item Barb Preserving;
  \item Cost Improving at credit $n$;
  \item Full contextual at environment $\env$.
  \end{itemize}
A system $\sysSP$ is said to be behaviourally as efficient as another system $\sys{\sVV}{Q}$ \wrt an observer \env, denoted as $\env \vDash \sysSP \piCost \sys{\sVV}{Q}$, whenever there exists an amortisation credit $n$ such that  $\env \vDash \sysSP \piCost^n \sys{\sVV}{Q}$.  Similarly, we can lift our preorder to processes: a process $P$ is said to be as efficient as $Q$ \wrt \sV\ and \env  whenever there exists an $n$ such that $\env \vDash \sysSP \piCost^n \sysS{Q}$
\end{defi}

\subsection{Soundness for \,\texorpdfstring{\piCostBis}{piCostBis}}
\label{sec:soundness}
Through Definition~\ref{def:cost-preorder} we are able to articulate why clients $\ptit{C}_2$ and $\ptit{C}'_2$ should be deemed to be behaviourally equally efficient \wrt $\env_1$ of \eqref{eq:env-sys-example:bis}: for an appropriate \sV, it turns out that we cannot differentiate between the two processes under any context allowed by \env.    Unfortunately, the universal quantification of contexts of Definition~\ref{def:contextuality} makes it hard to verify such a statement.    Through Theorem~\ref{thm:soundness} we can however establish that our bisimulation preorder of Definition~\ref{def:amortized-typed-bisim} provides a sound technique for determining behavioural efficiency.   This Theorem, in turn, relies on the lemmas we outline below.  In particular, Lemma~\ref{lem:-bisim-barb} and Lemma~\ref{lem:-bisim-reduc-closure}  prove that bisimulations are barb-preserving and cost-improving, whereas Lemma~\ref{lem:-bisim-weakening} proves that bisimulations are preserved under resource extensions. The required result then follows from Theorem~\ref{thm:costed-bisim-compositionality} of Section~\ref{sec:bisim-results}. 


\begin{lem}[Reductions and Bijective Renaming]\label{lem:reduction-bijective-rename}
  \begin{math}
    \text{For any bijective renaming }\sigma,\\ (\sys{\sV}{P})\sigma \piRed_k (\sys{\sV'}{P'})\sigma \text{ implies } \sys{\sV}{P} \piRed_k \sys{\sV}{P}
  \end{math}
\end{lem}

\begin{proof}
  By rule induction on $ (\sys{\sV}{P})\sigma \piRed_k (\sys{\sV'}{P'})\sigma$.
\end{proof}

\begin{lem}[Barb Preservation]\label{lem:-bisim-barb}
  \begin{displaymath}
  \env \vDash \sys{\sV}{P} \; \piCostBis \; \sys{\sVV}{Q} 
\text{ and } \confE{\sys{\sV}{P}} \piBarb{c} 
\text{ implies } \confE{\sys{\sVV}{Q}} \piBarb{c}
\end{displaymath}
\end{lem}

\begin{proof}
  By Definition~\ref{def:barb} we know 
  $\sys{\sV}{P} \piRedAst_l\piStruct (\sys{\sV'}{P'\piParal{\piOut{c}{\vec{d}}{P''}}})$ where $c\in\dom(\env)$. 
By Lemma~\ref{lem:reduc-and-transitions}(1) we obtain 
\begin{math}
  \confE{\sys{\sV}{P}} \piRedWDecCostPad{\quad}{l} \confE{\sys{\sV'}{P'''}}\;\text{where}\; 
  P''' \piStruct (P'\piParal{\piOut{c}{\vec{d}}{P''}}).
\end{math}
Moreover, by \rtit{lOut}, \rtit{lPar-R} and Lemma~\ref{lem:bisim-struct-equiv}
we deduce
\begin{math}
\confE{\sys{\sV}{P} } \piRedWDecCost{\actout{c}{\vec{d}}}{l} \piStruct\conf{\env'}{\sys{\sV}{P'\piParal P''}}
\end{math}. 
By   $\env \vDash \sys{\sV}{P} \; \piCostBis \; \sys{\sVV}{Q}$ 
we know that there exists a move
  $\confE{\sys{\sVV}{Q}} \piRedWDecCost{\actout{c}{\vec{d}}}{k} \conf{\env'}{\sys{\sVV'}{Q'}}$ 
and from this matching move, Lemma~\ref{lem:reduc-and-transitions}(2) (for the
initial $\tau$ moves of the weak action) and
Lemma~\ref{lem:transition-structure} we obtain 
  $(\sys{\sVV}{Q})\sigma_\env\piRedAst_{k_1}\piStruct (\sys{\sVV''}{Q''\piParal{\piOut{c}{\vec{d}}Q'''}})\sigma_\env$, 
which, together with $c \in \dom(\env)$ and
Lemma~\ref{lem:reduction-bijective-rename},  implies 
  $\sys{\sVV}{Q}\piRedAst_{k_1}\piStruct \sys{\sVV''}{Q''\piParal{\piOut{c}{\vec{d}}Q'''}}$ 
\ie $c$ is unaffected by the renaming $\sigma_\env$, and thus $\confE{\sys{\sVV}{Q}}
\piBarb{c}$.
\end{proof}

\begin{lem}[Cost Improving]\label{lem:-bisim-reduc-closure}
  \begin{math}
  \env \vDash \sys{\sV}{P} \; \piCostBisAmm{n} \; \sys{\sVV}{Q}
\text{ and } \sys{\sV}{P} \piRed_l \sys{\sV'}{P'}
  \end{math} 
 then there exist some  
\begin{math} \sys{\sVV'}{Q'} \text{ such that }
 \sys{\sVV}{Q} \piRedAst_k \sys{\sVV'}{Q'} 
\text{ and } \env \vDash \sys{\sV'}{P'} \; \piCostBisAmm{n+k-l} \; \sys{\sVV'}{Q'}
\end{math}
\end{lem}

\begin{proof}
  By $\sys{\sV}{P} \piRed_l \sys{\sV'}{P'}$ and Lemma~\ref{lem:reduc-and-transitions}(1) we know
  \begin{math}
    \confE{\sys{\sV}{P}} \piRedDecCost{\;\tau\;}{l} \confE{\sys{\sV}{P''}}\text{ where }P''\piStruct P'
  \end{math}.
  By \defref{def:amortized-typed-bisim} and assumption $\env \vDash \sys{\sV}{P} \; \piCostBisAmm{n} \; \sys{\sVV}{Q} $, this implies that
  \begin{math}
    {\confE{\sys{\sVV}{Q}}} \piRedWDecCost{\quad}{k} {\confE{\sys{\sV'}{Q'}}} 
  \end{math}  where
  \begin{equation}\label{eq:1:char}
   {\env \vDash \sys{\sV'}{P''}\piCostBisAmm{n+k-l} \sys{\sVV'}{Q'}}.
  \end{equation}  By Lemma~\ref{lem:reduc-and-transitions}(2) we deduce
  \begin{math}
    {(\sys{\sVV}{Q})\sigma_\env} \piRedAst_k {\sys{\sVV'}{Q'}}
  \end{math} and by Lemma~\ref{lem:reduction-bijective-rename} we obtain 
  \begin{math}
    {\sys{\sVV}{Q}} \piRedAst {\sys{\sVV''}{Q''}}
  \end{math} where $\sys{\sVV''}{Q''} = (\sys{\sVV'}{Q'})\sigma_\env$.
   The required result follows from $\env \vDash \sys{\sV'}{P'}\piCostBisAmm{0}\sys{\sV'}{P''}$, which we obtain from $P'\piStruct P''$ and Corollary~\ref{cor:Struct-eq-implies-bisim} (Structural Equivalence and Bisimilarity), \eqref{eq:1:char}, $\env \vDash \sys{\sVV''}{Q''} \piCostBisAmm{0} \sys{\sVV'}{Q'}$ which we obtain from Lemma~\ref{lem:reflexivity} (Reflexivity upto Renaming) and $\sys{\sVV''}{Q''} = (\sys{\sVV'}{Q'})\sigma_\env$, and Lemma~\ref{lem:transitivity}.
\end{proof}

\begin{lem}[Resource Extensions]\label{lem:-bisim-weakening}
  \begin{displaymath}
  \env \vDash \sys{\sV}{P} \; \piCostBisAmm{n} \; \sys{\sVV}{Q}  
\text{ implies }  \env,\emap{c}{\chantypU{\tVlst}} \vDash \sys{(\sV,c)}{P} \; \piCostBisAmm{n} \; \sys{(\sVV,c)}{Q}  
\end{displaymath}
\end{lem}

\begin{proof}
  By coinduction.
\end{proof}


\begin{thm}[Soundness]\label{thm:soundness}
  \begin{math}
  \env \vDash (\sys{\sV}{P}) \; \piCostBisAmm{n} \; (\sys{\sVV}{Q})
\text{ implies }    \env \vDash (\sys{\sV}{P}) \piCostAmm{n} (\sys{\sVV}{Q})  
\end{math}.
\end{thm}

\begin{proof}
  Follows from Lemma~\ref{lem:-bisim-barb} (Barb Preservation), Lemma~\ref{lem:-bisim-reduc-closure} (Cost Improving), Lemma~\ref{lem:-bisim-weakening} (Resource Extensions) and Theorem~\ref{thm:costed-bisim-compositionality} (Contextuality). 
\end{proof}

\begin{cor}[Soundness]\label{cor:soundness}
  \begin{math}
  \env \vDash (\sys{\sV}{P}) \; \piCostBis \; (\sys{\sVV}{Q})
\text{ implies }    \env \vDash (\sys{\sV}{P}) \piCost (\sys{\sVV}{Q})  
\end{math}.
\end{cor}

\subsection{Full Abstraction of \,\texorpdfstring{\piCost}{piCost}}
\label{sec:completeness}

To prove completeness, \ie that for every behavioural contextual preorder there exists a corresponding amortised typed-bisimulation,  we rely on the adapted notion of \emph{action definability}  \cite{Hennessy07,hennessy04behavioural}, which intuitively means that every action (label) used by our LTS can, in some sense,  be simulated (observed) by a specific test context.  For our specific case, two important aspects need to be taken into consideration:
\begin{itemize}
\item the \emph{typeability} of the testing context \wrt our substructural
  type system;
\item the \emph{cost} of the action simulation, which has to
  correspond to the cost of the action being observed.
\end{itemize}
These aspects are formalised in Definition~\ref{def:cost-definability}, which relies on the  functions definitions $\domm$ and $\codd$:
\begin{align*}
  \domm(\epsilon) &\deftxt \epsilon  & \codd(\epsilon) &\deftxt \epsilon\\
    \domm(\env,\envmap{c}{\tV}) &\deftxt \domm(\env),c  & \codd(\env,\envmap{c}{\tV}) &\deftxt \codd(\env),\tV
\end{align*}
These two meta-functions take a substructural type environment and returning respectively a \emph{list} of channel names and a \emph{list} of types.  For example, for the environment $\env=\envmap{c}{\chantypO{\tV}},\envmap{d}{\chantypW{\tV'}}, \envmap{c}{\chantypUU{\tV}{1}}$, we have $\domm(\env) = c,d,c$ and $\codd(\env)=\chantypO{\tV},\chantypW{\tV'},\chantypUU{\tV}{1}$.


Before stating cost-definability for actions, Definition~\ref{def:cost-definability}, we prove the technical Lemma~\ref{lem:actions-renaming} which allows us to express transitions in a convenient format for the respective definition without loss of generality.  

\begin{lem}[Transitions and Renaming]\label{lem:actions-renaming}
    $\conf{\env}{\sys{\sV}{P}} \piRedDecCostPad{\mu}{k}  \conf{\env'}{\sys{\sV'}{P'}}$ if and only if $\;\conf{\env}{\sys{\sV}{P}} \piRedDecCostPad{\mu}{k}\bigl(\conf{\env''}{\sys{\sV''}{P''}}\bigr)\sigma_\env$ for some $\sigma_\env,\env'',\sV'',P''$ where $\env'=\env''\sigma_\env$, $\sV'=\sV''\sigma_\env$ and $P'=P''\sigma_\env$.  
\end{lem}

\begin{proof}
  The \emph{if} case is immediate.  The proof for the \emph{only-if} is complicated by actions that perform channel allocation (see \rtit{lAll} and \rtit{lAllE} from \figref{fig:LTS}) because, in such cases, the renaming used in \rtit{lRen}'s premise cannot be used directly.  More precisely, from the premise we know:
\begin{equation*}
\begin{prooftree}
 \confE{\bigl(\sys{\sV}{P}\bigr)\sigma_\env} \;\piRedDecCostPre{\;\mu\;}{k}\; \conf{\env'}{\sys{\sV'}{P'}}
\justifiedBy{lRen}
    \confE{\sys{\sV}{P}} \;\piRedDecCost{\;\mu\;}{k}\; \conf{\env'}{\sys{\sV'}{P'}}
\end{prooftree}
\end{equation*}
and the required result follows if we prove the (slightly more cumbersome) sublemma:
\begin{lemm}[Transition and Renaming]
  $\confE{\bigl(\sys{\sV}{P}\bigr)\sigma_\env} \;\piRedDecCostPre{\;\mu\;}{k}\; \conf{\env'}{\sys{\sV'}{P'}}$ where $\fn(P) \subseteq \sV$ implies $\confE{\bigl(\sys{\sV}{P}\bigr)\sigma_\env} \;\piRedDecCostPre{\;\mu\;}{k}\; \bigl(\conf{\env''}{\sys{\sV''}{P''}}\bigr)\sigma'_\env$ for some $\sigma'_\env,\env'',\sV'',P''$ where
  \begin{itemize}
  \item $\env'=\env''\sigma'_\env$, $\sV'=\sV''\sigma'_\env$ and  $P'=P''\sigma'_\env$;
  \item $c\in\dom(\sV)$ implies $\sigma_\env(c) = \sigma'_\env(c)$ 
  \end{itemize}
\end{lemm}
The above sublemma is proved by rule induction on $\confE{\bigl(\sys{\sV}{P}\bigr)\sigma_\env} \;\piRedDecCostPre{\;\mu\;}{k}\; \conf{\env'}{\sys{\sV'}{P'}}$. We show one of the main cases:
\begin{description}
\item[\rtit{lAll}] We have $\confE{\bigl(\sys{\sV}{\piAll{x}{P}}\bigr)\sigma_\env} \piRedDecCostPre{\;\tau\;}{+1} \conf{\env}{\sys{\bigl((\sV)\sigma_\env,c\bigr)}{\bigl((P)\sigma_\env\subC{c}{x}}\bigr)}$. From the fact that $c\not\in (\sV\sigma_\env)$ --- it follows because $\bigl((\sV)\sigma_\env,c\bigr)$ is defined --- we know that $\sigma^{-1}_\env(c) \not\in \sV$. We thus choose some fresh channel $d$, \ie $d\not\in\bigl(\sV \cup (\sV\sigma_\env) \cup \dom(\env)\bigr)$\footnote{The condition that $d\not\in\dom(\env)$ is required since we do not state whether the triple \confESP\ is a configuration; otherwise, it is redundant --- see comments succeeding Definition~\ref{def:configuration}.}, and define $\sigma'_\env$ as $\sigma_\env$, except that it maps $d$ to $c$ and  also maps $\sigma^{-1}_\env(c)$ (\ie the channel name that mapped to $c$ in $\sigma_\env$) to $\sigma_\env(d)$, since this channel is not mapped to by $d$ anymore (in order to preserve bijectivity):
  \begin{equation*}
    \sigma'_\env(x) \deftxt
    \begin{cases}
     c & \text{if}\; x = d\\
     \sigma_\env(d)  & \text{if}\; x = \sigma^{-1}_\env(c)\\ 
      \sigma_\env(x) & \text{otherwise}
    \end{cases}
  \end{equation*}
We subsequently define
\begin{itemize}
\item $\env''$ as $\env$ since $\env\sigma'_\env = \env\sigma_\env=
  \env$;
\item $\sV''$ as $\sV,d$ since $(\sV,d)\sigma'_\env =
  \bigl((\sV)\sigma_\env,c\bigr)$; and
\item $P''$ as $P\subC{d}{x}$ since $P\subC{d}{x}\sigma'_\env = P\sigma_\env\subC{c}{x}$ \qedhere
\end{itemize}
\end{description}
\end{proof}

\begin{defi}[Cost Definable Actions]\label{def:cost-definability}
An action $\mu$ is cost-definable iff for any pair of type environments\footnote{Cost Definability cannot be defined \wrt the first environment only in the case of action \actall, since it non-deterministically allocates a fresh channel name and adds it to the residual environment - see \rtit{lAllE} in \figref{fig:LTS}.} $\env$ and  $\env'$, a corresponding substitution $\sigma_\env$, a set of channel names $C\in\Chans$,  and channel names $\ctit{succ}, \ctit{fail}\not\in C$, there exists a test $R$ such that ${\tproc{\env,\emap{\ctit{succ}}{\chantypO{\codd(\env')}}, \emap{\ctit{fail}}{\chantypO{}}, \emap{\ctit{fail}}{\chantypO{}}}{R}}$  and whenever $\sV\in C$:
\begin{enumerate}
\item \conf{\env}{\sys{\sV}{P}} \piRedDecCostPad{\mu}{k}  $\bigl(\conf{\env'}{\sys{\sV'}{P'}}\bigr)\sigma_\env$ \quad  implies \\$\sys{\sV,\ctit{succ}, \ctit{fail}}{P \piParal R} \piRedCostAst{k} \sys{\sV',\ctit{succ}, \ctit{fail}}{P'\piParal \piOutA{\ctit{succ}}{\bigl(\domm(\env')\bigr)}}$. 
\item $\sys{\sV,\ctit{succ}, \ctit{fail}}{P\piParal R} \piRedAst_k \sys{\sV''}{P''}$ where $\conf{\envmap{\ctit{succ}}{\chantypA{\codd(\env')}},\envmap{\ctit{fail}}{\chantypA{}}}{\sys{\sV''}{P''}} \piBarbNot{\ctit{fail}}$ and $\conf{\envmap{\ctit{succ}}{\chantypA{\codd(\env')}},\envmap{\ctit{fail}}{\chantypA{}}}{\sys{\sV''}{P''}} \piBarb{\ctit{succ}}$ implies \conf{\env}{\sys{\sV}{P}} \piRedWDecCostPad{\mu}{k}  $\bigl(\conf{\env'}{\sys{\sV'}{P'}}\bigr)\sigma_\env$ where $\sV''=\sV',\ctit{succ}, \ctit{fail}$ and $P'' \piStruct P'''\piParal \piOutA{\ctit{succ}}{\bigl(\domm(\env')\bigr)}$. 
 \end{enumerate}
\end{defi}

\begin{lem}[Action Cost-Definability]\label{lem:action-def} External actions $\mu \in \sset{\actout{c}{\vec{d}},\actin{c}{\vec{d}}, \actall, \actfree{c} \,|\, c,\vec{d} \subset \Chans}$  are cost-definable.
 \end{lem}

\begin{proof}
  The witness tests for  \actout{c}{\vec{d}} and \actin{c}{\vec{d}} are reasonably standard (see \cite{Hennessy07}), but need to take into account permission transfer. For instance, for the specific case of the action \actout{c}{d} where $d \not\in \domm(\env)$, if the transition $\conf{\env}{\sys{\sV}{P}} \piRedDecCostPad{\mu}{k} \bigl(\conf{\env'}{\sys{\sV'}{P'}\bigr)\sigma_\env}$ holds then we  know that, for some $\env_1$ and $\chantypA{\tV}$:     
  \begin{itemize}
  \item $\env = \env_1,\envmap{c}{\chantypA{\tV}}$;
  \item $\env'\sigma_\env = \env_1, \envmap{c}{\chantyp{\tV}{\aV-1}}, \envmap{d}{\tV}$ 
  \end{itemize}
In particular, when $\aV = \affine$ (affine),  using the permission to input on $c$ implicitly transfers the permission to process $P$ (see Section~\ref{sec:LTS}), potentially revoking the test's capability to perform name matching on channel name $c$  (see \rtit{tIf} in \figref{fig:typingrules}) --- this happens if $c\not\in\dom(\env_1)$.  For this reason, when $\aV = \affine$ the test is defined as 
\begin{displaymath}
      \piOutA{\ctit{fail}}{} \piParal \piIn{c}{x}{\piIf{\bigl(x\in \domm(\env_1)\bigr)}{\piNil}{\piIn{\ctit{fail}}{}{\piOutA{\ctit{succ}}{\bigl(\domm(\env')\bigr)}}}}
    \end{displaymath}
where $x  \in \domm(\env_1)$  is shorthand for a sequence of name comparisons as in \cite{Hennessy07}. Otherwise, the respective type assumption is not consumed from the observer environment and the test is defined as
\begin{displaymath}
      \piOutA{\ctit{fail}}{} \piParal \piIn{c}{x}{\piIf{\bigl(x\in \domm(\env)\bigr)}{\piNil}{\piIn{\ctit{fail}}{}{\piOutA{\ctit{succ}}{\bigl(\domm(\env')\bigr)}}}}
    \end{displaymath}
 Note that name comparisons on freshly acquired names are typeable since we also obtain the respective permissions upon input, \ie the explicit permission transfer (see Section~\ref{sec:LTS}).   The reader can verify that these tests typecheck \wrt the environment $\env,\emap{\ctit{succ}}{\chantypO{\codd(\env')}}, \emap{\ctit{fail}}{\chantypO{}}, \emap{\ctit{fail}}{\chantypO{}}$ and that they observe clauses $(1)$ and $(2)$ of Definition~\ref{def:cost-definability}. In the case of clause $(2)$, we note 
 that from the typing of the tests above, we know that $c\in\domm(\env)$ must hold (because both tests use channel $c$ for input); this is is a key requirement for the transition to fire --- see \rtit{lOut} of \figref{fig:LTS}. 

The witness tests for \actall\ and \actfree{c} involve less intricate permission transfer and are respectively defined as:
    \begin{displaymath}
      \piOutA{\ctit{fail}}{} \piParal \piAll{x}{\piIn{\ctit{fail}}{}{\piOutA{\ctit{succ}}{\bigl(\domm(\env),x\bigr)}}}
    \end{displaymath} and 
    \begin{displaymath}
       \piOutA{\ctit{fail}}{} \piParal \piFree{c}{\piIn{\ctit{fail}}{}{\piOutA{\ctit{succ}}{\bigl(\domm(\env')\bigr)}}}
    \end{displaymath}
We here focus on \actall\ and leave the analogous proof for \actfree{c} for the interested reader:
\begin{enumerate}
\item  If $\conf{\env}{\sys{\sV}{P}} \piRedDecCostPad{\actall}{k}  (\conf{\env'}{\sys{\sV'}{P'}})\sigma_\env$ we know that, for some $d\not\in\sV$ and $c\not\in\sV\sigma_\env$  where $\sigma_\env(d)=c$, we have  $(\env')\sigma_\env=(\env,\envmap{d}{\chantypU{\tV}})\sigma_\env  = \env,\envmap{c}{\chantypU{\tV}}$,  $\sV'=(\sV,d)$  and $P'=P$.  We can therefore simulate this action by the following sequence of reductions:
  \begin{align*}
    &\sys{\sV}{P\piParal \piOutA{\ctit{fail}}{} \piParal \piAll{x}{\piIn{\ctit{fail}}{}{\piOutA{\ctit{succ}}{\bigl(\domm(\env),x\bigr)}}}} \piRed\\
    &\quad \sys{\sV,d}{P\piParal \piOutA{\ctit{fail}}{} \piParal \piIn{\ctit{fail}}{}{\piOutA{\ctit{succ}}{\bigl(\domm(\env),d\bigr)}}} \piRed \sys{\sV,d}{P\piParal  \piOutA{\ctit{succ}}{\bigl(\domm(\env),d\bigr)}} 
  \end{align*}
\item From the structure of $R$ and the assumption that $\ctit{fail},\ctit{succ} \not\in\fn(P)$, we conclude that, if ${\conf{\envmap{\ctit{succ}}{\chantypA{\codd(\envv)}},\envmap{\ctit{fail}}{\chantypA{}}}{\sys{\sV'}{P'}} \piBarbNot{\ctit{fail}}}$ and $\conf{\envmap{\ctit{succ}}{\chantypA{\codd(\envv)}},\envmap{\ctit{fail}}{\chantypA{}}}{\sys{\sV'}{P'}} \piBarb{\ctit{succ}}$, then it must be the case that, for some $d\not\in\sV$, $P' = P'' \piParalL \piOutA{\ctit{succ}}{\bigl(\domm(\env),d\bigr)}$  where $\sV'' = (\sV',\ctit{succ},\ctit{fail},d)$ for some $\sV'$. 


 Since $P$ and $R$ do not share common channels there could not have been any interaction between the two processes in the reduction sequence $\sys{\sV,\ctit{succ}, \ctit{fail}}{P\piParal R} \piRedAst_k \sys{\sV'}{P'}$. Within this reduction sequence, from every reduction $\sys{\sV_i}{P_i\piParal R'} \piRed_{k_i} \sys{\sV_{i+1}}{P_{i+1}\piParal R'}$ resulting from derivatives of $P$, \ie $\sys{\sV_i}{P_i} \piRed_{k_i} \sys{\sV_{i+1}}{P_{i+1}}$ that happened before the allocation of channel $d$, we obtain a corresponding silent transition
 \begin{equation}
\conf{\env_i}{\sys{(\sV_i\setminus \sset{\ctit{succ},\ctit{fail}})}{P_{i}}} \piRedDecCostPad{\tau}{k_i}  \conf{\env_i}{\sys{(\sV_{i+1}\setminus \sset{\ctit{succ},\ctit{fail}})}{P_{i+1}}}\label{eq:definab:1}
\end{equation}
by Lemma~\ref{lem:reduc-and-transitions}(1) and an appropriate lemma that uses the fact ${\sset{\ctit{succ},\ctit{fail}} \cap \fn(P) = \emptyset}$ to allows us to shrink the allocated resources from $\sV_i$ to $(\sV_i\setminus \sset{\ctit{succ},\ctit{fail}})$.  A similar procedure can be carried out for reductions that happened after the allocation of $d$ as a result of reductions from $P$ derivatives, and by applying renaming $\sigma_\env$ we can obtain  
 \begin{equation}
\bigl(\conf{\env_i}{\sys{(\sV_i\setminus \sset{\ctit{succ},\ctit{fail}})}{P_{i}}}\bigr)\sigma_\env \piRedDecCostPad{\tau}{k_i}  \bigl(\conf{\env_i}{\sys{(\sV_{i+1}\setminus \sset{\ctit{succ},\ctit{fail}})}{P_{i+1}}\bigr)\sigma_\env}\label{eq:definab:2}
\end{equation}

The reduction
\begin{multline*}
\qquad\qquad\sys{\sV_i,\ctit{succ}, \ctit{fail}}{P_i\piParal \piAll{x}{\piIn{\ctit{fail}}{}{\piOutA{\ctit{succ}}{\bigl(\domm(\env),x\bigr)}}}} \piRed_{+1} \\
\sys{\sV_{i},\ctit{succ}, \ctit{fail},d}{P_{i}\piParal \piIn{\ctit{fail}}{}{\piOutA{\ctit{succ}}{\bigl(\domm(\env),d\bigr)}}}
\end{multline*}
can be substituted by the transition
\begin{equation}
\conf{\env_i}{\sys{\sV_i}{P_{i}}} \piRedWDecCostPad{\actall}{+1}  \conf{\env_i,\envmap{(d)\sigma_\env}{\chantypU{\tV}}}{\sys{\bigl((\sV_{i})\sigma_\env,(d)\sigma_\env\bigr)}{(P_{i})\sigma_\env}}\label{eq:definab:3}
\end{equation}
This follows from the fact that $d\not\in\sV_i$ and the fact that $\sigma_\env$ is a bijection, which implies that  $(d)\sigma_\env\not\in(\sV_i)\sigma_\env$ (necessary for $\bigl((\sV_{i})\sigma_\env,(d)\sigma_\env\bigr)$ to be a valid resource environment).  By joining together the transitions from \eqref{eq:definab:1},~\eqref{eq:definab:3} and \eqref{eq:definab:2} in the appropriate sequence we obtain the required weak transition. \qedhere
\end{enumerate}
\end{proof}

The proof of Theorem~\ref{thm:completeness} (Completeness) relies on Lemma~\ref{lem:action-def} to simulate a costed action by the appropriate test and is, for the most part, standard. As stated already, one novel aspect is that the cost semantics requires  the simulation  to incur the same cost as that of the costed action.  Through Reduction Closure, Lemma~\ref{lem:action-def} again, and then finally the Extrusion Lemma~\ref{lem:extrusion} we then obtain the matching bisimulation move which preserves the relative credit index.  Another novel aspect of the proof for Theorem~\ref{thm:completeness} is that the name matching in the presence of our substructural type environment requires a reformulation of the Extrusion Lemma.  More precisely, in the case of the output actions, the simulating test requires all of the environment permissions to perform all the necessary name comparisons.  We then make sure that these permissions are not lost by communicating them all again on \ctit{succ}; this passing on of permissions then allows us to show contextuality in Lemma~\ref{lem:extrusion}.

\begin{lem}[Extrusion]\label{lem:extrusion} Whenever $\confESP$ and $\conf{\env}{\sys{\sVV}{Q}}$ are configurations and $\vec{d} \not\in\dom(\env)$:
  \begin{displaymath}
    \emap{\ctit{succ}}{\chantyp{\codd(\env)}{\unique{1}}}
    \vDash \sys{\bigl(\sV,\ctit{succ},\vec{d}\bigr)}{P \piParal \piOutA{\ctit{succ}}{(\domm(\env))}} \piCostAmm{n} \sys{\bigl(\sVV,\ctit{succ},\vec{d}\bigr)}{Q \piParal \piOutA{\ctit{succ}}{(\domm(\env))}}
  \end{displaymath}   implies
  \begin{math}
    \env \vdash \sys{\sV}{P} \;\piCostAmm{n}\; \sys{\sVV}{Q}
  \end{math}
\end{lem}

\begin{proof}
  By coinduction we show that a family of  amortized typed relations \begin{math}
    {\env \vdash \sys{\sV}{P} \;\relR^n\; \sys{\sVV}{Q}}
  \end{math} observes the required properties of \defref{def:cost-preorder}.  Note that the environment $\emap{\ctit{succ}}{\chantyp{\codd(\env)}{\unique{1}}}$ ensures that $\ctit{succ}\not\in \names(P,Q)$ since both $P \piParal \piOutA{\ctit{succ}}{(\domm(\env))}$ and $Q \piParal \piOutA{\ctit{succ}}{(\domm(\env))}$ must typecheck \wrt a type environment that is consistent with $\emap{\ctit{succ}}{\chantyp{\codd(\env)}{\unique{1}}}$. Cost improving is straightforward and Barb Preserving and Contextuality follow standard techniques; see \cite{Hennessy07}.  

For instance, for barb preservation we are required to show that $\confESP \piBarb{c}$ implies  $\conf{\env}{\sys{\sVV}{Q}}\piBarb{c}$ (and viceversa).  From  $\confESP \piBarb{c}$ and Definition~\ref{def:barb} we know that $\envmap{c}{\chantypA{\vec{\tV}}}\in\env$ at some index $i$.  We can therefore define the process $R\deftri\piIn{\ctit{succ}}{\vec{x}}{\piIn{x_i}{\vec{y}}{\piOutA{\textit{ok}}{}}}$ where $|\vec{\tV}| = |\vec{y}|$; this test process typechecks \wrt $\emap{\ctit{succ}}{\chantyp{\codd(\env)}{\unique{1}}}, \emap{\textit{ok}}{\chantypO{}}$.  Now by Definition~\ref{def:contextuality}$(1)$ we know 
\begin{multline*}
    \emap{\ctit{succ}}{\chantyp{\codd(\env)}{\unique{1}}}, \emap{\textit{ok}}{\chantypU{}}
    \vDash
 \sys{\bigl(\sV,\ctit{succ},\vec{d},\textit{ok}\bigr)}{P \piParal \piOutA{\ctit{succ}}{(\domm(\env))}} \\\piCostAmm{n} \sys{\bigl(\sVV,\ctit{succ},\vec{d},\textit{ok}\bigr)}{Q \piParal \piOutA{\ctit{succ}}{(\domm(\env))}}
  \end{multline*}
and thus, by Definition~\ref{def:contextuality}$(2)$ and $\tproc{\emap{\ctit{succ}}{\chantyp{\codd(\env)}{\unique{1}}}, \emap{\textit{ok}}{\chantypO{}}}{R}$
\begin{equation}\label{eq:7:char}
  \begin{split}
    &\emap{\textit{ok}}{\chantypUU{}{1}}
    \vDash \sys{\bigl(\sV,\ctit{succ},\vec{d},\textit{ok}\bigr)}{P \piParal \piOutA{\ctit{succ}}{(\domm(\env))}\piParal R} \\
    &\qquad\qquad \qquad\qquad\qquad\qquad \qquad\qquad\piCostAmm{n}
    \sys{\bigl(\sVV,\ctit{succ},\vec{d},\textit{ok}\bigr)}{Q
      \piParal \piOutA{\ctit{succ}}{(\domm(\env))}\piParal R}
  \end{split}
  \end{equation}
Clearly, if $\confESP \piBarb{c}$ then $\bigl(\conf{\emap{\textit{ok}}{\chantypUU{}{1}}}{\sys{\bigl(\sV,\ctit{succ},\vec{d},\textit{ok}\bigr)}{(P \piParal \piOutA{\ctit{succ}}{(\domm(\env))}\piParal R)}}\bigr) \piBarb{\textit{ok}}$.  By \eqref{eq:7:char} and Definition~\ref{def:barb-preserving} we must have $\bigl(\conf{\emap{\textit{ok}}{\chantypUU{}{1}}}{\sys{\bigl(\sVV,\ctit{succ},\vec{d},\textit{ok}\bigr)}{(Q \piParal \piOutA{\ctit{succ}}{(\domm(\env))}\piParal R)}}\bigr) \piBarb{\textit{ok}}$ as well, which can only happen if $\sys{\sVV}{Q}\piRedAst\piStruct Q'\piParal \piOut{c}{\vec{d}}{Q''}$.  This means that $\confE{\sys{\sVV}{Q}} \piBarb{c}$.
\end{proof}


\begin{lem}\label{lem:environment-strenghten}
  $\env \vDash \sysSP \piCostAmm{n} \sys{\sVV}{Q}$ and $\env \struct \env'$ implies $\env' \vDash \sysSP \piCostAmm{n} \sys{\sVV}{Q}$
\end{lem}

\begin{proof}
  By coinduction.
\end{proof}

\begin{lem}\label{lem:rbc-renaming}
  $\env \vDash \sysSP \piCostAmm{n} \sys{\sVV}{Q}$ and $\sigma$ is a bijective renaming implies $\env\sigma \vDash \bigl(\sysSP\bigr)\sigma \piCostAmm{n} \bigl(\sys{\sVV}{Q}\bigr)\sigma$
\end{lem}

\begin{proof}
  By coinduction.
\end{proof}

\begin{thm}[Completeness]\label{thm:completeness}  
  \begin{math}
    \env\vDash (\sys{\sV}{P}) \piCostAmm{n} (\sys{\sVV}{Q})
  \end{math}
  implies 
  \begin{math}
    \env\vDash (\sys{\sV}{P}) \; \piCostBisAmm{n} \; (\sys{\sVV}{Q})
  \end{math}.
\end{thm}

\begin{proof}
  By coinduction, we show that for arbitrary $\env,n$, the family of relations included in \begin{math}
    \env\vDash \sys{\sV}{P} \piCostAmm{n} \sys{\sVV}{Q}
  \end{math}  observes the transfer properties of \defref{def:amortized-typed-bisim} at $\env,n$.  Assume
  \begin{equation}
\confE{\sysSP} \piRedDecCostPad{\mu}{k} \bigl(\conf{\env'}{\sys{\sV'}{P'}}\bigr)\sigma_\env\label{eq:6:char}
\end{equation}
If $\mu = \tau$, the matching move follows from Lemma~\ref{lem:reduc-and-transitions}, Definition~\ref{def:cost-improving} 
and Definition~\ref{def:cost-preorder}. 

 If $\mu \in \sset{\actout{c}{\vec{d}},\actin{c}{\vec{d}}, \actall, \actfree{c} \;|\; c,\vec{d} \in \Chans}$, by Lemma~\ref{lem:action-def} we know that there exists a test process that can simulate it; we choose one such test $R$ with channel names $\ctit{succ}, \ctit{fail} \not\in\sV,\sVV$.  By Definition~\ref{def:contextuality}(1) we know
\begin{equation*}
  \env, \envmap{\ctit{succ}}{\chantypU{\codd(\env)}}, \envmap{\ctit{fail}}{\chantypU{}}\vDash \sys{\sV,\ctit{succ},\ctit{fail}}{P} \piCostAmm{n} \sys{\sVV,\ctit{succ},\ctit{fail}}{Q}
\end{equation*}
and 
by Definition~\ref{def:contextuality}(2) and ${\tproc{\env,\emap{\ctit{succ}}{\chantypO{\codd(\envv)}}, \emap{\ctit{fail}}{\chantypO{}}, \emap{\ctit{fail}}{\chantypO{}}}{R}}$ 
(Definition~\ref{def:cost-definability}) we obtain
\begin{equation} \label{eq:2:char}
  \envmap{\ctit{succ}}{\chantypUU{\codd(\env)}{1}}, \envmap{\ctit{fail}}{\chantypUU{}{2}}\vDash \sys{(\sV,\ctit{succ},\ctit{fail})}{P\piParal R} \piCostAmm{n} \sys{(\sVV,\ctit{succ},\ctit{fail})}{Q\piParal R}
\end{equation}
From \eqref{eq:6:char} and Definition~\ref{def:cost-definability}$(1)$, we know
\begin{equation*}
   \sys{(\sV,\ctit{succ},\ctit{fail})}{P\piParal R} \piRedCostAst{k} \sys{(\sV',\ctit{succ},\ctit{fail})}{P' \piParal \piOutA{\ctit{succ}}{\domm(\env')}} 
\end{equation*}
By \eqref{eq:2:char} and Definition~\ref{def:cost-improving} (Cost Improving) we know
\begin{equation*}
   \sys{(\sVV,\ctit{succ},\ctit{fail})}{Q\piParal R} \piRedCostAst{l} \sys{\sVV''}{Q''}  
\end{equation*}
where
\begin{equation}\label{eq:5:char}
  \envmap{\ctit{succ}}{\chantypUU{\codd(\env)}{1}}, \envmap{\ctit{fail}}{\chantypUU{}{2}}\vDash \sys{(\sV',\ctit{succ},\ctit{fail})}{P' \piParal \piOutA{\ctit{succ}}{\domm(\env')}} \piCostAmm{n+l-k} \sys{\sVV''}{Q''}
\end{equation}
By Definition~\ref{def:barb-preserving} (Barb Preservation), this means that $\conf{\envmap{\ctit{succ}}{\chantypUU{\codd(\env)}{1}}, \envmap{\ctit{fail}}{\chantypUU{}{2}}}{\sys{\sVV'}{Q'}} \piBarbNot{\ctit{fail}}$ and also that $\conf{\envmap{\ctit{succ}}{\chantypUU{\codd(\env)}{1}}, \envmap{\ctit{fail}}{\chantypUU{}{2}}}{\sys{\sVV'}{Q'}} \piBarb{\ctit{succ}}$. By Definition~\ref{def:cost-definability}$(2)$ we obtain 
\begin{eqnarray}
  \label{eq:3:char}
    && Q''\piStruct Q'\piParal \piOutA{\ctit{succ}}{\domm(\env')} \text{ and } \sVV'' = (\sVV',\ctit{succ},\ctit{fail})  \\
  \label{eq:4:char}
    && \conf{\env}{\sys{\sVV}{Q}} \piRedWDecCostPad{\mu}{l} \bigl(\conf{\env'}{\sys{\sVV'}{Q'}}\bigr)\sigma_\env  
\end{eqnarray}
Transition \eqref{eq:4:char} is the matching move because by \eqref{eq:5:char} and Lemma~\ref{lem:environment-strenghten} we obtain
\begin{equation*}
    \envmap{\ctit{succ}}{\chantypUU{\codd(\env)}{1}}\vDash \sys{(\sV',\ctit{succ},\ctit{fail})}{P' \piParal \piOutA{\ctit{succ}}{\domm(\env')}} \piCostAmm{n+l-k} \sys{\sVV''}{Q''}
\end{equation*}
By \eqref{eq:3:char},  and Lemma~\ref{lem:extrusion} we obtain 
$\env' \vDash \sys{\sV'}{P'} \piCostAmm{n+l-k} \sys{\sVV'}{Q'}$ and subsequently by Lemma~\ref{lem:rbc-renaming} we obtain
\[
  \env'\sigma_\env \vDash \bigl(\sys{\sV'}{P'}\bigr)\sigma_\env \piCostAmm{n+l-k} \bigl(\sys{\sVV'}{Q'}\bigr)\sigma_\env 
\]
as required.\qedhere
\end{proof}



\section{Revisiting the Case Study}
\label{sec:proofs-relat-effic}

We can formally express that \pBufE is (strictly) more efficient than \pBuf in terms of the reduction semantics outlined in Section~\ref{sec:language} through the following statements:
\begin{eqnarray}
  \label{eq:11cs}
  \envExt \;\vDash\;& \sysS{\pBufE} \,\;\piCost\; \sysS{\pBuf}\\
 \label{eq:12cs}
    \envExt \;\vDash\;& \sysS{\pBuf}\; \not\!\piCost\; \sysS{\pBufE}
\end{eqnarray}

In order to show that the second statement \eqref{eq:12cs} holds, we need to prove that\emph{ there is no amortisation credit} $n$ for which   $\envExt \;\vDash\; \sysS{\pBuf}\; \piCost^{n}\; \sysS{\pBufE}$.  By choosing the set of inductively defined contexts $R_n$  where:\footnote{Note that \tproc{\envExt}{R_n} for any $n$.}
\begin{align*}
  R_0 & \deftri \piNil & R_{n+1} & \deftri \piOut{\ctit{in}}{v}
  {\piIn{\ctit{out}}{x}{R_{n}}} 
\end{align*}
we can argue by analysing the reduction graph of the respective systems that, for any $n\geq 0$:
$$\envExt \;\vDash\; \sysS{(\pBuf\piParal R_{n+1})}\; \not\piCost^{n}\; \sysS{(\pBufE\piParal R_{n+1})}$$
since it violates the Cost Improving property of Definition~\ref{def:cost-preorder}.

Another way how to prove \eqref{eq:12cs} is by exploiting completeness of our bisimulation proof technique \wrt our behavioural preorder, Theorem~\ref{thm:completeness}, and work at the level of the transition system of Section~\ref{sec:cost-bisim} showing that, for all $n \geq 0$,  the following holds: 
\begin{equation}
\envExt \;\vDash\; \sysS{\pBuf}\;\not\!\piCostBisAmm{n}\; \sysS{\pBufE} \label{eq:24:cs}
\end{equation}
We prove the above statement as Theorem~\ref{thm:strictly-ineffiscient} of Section~\ref{sec:prov-strict-ineff}.  

Property \eqref{eq:11cs}, \emph{prima facie}, seems even harder to prove than \eqref{eq:12cs}, because we are required to show that Barb Preservation and Cost Improving hold under every possible valid context interacting with the two buffer implementations.  Once again, we use the transition system of Section~\ref{sec:cost-bisim} and show instead that:
\begin{equation}
\envExt \;\vDash\; \sysS{\pBufE} \,\;\piCostBisAmm{0}\; \sysS{\pBuf}\label{eq:25:cs}
\end{equation}

The required result then follows from Theorem~\ref{thm:soundness}.  The proof for this statement is presented in Section~\ref{sec:prov-relat-effc}.

\medskip
 
In order to make the presentation of these proofs more manageable, we define the following macro definitions for sub-processes making up the derivatives of \conf{\envExt}{\sysS{\pBuf}} and \conf{\envExt}{\sysS{\pBufE}}.
\begin{align*}
  \small \pFf & \small \deftxt \piIn{b}{x}{ \piIn{\ctit{in}}{y}{\piAll{z}{ \bigl(\pF\piParal \piOutA{b}{z} \piParal \piOutA{x}{(y,z)}\bigr)}}} & \!\!\!\!\!\!\small \pBb & \small \deftxt  \piIn{d}{x}{ \piIn{x}{(y,z)}{ \piOut{\ctit{out}}{y}{ \bigl(\pB \piParal \piOutA{d}{z}\bigr)}}}\\
  \small \pFff{x} & \small\deftxt \piIn{\ctit{in}}{y}{\piAll{z}{ \bigl(\pF\piParal \piOutA{b}{z} \piParal \piOutA{x}{(y,z)}\bigr)}} & \!\!\!\!\!\!\small\pBbb{x} &\small \deftxt \piIn{x}{(y,z)}{ \piOut{\ctit{out}}{y}{ \bigl(\pB \piParal \piOutA{d}{z}\bigr)}}\\
  \small\pFfff{x}{y} & \small\deftxt  \piAll{z}{ \bigl(\pF\piParal \piOutA{b}{z} \piParal \piOutA{x}{(y,z)}\bigr)} & \!\!\!\!\!\!\small\pBbbb{y}{z} & \small\deftxt \piOut{\ctit{out}}{y}{ \bigl(\pB \piParal \piOutA{d}{z}\bigr)}\\[5pt]
   \small \pBEe & \small \deftxt \piIn{d}{x}{ \piIn{x}{(y,z)}{ \piFree{x}{\piOut{\ctit{out}}{y}{ \bigl(\pBE \piParal \piOutA{d}{z}\bigr)}}}} &
   \!\!\!\!\!\!\small \pBEee{x} & \small \deftxt \piIn{x}{(y,z)}{\piFree{x}{\piOut{\ctit{out}}{y}{ \bigl(\pBE \piParal \piOutA{d}{z}\bigr)}}}\\
\small  \pBEeee{x}{y}{z} & \small\deftxt \piFree{x}{\piOut{\ctit{out}}{y}{ \bigl(\pBE \piParal \piOutA{d}{z}\bigr)}} & 
 \!\!\!\!\!\!\small \pBEeeee{y}{z} &\small \deftxt \piOut{\ctit{out}}{y}{ \bigl(\pBE \piParal \piOutA{d}{z}\bigr)}
\end{align*}
We can thus express the definitions for \pBuf and \pBufE as:
\begin{align}\label{eq:24cs} 
  \pBuf & \deftxt \pFff{c_1} \piParal \pBbb{c_1} & \pBufE & \deftxt \pFff{c_1} \piParal \pBEee{c_1} 
\end{align}

\subsection{Proving Strict Inefficiency}
\label{sec:prov-strict-ineff}

In order to prove \eqref{eq:24:cs}, we do not need to explore the entire state space for \conf{\envExt}{\sysS{\pBuf}} and \conf{\envExt}{\sysS{\pBufE}}.  Instead, it suffices to limit external interactions with the observer to traces of the form $\bigl(\piRedWDec{\actin{\ctit{in}}{v}\,\cdot\,\actout{\ctit{out}}{v}}\bigr)^\ast$, which simulate interactions with the observing processes $R_n$ discussed in Section~\ref{sec:proofs-relat-effic}.  It is instructive to visualise the transition graphs for both \conf{\envExt}{\sys{\sV}{\pBuf}} and \conf{\envExt}{\sys{\sV}{\pBufE}} for a single iteration \piRedWDec{\actin{\ctit{in}}{v}\,\cdot\,\actout{\ctit{out}}{v}} as depicted in \figref{fig:TrnGrphBuff} and \figref{fig:TrnGrphBuff2}: due to lack of space, the nodes in these graphs  abstract away from the environment \envExt and appropriate resource environments $\sV, \sVV,\ldots$ containing internal channels $c_1, c_2, \ldots$  as required.\footnote{The transition graph also abstracts away from environment moves.}  For instance the first node of the graph in \figref{fig:TrnGrphBuff}, $\pFff{c_1} \piParalL \pBbb{c_1}$, \ie \pBuf, stands for $\conf{\envExt}{\sys{\sV}{(\pFff{c_1} \piParalL \pBbb{c_1})}}$,  where $c_1 \in \sV$, whereas the third node in the same graph, $\pF\piParalL \piOutA{b}{c_2} \piParalL \piOutA{c_1}{(v,c_2)}\piParalL \pBbb{c_1}$, stands for $\conf{\envExt}{\sys{\sVV}{\bigl(\pF\piParalL \piOutA{b}{c_2} \piParalL \piOutA{c_1}{(v,c_2)}\piParalL \pBbb{c_1}\bigr)}}$,  where $c_1,c_2\in\sVV$.

For instance, the graph in \figref{fig:TrnGrphBuff} shows that after the input action and the channel allocation for $c_2$ $\tau$-action (with a cost of $+1$) the inefficient buffer implementation reaches a state where it can perform a number of internal transitions: either the subcomponent \pF may take a recursion unfold step (the first right $\tau$-action) followed by an input on channel $b$ that instantiates the continuation with channel $c_2$ (the second right $\tau$-action), or else the subcomponent $\pBbb{c_1}$ reads from the head of the buffer $\piOutA{c_1}{(v,c_2)}$ (the first downwards $\tau$-action).    These $\tau$-actions may be interleaved, but no other silent transitions are possible until an output action is performed, after which the backend subcomponent can perform an unfold $\tau$-action  (the first downwards $\tau$-action following action $\actout{\ctit{out}}{v}$) followed by an instantiation communication on channel $d$ (the first downwards $\tau$-action following action $\actout{\ctit{out}}{v}$),  When all of these actions are completed we reach again the starting process, instantiated with channel $c_2$ instead.  The transitions in \figref{fig:TrnGrphBuff2} are analogous, but include a deallocation transition with a cost of $-1$. 

\begin{display}{Transition graph for \conf{\envExt}{\sys{\sV}{\pBuf}}  restricted to $\piRedWDecPad{\actin{\ctit{in}}{v}\,\cdot\,\actout{\ctit{out}}{v}}$}{fig:TrnGrphBuff}
\begin{tikzpicture}
  \node at (0,0) (1) {\scriptsize \pFff{c_1} \piParal \pBbb{c_1} = \pBuf};
  \node at (-5.5,0) (2) {\scriptsize \pFfff{c_1}{v}\piParal \pBbb{c_1}};
  \node at (-5.5,-1.5) (3) {\scriptsize\pF\piParal \piOutA{b}{c_2} \piParal \piOutA{c_1}{(v,c_2)}\piParal \pBbb{c_1}};
   \node at (0,-1.5) (4) {\scriptsize \pFf\piParal \piOutA{b}{c_2} \piParal \piOutA{c_1}{(v,c_2)}\piParal \pBbb{c_1}};
  \node at (5.5,-1.5) (5) {\scriptsize \pFff{c_2} \piParal \piOutA{c_1}{(v,c_2)}\piParal \pBbb{c_1}};
  \node at (-5.5,-3) (6) {\scriptsize\pF\piParal \piOutA{b}{c_2} \piParal  \pBbbb{v}{c_2}};
   \node at (0,-3) (7) {\scriptsize \pFf\piParal \piOutA{b}{c_2} \piParal \pBbbb{v}{c_2}};
  \node at (5.5,-3) (8) {\scriptsize \pFff{c_2} \piParal  \pBbbb{v}{c_2}};
   \node at (-5.5,-4.5) (9) {\scriptsize\pF\piParal \piOutA{b}{c_2} \piParal  \pB \piParal \piOutA{d}{c_2}};
   \node at (0,-4.5) (10) {\scriptsize \pFf\piParal \piOutA{b}{c_2} \piParal \pB \piParal \piOutA{d}{c_2}};
  \node at (5.5,-4.5) (11) {\scriptsize \pFff{c_2} \piParal  \pB \piParal \piOutA{d}{c_2}};
  \node at (-5.5,-6) (12) {\scriptsize\pF\piParal \piOutA{b}{c_2} \piParal  \pBb \piParal \piOutA{d}{c_2}};
   \node at (0,-6) (13) {\scriptsize \pFf\piParal \piOutA{b}{c_2} \piParal \pBb \piParal \piOutA{d}{c_2}};
  \node at (5.5,-6) (14) {\scriptsize \pFff{c_2} \piParal  \pBb \piParal \piOutA{d}{c_2}};
  \node at (-5.5,-7.5) (15) {\scriptsize\pF\piParal \piOutA{b}{c_2} \piParal  \pBbb{c_2}};
   \node at (0,-7.5) (16) {\scriptsize \pFf\piParal \piOutA{b}{c_2} \piParal \pBbb{c_2}};
  \node at (5.5,-7.5) (17) {\scriptsize \pFff{c_2} \piParal  \pBbb{c_2} };
  
  \draw[->] (1)  to node[above] {\scriptsize \actin{\ctit{in}}{v}} (2);
  \draw[->] (2)  to node[right] {\scriptsize \acttau} node[left] {\scriptsize $+1$} (3);
  \draw[->] (3)  to node[right] {\scriptsize \acttau}  (6); \draw[->] (3)  to node[above] {\scriptsize \acttau}  (4);
  \draw[->] (4)  to node[right] {\scriptsize \acttau}  (7); \draw[->] (4)  to node[above] {\scriptsize \acttau}  (5);
  \draw[->] (5)  to node[right] {\scriptsize \acttau}  (8);
  \draw[->] (6)  to node[right] {\scriptsize \actout{\ctit{out}}{v}}  (9); \draw[->] (6)  to node[above] {\scriptsize \acttau}  (7);
  \draw[->] (7)  to node[right] {\scriptsize \actout{\ctit{out}}{v}}  (10); \draw[->] (7)  to node[above] {\scriptsize \acttau}  (8);
  \draw[->] (8)  to node[right] {\scriptsize \actout{\ctit{out}}{v}}  (11);
  \draw[->] (9)  to node[right] {\scriptsize \acttau}  (12); \draw[->] (9)  to node[above] {\scriptsize \acttau}  (10);
  \draw[->] (10)  to node[right] {\scriptsize \acttau}  (13); \draw[->] (10)  to node[above] {\scriptsize \acttau}  (11);
  \draw[->] (11)  to node[right] {\scriptsize \acttau}  (14); 
 \draw[->] (12)  to node[right] {\scriptsize \acttau}  (15); \draw[->] (12)  to node[above] {\scriptsize \acttau}  (13);
  \draw[->] (13)  to node[right] {\scriptsize \acttau}  (16); \draw[->] (13)  to node[above] {\scriptsize \acttau}  (14);
  \draw[->] (14)  to node[right] {\scriptsize \acttau}  (17); 
 \draw[->] (15)  to node[above] {\scriptsize \acttau}  (16);
 \draw[->] (16)  to node[above] {\scriptsize \acttau}  (17);
\end{tikzpicture}
\end{display}

\begin{display}{Transition graph for \conf{\envExt}{\sys{\sV}{\pBufE}} restricted to $\piRedWDecPad{\actin{\ctit{in}}{v}\,\cdot\,\actout{\ctit{out}}{v}}$}{fig:TrnGrphBuff2}
\begin{tikzpicture}
  \node at (0,0) (1) {\scriptsize \pFff{c_1} \piParal \pBEee{c_1} = \pBufE};
  \node at (-5.5,0) (2) {\scriptsize \pFfff{c_1}{v}\piParal \pBEee{c_1}};
  \node at (-5.5,-1.5) (3) {\scriptsize\pF\piParal \piOutA{b}{c_2} \piParal \piOutA{c_1}{(v,c_2)}\piParal \pBEee{c_1}};
   \node at (0,-1.5) (4) {\scriptsize \pFf\piParal \piOutA{b}{c_2} \piParal \piOutA{c_1}{(v,c_2)}\piParal \pBEee{c_1}};
  \node at (5.5,-1.5) (5) {\scriptsize \pFff{c_2} \piParal \piOutA{c_1}{(v,c_2)}\piParal \pBEee{c_1}};
  \node at (-5.5,-3) (6) {\scriptsize\pF\piParal \piOutA{b}{c_2} \piParal  \pBEeee{c_1}{v}{c_2}};
   \node at (0,-3) (7) {\scriptsize \pFf\piParal \piOutA{b}{c_2} \piParal \pBEeee{c_1}{v}{c_2}};
  \node at (5.5,-3) (8) {\scriptsize \pFff{c_2} \piParal  \pBEeee{c_1}{v}{c_2}};
    \node at (-5.5,-4.5) (9) {\scriptsize\pF\piParal \piOutA{b}{c_2} \piParal  \pBEeeee{v}{c_2}};
   \node at (0,-4.5) (10) {\scriptsize \pFf\piParal \piOutA{b}{c_2} \piParal \pBEeeee{v}{c_2}};
  \node at (5.5,-4.5) (11) {\scriptsize \pFff{c_2} \piParal  \pBEeeee{v}{c_2}};
   \node at (-5.5,-6) (12) {\scriptsize\pF\piParal \piOutA{b}{c_2} \piParal  \pBE \piParal \piOutA{d}{c_2}};
   \node at (0,-6) (13) {\scriptsize \pFf\piParal \piOutA{b}{c_2} \piParal \pBE \piParal \piOutA{d}{c_2}};
  \node at (5.5,-6) (14) {\scriptsize \pFff{c_2} \piParal  \pBE \piParal \piOutA{d}{c_2}};
  \node at (-5.5,-7.5) (15) {\scriptsize\pF\piParal \piOutA{b}{c_2} \piParal  \pBEe \piParal \piOutA{d}{c_2}};
   \node at (0,-7.5) (16) {\scriptsize \pFf\piParal \piOutA{b}{c_2} \piParal \pBEe \piParal \piOutA{d}{c_2}};
  \node at (5.5,-7.5) (17) {\scriptsize \pFff{c_2} \piParal  \pBEe \piParal \piOutA{d}{c_2}};
  \node at (-5.5,-9) (18) {\scriptsize\pF\piParal \piOutA{b}{c_2} \piParal  \pBEee{c_2}};
   \node at (0,-9) (19) {\scriptsize \pFf\piParal \piOutA{b}{c_2} \piParal \pBEee{c_2}};
  \node at (5.5,-9) (20) {\scriptsize \pFff{c_2} \piParal  \pBEee{c_2} };
  
  \draw[->] (1)  to node[above] {\scriptsize \actin{\ctit{in}}{v}} (2);
  \draw[->] (2)  to node[right] {\scriptsize \acttau} node[left] {\scriptsize $+1$} (3);
  \draw[->] (3)  to node[right] {\scriptsize \acttau}  (6); \draw[->] (3)  to node[above] {\scriptsize \acttau}  (4);
  \draw[->] (4)  to node[right] {\scriptsize \acttau}  (7); \draw[->] (4)  to node[above] {\scriptsize \acttau}  (5);
  \draw[->] (5)  to node[right] {\scriptsize \acttau}  (8);
 \draw[->] (6)  to node[right] {\scriptsize \acttau} node[left] {\scriptsize $-1$} (9); \draw[->] (6)  to node[above] {\scriptsize \acttau}  (7);
  \draw[->] (7)  to node[right] {\scriptsize \acttau} node[left] {\scriptsize $-1$} (10); \draw[->] (7)  to node[above] {\scriptsize \acttau}  (8);
  \draw[->] (8)  to node[right] {\scriptsize \acttau} node[left] {\scriptsize $-1$} (11);
  \draw[->] (9)  to node[right] {\scriptsize \actout{\ctit{out}}{v}}  (12); \draw[->] (9)  to node[above] {\scriptsize \acttau}  (10);
  \draw[->] (10)  to node[right] {\scriptsize \actout{\ctit{out}}{v}}  (13); \draw[->] (10)  to node[above] {\scriptsize \acttau}  (11);
  \draw[->] (11)  to node[right] {\scriptsize \actout{\ctit{out}}{v}}  (14);
  \draw[->] (12)  to node[right] {\scriptsize \acttau}  (15); \draw[->] (12)  to node[above] {\scriptsize \acttau}  (13);
  \draw[->] (13)  to node[right] {\scriptsize \acttau}  (16); \draw[->] (13)  to node[above] {\scriptsize \acttau}  (14);
  \draw[->] (14)  to node[right] {\scriptsize \acttau}  (17); 
 \draw[->] (15)  to node[right] {\scriptsize \acttau}  (18); \draw[->] (15)  to node[above] {\scriptsize \acttau}  (16);
  \draw[->] (16)  to node[right] {\scriptsize \acttau}  (19); \draw[->] (16)  to node[above] {\scriptsize \acttau}  (17);
  \draw[->] (17)  to node[right] {\scriptsize \acttau}  (20); 
 \draw[->] (18)  to node[above] {\scriptsize \acttau}  (19);
 \draw[->] (19)  to node[above] {\scriptsize \acttau}  (20);
\end{tikzpicture}

\end{display}

\newcommand{\psA}{\ensuremath{\text{Prc}_\text{\rm A}}\xspace}
\newcommand{\psB}{\ensuremath{\text{Prc}_\text{\rm B}}\xspace}
\newcommand{\psC}{\ensuremath{\text{Prc}_\text{\rm C}}\xspace}

Theorem~\ref{thm:strictly-ineffiscient}, which proves \eqref{eq:24:cs}, relies on two lemmas.  The main one is Lemma~\ref{lem:negative-induction}, which establishes that a number of derivatives from the configurations \conf{\envExt}{\sysS{\pBuf}} and \conf{\envExt}{\sysS{\pBufE}} cannot be related for \emph{any} amortisation credit.  This Lemma, in turn, relies on Lemma~\ref{lem:Neg:implications}, which establishes that, for a particular amortisation credit $n$, if some  pair of derivatives of the configurations \conf{\envExt}{\sysS{\pBuf}} and \conf{\envExt}{\sysS{\pBufE}} \resp cannot be related, then other pairs of derivatives cannot be related either.  Lemma~\ref{lem:Neg:implications} is used again by Theorem~\ref{thm:strictly-ineffiscient} to derive that, from the unrelated pairs identified by Lemma~\ref{lem:negative-induction}, the required pair of configurations \conf{\envExt}{\sysS{\pBuf}} and \conf{\envExt}{\sysS{\pBufE}} cannot be related for any amortisation credit.  Upon first reading, the reader who is only interested in the eventual result may safely skip to the statement of Theorem~\ref{thm:strictly-ineffiscient} and treat Lemma~\ref{lem:negative-induction} and Lemma~\ref{lem:Neg:implications} as black-boxes.


In order to be able to state Lemma~\ref{lem:Neg:implications} and Lemma~\ref{lem:negative-induction} more succinctly, we find it convenient to delineate groups of processes relating to derivatives of \pBuf and \pBufE.  For instance, we can partition the processes depicted in the transition graph of \figref{fig:TrnGrphBuff2} (derivatives of \pBufE) into three sets:
\begin{align*}
\psA & \deftxt \sset{
     \begin{array}{l|l}
\bigl(\pF\piParalL \piOutA{b}{c_2} \piParalL \piOutA{c_1}{(v,c_2)}\piParalL \pBEee{c_1}\bigr),\\
  \bigl(\pFf\piParalL \piOutA{b}{c_2} \piParalL \piOutA{c_1}{(v,c_2)}\piParalL \pBEee{c_1}\bigr),& c_1 \neq c_2 \in  \\
  \bigl(\pFff{c_2} \piParalL \piOutA{c_1}{(v,c_2)}\piParalL \pBEee{c_1}\bigr),\,
  \bigl(\pF\piParalL \piOutA{b}{c_2} \piParalL  \pBEeee{c_1}{v}{c_2}\bigr), &  \Chans \setminus\sset{\ctit{in},\ctit{out},b,d}\\
   \bigl(\pFf\piParal \piOutA{b}{c_2} \piParal \pBEeee{c_1}{v}{c_2}\bigr),\,
  \bigl(\pFff{c_2} \piParal  \pBEeee{c_1}{v}{c_2}\bigr)
\end{array}\!\!}\\[5pt]
\psB & \deftxt  \sset{
  \begin{array}{l|l}
\bigl(\pF\piParalL \piOutA{b}{c_2} \piParalL  \pBEeeee{v}{c_2}\bigr),\,\bigl(\pFf\piParalL \piOutA{b}{c_2} \piParalL \pBEeeee{v}{c_2}\bigr),& c_2 \in\Chans \setminus\sset{\ctit{in},\ctit{out},b,d}\\
\bigl(\pFff{c_2} \piParalL  \pBEeeee{v}{c_2}\bigr)
\end{array}
}\\[5pt]
\psC & \deftxt   \sset{
  \begin{array}{l|l}
     \bigl(\pF\piParal \piOutA{b}{c_2} \piParal  \pBE \piParal \piOutA{d}{c_2}\bigr),\,
     \bigl(\pFf\piParal \piOutA{b}{c_2} \piParal \pBE \piParal \piOutA{d}{c_2}\bigr),\\
  \bigl(\pFff{c_2} \piParal  \pBE \piParal \piOutA{d}{c_2}\bigr),\,
  \bigl(\pF\piParal \piOutA{b}{c_2} \piParal  \pBEe \piParal \piOutA{d}{c_2}\bigr),\\
   \bigl(\pFf\piParal \piOutA{b}{c_2} \piParal \pBEe \piParal \piOutA{d}{c_2}\bigr),\,
  \bigl(\pFff{c_2} \piParal  \pBEe \piParal \piOutA{d}{c_2}\bigr),& c_2 \in \Chans \setminus\sset{\ctit{in},\ctit{out},b,d}\\
  \bigl(\pF\piParal \piOutA{b}{c_2} \piParal  \pBEee{c_2}\bigr),\,
   \bigl(\pFf\piParal \piOutA{b}{c_2} \piParal \pBEee{c_2}\bigr),\\
   \bigl(\pFff{c_2} \piParal  \pBEee{c_2}\bigr)
  \end{array}
}
\end{align*}
With respect to the transition graph of \figref{fig:TrnGrphBuff2}, \psA groups the processes \emph{after the allocation} of an (arbitrary) internal channel $c_2$ but not before any deallocation, \ie the second and third rows of the graph. The set \psB groups the processes \emph{after the deallocation} of the (arbitrary) internal channel $c_1$, \ie the fourth row of the graph. Finally, the set \psC groups processes \emph{after the output action} \actout{\ctit{out}}{v} is performed (before an input action is performed), \ie the last three rows of the graph.

\begin{lem}[Related Negative Results]\label{lem:Neg:implications} \quad
  \begin{enumerate}
  \item For any amortisation credit $n$ and appropriate $\sV,\sVV$,  whenever:
    \begin{itemize}
    \item   $\envExt \vDash \sys{\sV}{\pFfff{c_1}{v}\piParalL
        \pBbb{c_1}} \; \not\!\piCostBisAmm{n}\;
      \sys{\sVV}{\pFfff{c'_1}{v}\piParalL \pBEee{c'_1}}$
    \item For any $Q \in \psA$ we have   $\envExt \vDash \sys{\sV}{\pFfff{c_1}{v}\piParalL
        \pBbb{c_1}} \; \not\!\piCostBisAmm{n+1}\;
      \sys{\sVV}{Q}$\; 
    \item For any $Q \in \psB$ we have  $\envExt \vDash \sys{\sV}{\pFfff{c_1}{v}\piParalL
        \pBbb{c_1}} \; \not\!\piCostBisAmm{n}\;
      \sys{\sVV}{Q}$
    \end{itemize}
then, for any $P \in \psC$, we have\;
$\envExt \vDash \sys{\sV}{\pFff{c_1} \piParalL \pBbb{c_1}} \; \not\!\piCostBisAmm{n}\; \sys{\sVV}{P}$.
\item For any amortisation credit $n$ and appropriate $\sV,\sVV$,  and  for any $Q \in \psC$: 
  \begin{enumerate}
  \item $\envExt \vDash \sys{\sV}{\pFff{c_1} \piParalL \pBbb{c_1}} \;
    \not\!\piCostBisAmm{n}\;
    \sys{\sVV}{Q}$  implies \\
    \qquad for any $P \in \psC$\;  $\envExt \vDash \sys{\sV}{\pFf \piParalL \piOutA{b}{c_1}
      \piParalL \pBbb{c_1}} \; \not\!\piCostBisAmm{n}\; \sys{\sVV}{P}$
  \item 
    $\envExt \vDash \sys{\sV}{\pFf \piParalL  \piOutA{b}{c_1} \piParalL \pBbb{c_1}} \; \not\!\piCostBisAmm{n}\;
      \sys{\sVV}{Q}$  implies \\
    \hspace{1cm}for any $P \in \psC$\; $\envExt \vDash \sys{\sV}{\pF \piParalL  \piOutA{b}{c_1} \piParalL \pBbb{c_1}} \; \not\!\piCostBisAmm{n}\;
      \sys{\sVV}{P}$  
  \item 
    $\envExt \vDash \sys{\sV}{\pF \piParalL  \piOutA{b}{c_1} \piParalL \pBbb{c_1}} \; \not\!\piCostBisAmm{n}\;
      \sys{\sVV}{Q}$  implies \\
     for any $P \in \psC$\; $\envExt \vDash \sys{\sV}{\pF \piParalL  \piOutA{b}{c_1} \piParalL \pBb \piParalL \piOutA{d}{c_1}} \; \not\!\piCostBisAmm{n}\;
      \sys{\sVV}{P}$   
   \item 
    $\envExt \vDash \sys{\sV}{\pF \piParalL  \piOutA{b}{c_1} \piParalL \pBb \piParalL \piOutA{d}{c_1}} \; \not\!\piCostBisAmm{n}\;
      \sys{\sVV}{Q}$  implies \\
     for any $P \in \psC$\; $\envExt \vDash \sys{\sV}{\pF \piParalL  \piOutA{b}{c_1} \piParalL \pB \piParalL \piOutA{d}{c_1}} \; \not\!\piCostBisAmm{n}\;
      \sys{\sVV}{P}$  
  \end{enumerate}
  \item For any amortisation credit $n$ and appropriate $\sV,\sVV$,  and  for any $R\in \psB$, $Q \in \psC$:
  \begin{enumerate}
    \item 
    $\envExt \vDash \sys{\sV}{\pF \piParalL  \piOutA{b}{c_2} \piParalL \pB \piParalL \piOutA{d}{c_2}} \; \not\!\piCostBisAmm{n}\;
      \sys{\sVV}{Q}$  implies \\
     for any $P \in \psB$\; $\envExt \vDash \sys{\sV}{\pF \piParalL  \piOutA{b}{c_2} \piParalL \pBbbb{v}{c_2}} \; \not\!\piCostBisAmm{n}\;
      \sys{\sVV}{P}$  
    \item 
    $\envExt \vDash \sys{\sV}{\pF \piParalL  \piOutA{b}{c_2} \piParalL \pBbbb{v}{c_2}} \; \not\!\piCostBisAmm{n}\;
      \sys{\sVV}{R}$
       implies \\
     for any $P \in \psB$\;$\envExt \vDash \sys{\sV}{\pF \piParalL  \piOutA{b}{c_2} \piParalL \piOutA{c_1}{(v,c_2)}\piParalL \pBbb{c_1}} \; \not\!\piCostBisAmm{n}\;
      \sys{\sVV}{P}$  
  \end{enumerate}
  \item  For any amortisation credit $n$ and appropriate $\sV,\sVV$,  and  for any  $Q \in \psC$:
    \begin{enumerate}
    \item $\envExt \vDash \sys{\sV}{\pF \piParalL  \piOutA{b}{c_1} \piParalL \pB \piParalL \piOutA{d}{c_1}} \; \not\!\piCostBisAmm{n}\;
      \sys{\sVV}{Q}$  implies \\
      for any $P\in \psA$\; $\envExt \vDash \sys{\sV}{\pF \piParalL  \piOutA{b}{c_1} \piParalL \pBbbb{v}{c_1}} \; \not\!\piCostBisAmm{n+1}\;
      \sys{\sVV}{P}$
     \item $\envExt \vDash \sys{\sV}{\pF \piParalL  \piOutA{b}{c_1} \piParalL \pB \piParalL \piOutA{d}{c_1}} \; \not\!\piCostBisAmm{n}\;
      \sys{\sVV}{Q}$  implies \\
      $\envExt \vDash \sys{\sV}{\pF \piParalL  \piOutA{b}{c_1} \piParalL \pBbbb{v}{c_1}} \; \not\!\piCostBisAmm{n}\;
      \sys{\sVV}{\pFfff{c'_1}{v}\piParal \pBEee{c'_1}}$
     \item $\envExt \vDash \sys{\sV}{\pF \piParalL  \piOutA{b}{c_1} \piParalL \pB \piParalL \piOutA{d}{c_1}} \; \not\!\piCostBisAmm{n}\;
      \sys{\sVV}{Q}$  implies \\
      $\envExt \vDash \sys{\sV}{\pF \piParalL  \piOutA{b}{c_2} \piParalL \piOutA{c_1}{(v,c_2)}\piParalL \pBbb{c_1}} \; \not\!\piCostBisAmm{n}\;
      \sys{\sVV}{\pFfff{c'_1}{v}\piParal \pBEee{c'_1}}$
    \end{enumerate}

  \end{enumerate}
\end{lem}

\begin{proof} Each case is proved by contradiction:
  \begin{enumerate}
  \item Assume the premises together with the inverse of the conclusion, \ie $$\envExt \vDash \sys{\sV}{\pFff{c_1} \piParalL \pBbb{c_1}} \; \piCostBisAmm{n}\; \sys{\sVV}{P}.$$ Consider the transition from the left-hand configuration:
    $$\conf{\envExt}{\sys{\sV}{\pFff{c_1} \piParalL \pBbb{c_1}}} \piRedDecCostPad{\actin{\ctit{in}}{v}}{0}  \conf{\envExt}{\sys{\sV}{\pFfff{c_1}{v}\piParalL
        \pBbb{c_1}}}.$$  
   For any $P \in \psC$, this can only be matched by the right-hand configuration, $\conf{\envExt}{\sys{\sVV}{P}}$, through  either of the following cases:
    \begin{enumerate}
    \item $\conf{\envExt}{\sys{\sVV}{P}} \piRedWDecCostPad{\actin{\ctit{in}}{v}}{0} \conf{\envExt}{\sys{\sVV}{\pFfff{c'_1}{v}\piParalL \pBEee{c'_1}}}$, \ie a weak input action without trailing $\tau$-moves after the external action $\actin{\ctit{in}}{v}$ --- see first row of the  graph in \figref{fig:TrnGrphBuff2}.  But we know  $\envExt \vDash \sys{\sV}{\pFfff{c_1}{v}\piParalL
        \pBbb{c_1}} \; \not\!\piCostBisAmm{n}\;
      \sys{\sVV}{\pFfff{c'_1}{v}\piParalL \pBEee{c'_1}}$ from the first premise. 
    \item $\conf{\envExt}{\sys{\sVV}{P}} \piRedWDecCostPad{\actin{\ctit{in}}{v}}{+1} \conf{\envExt}{\sys{\sVV}{Q}}$ for some $Q \in \psA$.  However, from the second premise we know that $\envExt \vDash \sys{\sV}{\pFfff{c_1}{v}\piParalL
        \pBbb{c_1}} \; \not\!\piCostBisAmm{n+1}\;
      \sys{\sVV}{Q}$
    \item $\conf{\envExt}{\sys{\sVV}{P}} \piRedWDecCostPad{\actin{\ctit{in}}{v}}{0} \conf{\envExt}{\sys{\sVV}{Q}}$ for some $Q \in \psB$. Again, from the third premise we know that $\envExt \vDash \sys{\sV}{\pFfff{c_1}{v}\piParalL
        \pBbb{c_1}} \; \not\!\piCostBisAmm{n}\;
      \sys{\sVV}{Q}$
    \end{enumerate}
    Since \conf{\envExt}{\sys{\sVV}{P}} cannot perform a matching move, we obtain a contradiction.
  \item We here prove case $(a)$. The other cases are analogous. 

  Assume  $\envExt \vDash \sys{\sV}{\pFf \piParalL  \piOutA{b}{c_1} \piParalL \pBbb{c_1}} \; \piCostBisAmm{n}\;
      \sys{\sVV}{P}$ and consider the action
      \begin{equation*}
       \conf{\envExt}{\sys{\sV}{\pFf \piParalL  \piOutA{b}{c_1} \piParalL \pBbb{c_1}}} \piRedDecCostPad{\acttau}{0}  \conf{\envExt}{\sys{\sV}{\pFff{c_1} \piParalL \pBbb{c_1}}}.
     \end{equation*}
     For our assumption to hold,  $\conf{\envExt}{\sys{\sVV}{P}}$ would need to match this move 
     by a (weak) silent action leading to a configuration that can match  $\conf{\envExt}{\sys{\sV}{\pFff{c_1} \piParalL \pBbb{c_1}}}$. 
     The only matching move can be
     \begin{equation*}
       \conf{\envExt}{\sys{\sVV}{P}}  \piRedWDecCostPad{\phantom{\acttau}}{0}  \conf{\envExt}{\sys{\sVV}{Q}} \qquad\text{for some } Q \in \psC.
     \end{equation*}
       However, from our premise we know $\envExt \vDash \sys{\sV}{\pFff{c_1} \piParalL \pBbb{c_1}} \; \not\!\piCostBisAmm{n}\;
      \sys{\sVV}{Q'}$ for any amortisation credit $n$   and  $Q' \in \psC$ and therefore conclude that the move cannot be matched, thereby obtaining a contradiction.
    \item We here prove case $(a)$. Case $(b)$ is analogous. 

  Assume  $\envExt \vDash \sys{\sV}{\pF \piParalL  \piOutA{b}{c_2} \piParalL \pBbbb{v}{c_2}} \; \piCostBisAmm{n}\;
      \sys{\sVV}{P}$ and consider the action 
      \begin{equation*}
        \conf{\envExt}{\sys{\sV}{\pF \piParalL  \piOutA{b}{c_2} \piParalL \pBbbb{v}{c_2}}} \piRedDecCostPad{\actout{\ctit{out}}{v}}{0} \conf{\envExt}{\sys{\sV}{\pF \piParalL  \piOutA{b}{c_2} \piParalL \pB \piParalL \piOutA{d}{c_2}}}
      \end{equation*}
     This action can only be matched by a transition of the form
     \begin{equation*}
        \conf{\envExt}{\sys{\sVV}{P}}  \piRedWDecCostPad{\actout{\ctit{out}}{v}}{0}  \conf{\envExt}{\sys{\sVV}{Q}}\qquad\text{for some } Q \in \psC.
     \end{equation*}
      However, from our premise we know $\envExt \vDash \sys{\sV}{\pF \piParalL  \piOutA{b}{c_2} \piParalL \pB \piParalL \piOutA{d}{c_2}} \; \not\!\piCostBisAmm{n}\;
      \sys{\sVV}{Q}$ for any  amortisation credit $n$   and  $Q\in \psC$. Thus we conclude that the move cannot be matched, thereby obtaining a contradiction.
    \item Cases $(a)$ and $(b)$ are analogous to $3(a)$ and $3(b)$.  We here outline the proof for case $(c)$.

     First we note that from the premise $\envExt \vDash \sys{\sV}{\pF \piParalL  \piOutA{b}{c_1} \piParalL \pB \piParalL \piOutA{d}{c_1}} \; \not\!\piCostBisAmm{n}\;
      \sys{\sVV}{Q}$ (for any $Q \in \psC$) and Lemma~$\ref{lem:Neg:implications}.3(a)$, Lemma~$\ref{lem:Neg:implications}.4(a)$ and Lemma~$\ref{lem:Neg:implications}.4(b)$ \resp we obtain:
      \begin{align}\label{eq:13cs}
         & \envExt \vDash \sys{\sV}{\pF \piParalL  \piOutA{b}{c_2} \piParalL \pBbbb{v}{c_2}} \; \not\!\piCostBisAmm{n}\;
      \sys{\sVV}{P}  \qquad\text{for any } P\in \psB \\
         \label{eq:14cs}
         & \envExt \vDash \sys{\sV}{\pF \piParalL  \piOutA{b}{c_1} \piParalL \pBbbb{v}{c_1}} \; \not\!\piCostBisAmm{n+1}\;
      \sys{\sVV}{P} \qquad\text{for any }  P\in \psA\\
      \label{eq:15cs}
        & \envExt \vDash \sys{\sV}{\pF \piParalL  \piOutA{b}{c_1} \piParalL \pBbbb{v}{c_1}} \; \not\!\piCostBisAmm{n}\;
      \sys{\sVV}{\pFfff{c'_1}{v}\piParal \pBEee{c'_1}}
      \end{align}
      
      We  assume $\envExt \vDash \sys{\sV}{\pF \piParalL  \piOutA{b}{c_2} \piParalL \piOutA{c_1}{(v,c_2)}\piParalL \pBbb{c_1}} \; \piCostBisAmm{n}\;
      \sys{\sVV}{\pFfff{c'_1}{v}\piParal \pBEee{c'_1}}$ and then showing that this leads to a contradiction.  Consider the move
      \begin{equation*}
        \conf{\envExt}{\sys{\sV}{\pF \piParalL  \piOutA{b}{c_2} \piParalL \piOutA{c_1}{(v,c_2)}\piParalL \pBbb{c_1}}} \piRedDecCostPad{\acttau}{0} \conf{\envExt}{\sys{\sV}{\pF \piParalL  \piOutA{b}{c_2} \piParalL \pBbbb{v}{c_2}}}
      \end{equation*}
      This can be matched by $\conf{\envExt}{\sys{\sVV}{\pFfff{c'_1}{v}\piParal \pBEee{c'_1}}}$ using either of the following moves:
      \begin{itemize}
      \item $\conf{\envExt}{\sys{\sVV}{\pFfff{c'_1}{v}\piParal \pBEee{c'_1}}} \piRedWDecCostPad{\phantom{\acttau}}{0} \conf{\envExt}{\sys{\sVV}{\pFfff{c'_1}{v}\piParal \pBEee{c'_1}}}$. But \eqref{eq:15cs} prohibits this from being the matching move. 
      \item $\conf{\envExt}{\sys{\sVV}{\pFfff{c'_1}{v}\piParal \pBEee{c'_1}}} \piRedWDecCostPad{\phantom{\acttau}}{+1} \conf{\envExt}{\sys{\sVV}{Q}}$ for some $Q \in \psA$. But \eqref{eq:14cs} prohibits this from being the matching move.
      \item $\conf{\envExt}{\sys{\sVV}{\pFfff{c'_1}{v}\piParal \pBEee{c'_1}}} \piRedWDecCostPad{\phantom{\acttau}}{0} \conf{\envExt}{\sys{\sVV}{Q}}$ for some $Q \in \psB$. But \eqref{eq:13cs}  prohibits this from being the matching move.
      \end{itemize}
      This contradicts our earlier assumption. \qedhere
  \end{enumerate}
\end{proof}

\begin{lem}\label{lem:negative-induction}  For all $n\in\Nats$ and appropriate $\sV,\sVV$:
  \begin{enumerate}
  \item For any $Q \in \psA$ we have $\;\envExt \vDash \sys{\sV}{\pF\piParal \piOutA{b}{c_2} \piParal  \pBbbb{v}{c_2}}  \; \not\!\piCostBisAmm{n}\; \sys{\sVV}{Q}$ 
  \item For any $Q \in \psA$ we have $\;\envExt \vDash \sys{\sV}{\pF\piParal \piOutA{b}{c_2} \piParal \piOutA{c_1}{(v,c_2)}\piParal \pBbb{c_1}}  \; \not\!\piCostBisAmm{n}\; \sys{\sVV}{Q}$
  \item For any $Q \in \psA$ we have $\;\envExt \vDash \sys{\sV}{\pFfff{c_1}{v}\piParal \pBbb{c_1}}  \; \not\!\piCostBisAmm{{n+1}}\; \sys{\sVV}{Q}$
  \item For any $Q \in \psB$ we have $\;\envExt \vDash \sys{\sV}{\pFfff{c_1}{v}\piParal \pBbb{c_1}}  \; \not\!\piCostBisAmm{{n}}\; \sys{\sVV}{Q}$
  \item $\envExt \vDash \sys{\sV}{\pFfff{c_1}{v}\piParal \pBbb{c_1}}  \; \not\!\piCostBisAmm{{n}}\; \sys{\sVV}{\pFfff{c'_1}{v}\piParal \pBEee{c'_1}}$
  \end{enumerate}
\end{lem}

\begin{proof} We prove statements $(1)$  to $(5)$ simultaneously, by induction on $n$.
  \begin{description}
  \item[$n=0$] We prove each clause by contradiction:
    \begin{enumerate}
    \item Assume $\envExt \vDash \sys{\sV}{\pF\piParal \piOutA{b}{c_2} \piParal  \pBbbb{v}{c_2}}  \; \piCostBisAmm{0}\; \sys{\sVV}{Q}$ for some $Q \in \psA$ and consider the transition
      \begin{equation*}
         \conf{\envExt}{\sys{\sV}{\pF\piParal \piOutA{b}{c_2} \piParal  \pBbbb{v}{c_2}}} \piRedDecCostPad{\actout{\ctit{out}}{v}}{0} \conf{\envExt}{\sys{\sV}{\pF\piParal \piOutA{b}{c_2} \piParal  \pB \piParal \piOutA{d}{c_2}}}
       \end{equation*}
         For any $Q \in \psA$, this cannot be matched by any move from  $\conf{\envExt}{\sys{\sVV}{Q}}$ since output actions must be preceded by a channel deallocation, which incurs a \emph{negative} cost --- see second and third rows of the  graph in \figref{fig:TrnGrphBuff2}.  Stated otherwise, every matching move can only be of the form 
$$\conf{\envExt}{\sys{\sVV}{Q}} \piRedWDecCostPad{\actout{\ctit{out}}{v}}{-1} \conf{\envExt}{\sys{\sVV'}{Q'}}$$
where $\sVV=\bigl(\sVV',c'_1\bigr)$ for some $c'_1$ and $Q'\in \psC$.  However, since the amortisation credit cannot be negative,  we can never have  $\envExt \vDash \sys{\sV}{\pF\piParal \piOutA{b}{c_2} \piParal  \pB \piParal \piOutA{d}{c_2}}  \; \piCostBisAmm{-1}\; \sys{\sVV'}{Q'}$.  We therefore obtain a contradiction.

\item Assume $\envExt \vDash \sys{\sV}{\pF\piParal \piOutA{b}{c_2} \piParal \piOutA{c_1}{(v,c_2)}\piParal \pBbb{c_1}}  \; \piCostBisAmm{0}\; \sys{\sVV}{Q}$ for some  $Q \in \psA$ and consider the transition
  \begin{align*}
 \qquad\qquad&\conf{\envExt}{\sys{\sV}{\pF\piParal \piOutA{b}{c_2} \piParal  \piOutA{c_1}{(v,c_2)}\piParal \pBbb{c_1}}}   \piRedDecCostPad{\acttau}{0} \conf{\envExt}{\sys{\sV}{\pF\piParal \piOutA{b}{c_2} \piParal  \pBbbb{v}{c_2}}}
\end{align*}
 Since the amortisation credit can never be negative, the matching move can only be of the form 
$$\conf{\envExt}{\sys{\sVV}{Q}} \piRedWDecCostPad{\phantom{\acttau}}{0} \conf{\envExt}{\sys{\sVV}{Q'}}$$
for some  $Q'\in \psA$.  But then we get a contradiction since, from the previous clause, we know that $\envExt \vDash \sys{\sV}{\pF\piParal \piOutA{b}{c_2} \piParal  \pBbbb{v}{c_2}}  \; \not\!\piCostBisAmm{0}\; \sys{\sVV}{Q'}$.

\item Assume $\envExt \vDash \sys{\sV}{\pFfff{c_1}{v}\piParal \pBbb{c_1}}  \; \piCostBisAmm{{1}}\; \sys{\sVV}{Q}$ for some $Q \in \psA$ and consider the transition
  \begin{align*}
   \qquad\qquad&\conf{\envExt}{\sys{\sV}{\pFfff{c_1}{v}\piParal \pBbb{c_1}}}    \piRedDecCostPad{\acttau}{+1} \conf{\envExt}{\sys{\sV,c_2}{\pF\piParal \piOutA{b}{c_2} \piParal  \piOutA{c_1}{(v,c_2)}\piParal \pBbb{c_1}}}
  \end{align*}
  for some newly allocated channel $c_2$. As in the previous case, since the amortisation credit can never be negative, the matching move can only be of the form 
$$\conf{\envExt}{\sys{\sVV}{Q}} \piRedWDecCostPad{\phantom{\acttau}}{0} \conf{\envExt}{\sys{\sVV}{Q'}}$$
for some  $Q'\in \psA$. But then we get a contradiction since, from the previous clause, we know that $\;\envExt \vDash \sys{(\sV,c_2)}{\pF\piParal \piOutA{b}{c_2} \piParal \piOutA{c_1}{(v,c_2)}\piParal \pBbb{c_1}}  \; \not\!\piCostBisAmm{0}\; \sys{\sVV}{Q'}$.

\item Analogous to the previous case.

\item Assume $\envExt \vDash \sys{\sV}{\pFfff{c_1}{v}\piParal \pBbb{c_1}}  \; \piCostBisAmm{{0}}\; \sys{\sVV}{\pFfff{c'_1}{v}\piParal \pBEee{c'_1}}$ and consider the transition
     \begin{align*}
   \qquad\qquad&\conf{\envExt}{\sys{\sV}{\pFfff{c_1}{v}\piParal \pBbb{c_1}}}    \piRedDecCostPad{\acttau}{+1} \conf{\envExt}{\sys{\sV,c_2}{\pF\piParal \piOutA{b}{c_2} \piParal  \piOutA{c_1}{(v,c_2)}\piParal \pBbb{c_1}}}
  \end{align*}
  Since the transition incurred a cost of $+1$ and the current amortisation credit is $0$, the matching weak transition must also incur a  cost of $+1$ and thus $\conf{\envExt}{\sys{\sVV}{\pFfff{c'_1}{v}\piParal \pBEee{c'_1}}}$ can only match this by the move
  \begin{equation*}
    \conf{\envExt}{\sys{\sVV}{\pFfff{c'_1}{v}\piParal \pBEee{c'_1}}} \piRedWDecCostPad{\acttau}{+1} \conf{\envExt}{\sys{\sVV,c'_2}{Q}}
  \end{equation*}
  for some $Q \in \psA$.  But then we still get a contradiction since, from clause $(2)$, we know  $\envExt \vDash \sys{\sV}{\pF\piParal \piOutA{b}{c_2} \piParal \piOutA{c_1}{(v,c_2)}\piParal \pBbb{c_1}}  \; \not\!\piCostBisAmm{0}\; \sys{\sVV}{Q}$.\\
    \end{enumerate}

  \item[$n=k+1$] We prove each clause by contradiction.  However before we tackle each individual clause, we note that from clauses $(3)$, $(4)$ and $(5)$ of the I.H. we know
      \begin{align*}
        &  \text{For any  }Q \in \psA \text{ we have }\envExt \vDash \sys{\sV}{\pFfff{c_1}{v}\piParal \pBbb{c_1}}  \; \not\!\piCostBisAmm{{k+1}}\; \sys{\sVV}{Q}\\
        & \text{For any } Q \in \psB \text{ we have }\envExt \vDash \sys{\sV}{\pFfff{c_1}{v}\piParal \pBbb{c_1}}  \; \not\!\piCostBisAmm{{k}}\; \sys{\sVV}{Q} \\
        & \envExt \vDash \sys{\sV}{\pFfff{c_1}{v}\piParal \pBbb{c_1}}  \; \not\!\piCostBisAmm{{k}}\; \sys{\sVV}{\pFfff{c_1}{v}\piParal \pBEee{c_1}}
      \end{align*}
      By Lemma~$\ref{lem:Neg:implications}.1$ we obtain, for any $Q'\in\psC$  and appropriate $\sVV'$:
      \begin{align*}
        & \envExt \vDash \sys{\sV}{\pFff{c_1} \piParalL \pBbb{c_1}} \; \not\!\piCostBisAmm{k}\; \sys{\sVV'}{Q'} 
      \end{align*}
      and by Lemma~$\ref{lem:Neg:implications}.2(a)$, Lemma~$\ref{lem:Neg:implications}.2(b)$, Lemma~$\ref{lem:Neg:implications}.2(c)$ and Lemma~$\ref{lem:Neg:implications}.2(d)$  we obtain, for any $Q'\in\psC$  and appropriate $\sVV'$:
     \begin{align}\label{eq:20cs}
        & \envExt \vDash \sys{\sV}{\pF\piParal \piOutA{b}{c_2} \piParal  \pB \piParal \piOutA{d}{c_2}} \; \not\!\piCostBisAmm{k}\;  \sys{\sVV}{Q'} 
      \end{align}  
      Also, by \eqref{eq:20cs}, Lemma~$\ref{lem:Neg:implications}.3(a)$ and Lemma~$\ref{lem:Neg:implications}.3(b)$ we obtain, for any $Q''\in \psB$:
      \begin{align}
         \label{eq:16cs}
         &\envExt \vDash \sys{\sV}{\pF \piParalL  \piOutA{b}{c_2} \piParalL \pBbbb{v}{c_2}} \; \not\!\piCostBisAmm{k}\;     \sys{\sVV}{Q''}\\
         \label{eq:17cs}
         & \envExt \vDash \sys{\sV}{\pF \piParalL  \piOutA{b}{c_2} \piParalL \piOutA{c_1}{(v,c_2)}\piParalL \pBbb{c_1}} \; \not\!\piCostBisAmm{k}\;\sys{\sVV}{Q''}   
      \end{align}
  Moreover, by \eqref{eq:20cs}, Lemma~$\ref{lem:Neg:implications}.4(a)$, Lemma~$\ref{lem:Neg:implications}.4(b)$ and Lemma~$\ref{lem:Neg:implications}.4(c)$ we obtain:
  \begin{align}
    \label{eq:18cs}
    & \envExt \vDash \sys{\sV}{\pF \piParalL  \piOutA{b}{c_2} \piParalL \piOutA{c_1}{(v,c_2)}\piParalL \pBbb{c_1}} \; \not\!\piCostBisAmm{k}\;
      \sys{\sVV}{\pFfff{c'_1}{v}\piParal \pBEee{c'_1}}
  \end{align}

    The proofs for each clause are as follows:
    \begin{enumerate}
    \item Assume $\envExt \vDash \sys{\sV}{\pF\piParal \piOutA{b}{c_2} \piParal  \pBbbb{v}{c_2}}  \; \piCostBisAmm{{k+1}}\; \sys{\sVV}{Q}$ for some $Q \in \psA$ and consider the transition
      \begin{equation*}
         \conf{\envExt}{\sys{\sV}{\pF\piParal \piOutA{b}{c_2} \piParal  \pBbbb{v}{c_2}}} \piRedDecCostPad{\actout{\ctit{out}}{v}}{0} \conf{\envExt}{\sys{\sV}{\pF\piParal \piOutA{b}{c_2} \piParal  \pB \piParal \piOutA{d}{c_2}}}
       \end{equation*}
       For any $Q \in \psA$, this can (only) be matched by any move of the form  
       \begin{align*}
         \qquad\qquad & \conf{\envExt}{\sys{\sVV}{Q}} \piRedWDecCostPad{\actout{\ctit{out}}{v}}{-1} \conf{\envExt}{\sys{\sVV'}{Q'}}
       \end{align*}
      where $\sVV=\bigl(\sVV',c'_1\bigr)$ for some $c'_1$, $Q'\in \psC$, and the external action \actout{\ctit{out}}{v} is preceded by a $\tau$-move deallocating $c'_1$.  For our initial assumption to hold we need to show that \emph{at least one} of these configurations $\conf{\envExt}{\sys{\sVV'}{Q'}}$ satisfies the property $$\envExt \vDash \sys{\sV}{\pF\piParal \piOutA{b}{c_2} \piParal  \pB \piParal \piOutA{d}{c_2}} \; \piCostBisAmm{k}\;  \sys{\sVV'}{Q'}. $$
 But by \eqref{eq:20cs} we know that no such configuration exists,  thereby contradicting our initial assumption.

    \item Assume $\envExt \vDash \sys{\sV}{\pF\piParal \piOutA{b}{c_2} \piParal \piOutA{c_1}{(v,c_2)}\piParal \pBbb{c_1}}  \; \piCostBisAmm{{k+1}}\; \sys{\sVV}{Q}$ for some $Q \in \psA$ and consider the transition
      \begin{align*}
 \qquad\qquad&\conf{\envExt}{\sys{\sV}{\pF\piParal \piOutA{b}{c_2} \piParal  \piOutA{c_1}{(v,c_2)}\piParal \pBbb{c_1}}}   \piRedDecCostPad{\acttau}{0} \conf{\envExt}{\sys{\sV}{\pF\piParal \piOutA{b}{c_2} \piParal  \pBbbb{v}{c_2}}}
\end{align*}
     This transition can be matched by \conf{\envExt}{\sys{\sVV}{Q}} through either of the following moves:
     \begin{enumerate}
     \item $\conf{\envExt}{\sys{\sVV}{Q}} \piRedWDecCostPad{\phantom{\acttau}}{0} \conf{\envExt}{\sys{\sVV}{Q'}}$ for some $Q'\in\psA$.  However, from the previous clause, \ie clause $(1)$ when $n=k+1$, we know that this cannot be the matching move since $\envExt \vDash \sys{\sV}{\pF\piParal \piOutA{b}{c_2} \piParal  \pBbbb{v}{c_2}}  \; \not\!\piCostBisAmm{{k+1}}\; \sys{\sVV}{Q'}$.
     \item $\conf{\envExt}{\sys{\sVV}{Q}} \piRedWDecCostPad{\phantom{\acttau}}{-1} \conf{\envExt}{\sys{\sVV'}{Q'}}$ for some $Q'\in\psB$ and $\sVV=\bigl(\sVV',c'_1\bigr)$. However, from \eqref{eq:16cs},  we know that this cannot be the matching move since $\envExt \vDash \sys{\sV}{\pF \piParalL  \piOutA{b}{c_2} \piParalL \pBbbb{v}{c_2}} \; \not\!\piCostBisAmm{k}\;     \sys{\sVV'}{Q'}$.
     \end{enumerate}
     Thus, we obtain a contradiction.

   \item Assume  $\envExt \vDash \sys{\sV}{\pFfff{c_1}{v}\piParal \pBbb{c_1}}  \; \piCostBisAmm{{k+2}}\; \sys{\sVV}{Q}$, where $Q \in \psA$, and consider the transition:
       \begin{align*}
   \qquad\qquad&\conf{\envExt}{\sys{\sV}{\pFfff{c_1}{v}\piParal \pBbb{c_1}}}    \piRedDecCostPad{\acttau}{+1} \conf{\envExt}{\sys{\sV,c_2}{\pF\piParal \piOutA{b}{c_2} \piParal  \piOutA{c_1}{(v,c_2)}\piParal \pBbb{c_1}}}
  \end{align*}
  for some newly allocated channel $c_2$.  This can be matched by \conf{\envExt}{\sys{\sVV}{Q}} through either of the following moves:
  \begin{enumerate}
  \item $\conf{\envExt}{\sys{\sVV}{Q}} \piRedWDecCostPad{\phantom{\acttau}}{0} \conf{\envExt}{\sys{\sVV}{Q'}}$ for some $Q'\in\psA$. However, from the previous clause, \ie clause $(2)$ when $n=k+1$, we know that this cannot be the matching move since $\envExt \vDash \sys{\sV,c_2}{\pF\piParal \piOutA{b}{c_2} \piParal  \piOutA{c_1}{(v,c_2)}\piParal \pBbb{c_1}}  \; \not\!\piCostBisAmm{{k+1}}\; \sys{\sVV}{Q'}$.
  \item $\conf{\envExt}{\sys{\sVV}{Q}} \piRedWDecCostPad{\phantom{\acttau}}{-1} \conf{\envExt}{\sys{\sVV'}{Q'}}$ for some $Q'\in\psB$ and $\sVV=\bigl(\sVV',c'_1\bigr)$. However, from \eqref{eq:17cs},  we know that this cannot be the matching move since $\envExt \vDash \sys{\sV}{\pF \piParalL  \piOutA{b}{c_2} \piParalL \piOutA{c_1}{(v,c_2)}\piParalL \pBbb{c_1}} \; \not\!\piCostBisAmm{k}\;\sys{\sVV'}{Q'}$.
  \end{enumerate}
     Thus, we obtain a contradiction.
 
  \item Analogous to the proof for the previous clause and relies on \eqref{eq:17cs} again.

  \item Assume $\envExt \vDash \sys{\sV}{\pFfff{c_1}{v}\piParal \pBbb{c_1}}  \; \piCostBisAmm{{k+1}}\; \sys{\sVV}{\pFfff{c'_1}{v}\piParal \pBEee{c'_1}}$ and consider the transition
     \begin{align*}
   \qquad\qquad&\conf{\envExt}{\sys{\sV}{\pFfff{c_1}{v}\piParal \pBbb{c_1}}}    \piRedDecCostPad{\acttau}{+1} \conf{\envExt}{\sys{\sV,c_2}{\pF\piParal \piOutA{b}{c_2} \piParal  \piOutA{c_1}{(v,c_2)}\piParal \pBbb{c_1}}}
  \end{align*}
   for some newly allocated channel $c_2$.  This can be matched by the right-hand configuration \conf{\envExt}{\sys{\sVV}{\pFfff{c'_1}{v}\piParal \pBEee{c'_1}}} through either of the following moves:
   \begin{enumerate}
   \item  $\conf{\envExt}{\sys{\sVV}{\pFfff{c'_1}{v}\piParal \pBEee{c'_1}}} \piRedWDecCostPad{\phantom{\acttau}}{0} \conf{\envExt}{\sys{\sVV}{\pFfff{c'_1}{v}\piParal \pBEee{c'_1}}}$, \ie no transitions.  However, from \eqref{eq:18cs},   this cannot be the matching move since $\envExt \vDash \sys{\sV}{\pF \piParalL  \piOutA{b}{c_2} \piParalL \piOutA{c_1}{(v,c_2)}\piParalL \pBbb{c_1}} \; \not\!\piCostBisAmm{k}\;\sys{\sVV}{\pFfff{c'_1}{v}\piParal \pBEee{c'_1}}$.
   \item $\conf{\envExt}{\sys{\sVV}{\pFfff{c'_1}{v}\piParal \pBEee{c'_1}}} \piRedWDecCostPad{\acttau}{+1} \conf{\envExt}{\sys{\sVV,c'_2}{Q'}}$ for some $Q'\in\psA$ and $c'_2 \not\in \sVV$. However, from  clause $(2)$ when $n=k+1$, this cannot be the matching move since $\envExt \vDash \sys{\sV,c_2}{\pF\piParal \piOutA{b}{c_2} \piParal  \piOutA{c_1}{(v,c_2)}\piParal \pBbb{c_1}}  \; \not\!\piCostBisAmm{{k+1}}\; \sys{\sVV,c'_2}{Q'}$.
   \item $\conf{\envExt}{\sys{\sVV}{\pFfff{c'_1}{v}\piParal \pBEee{c'_1}}} \piRedWDecCostPad{\acttau}{0} \conf{\envExt}{\sys{\bigl(\sVV',c'_2\bigr)}{Q'}}$ for some $Q'\in\psB$, $\sVV = \bigl(\sVV',c'_1\bigr)$ and $c'_2 \not\in \sVV$.   However, from \eqref{eq:17cs},   this cannot be the matching move since $\envExt \vDash \sys{\sV}{\pF \piParalL  \piOutA{b}{c_2} \piParalL \piOutA{c_1}{(v,c_2)}\piParalL \pBbb{c_1}} \; \not\!\piCostBisAmm{k}\;\sys{\bigl(\sVV',c'_2\bigr)}{Q'}$. \qedhere
   \end{enumerate}

    \end{enumerate}

  \end{description}
\end{proof}

\smallskip

\begin{thm}[Strict Inefficiency]\label{thm:strictly-ineffiscient} For all $ n \geq 0 \text{ and appropriate }\sV\text{ we have }$ $$\envExt \;\vDash\; \sysS{\pBuf}\;\not\!\piCostBisAmm{n}\; \sysS{\pBufE}$$
\end{thm}

\begin{proof} Since:
\begin{align*}
  \pBuf & \deftxt \pFff{c_1} \piParal \pBbb{c_1} & \pBufE & \deftxt \pFff{c_1} \piParal \pBEee{c_1} 
\end{align*}
we need to show that $$\envExt \;\vDash\; \sysS{\pFff{c_1} \piParal \pBbb{c_1}}\;\not\!\piCostBisAmm{n}\; \sysS{\pFff{c_1} \piParal \pBEee{c_1}}$$ for any arbitrary $n$.  By Lemma~$\ref{lem:negative-induction}.3$, Lemma~$\ref{lem:negative-induction}.4$ and Lemma~$\ref{lem:negative-induction}.5$ we know that for any $n$:
\begin{align}
\label{eq:21cs}
& \text{For any } Q \in \psA \text{ we have }    \envExt \vDash
  \sys{\sV}{\pFfff{c_1}{v}\piParal \pBbb{c_1}} \;
  \not\!\piCostBisAmm{{n+1}}\; \sys{\sV}{Q}\\
\label{eq:22cs}
& \text{For any }Q \in \psB \text{ we have }   \envExt \vDash
  \sys{\sV}{\pFfff{c_1}{v}\piParal \pBbb{c_1}} \;
  \not\!\piCostBisAmm{{n}}\; \sys{\sV}{Q}\\
\label{eq:23cs}
& \envExt \vDash \sys{\sV}{\pFfff{c_1}{v}\piParal \pBbb{c_1}} \;
  \not\!\piCostBisAmm{{n}}\; \sys{\sV}{\pFfff{c_1}{v}\piParal
    \pBEee{c_1}}
\end{align}
Since $\bigl(\pFff{c_1} \piParal \pBEee{c_1}\bigr) \in \psC$, by Lemma~$\ref{lem:Neg:implications}.1$, \eqref{eq:21cs},~\eqref{eq:22cs} and \eqref{eq:23cs} we conclude
\begin{displaymath}
   \envExt \;\vDash\; \sysS{\pFff{c_1} \piParal
    \pBbb{c_1}}\;\not\!\piCostBisAmm{n}\; \sysS{\pFff{c_1} \piParal
    \pBEee{c_1}} 
\end{displaymath}
as required.
\end{proof}

\subsection{Proving Relative Efficiency}
\label{sec:prov-relat-effc}

\newcommand{\envFrn}{\ensuremath{\env_\text{Frn}}}

As opposed to Theorem~\ref{thm:strictly-ineffiscient}, the proof for \eqref{eq:25:cs} 
requires us to consider the entire state-space of 
 \conf{\envExt}{\sysS{\pBuf}} and  \conf{\envExt}{\sysS{\pBufE}}. Fortunately, we can apply the compositionality result of Theorem~\ref{thm:costed-bisim-compositionality}  to prove \eqref{eq:11cs} and focus on a subset of this state-space.  More precisely, we recall from \eqref{eq:24cs} that 
\begin{align*}
  \pBuf & \deftxt \pFff{c_1} \piParal \pBbb{c_1} & \pBufE & \deftxt \pFff{c_1} \piParal \pBEee{c_1} 
\end{align*}
where both buffer implementation share the common sub-process $\pFff{c_1}$.  We also recall from \eqref{eq:9cs} that this common sub-process was typed \wrt the type environment $$\envFrn = \emap{\ctit{in}}{\chantypW{\tV}},\,  \emap{b}{\chantypW{\tVR}},  \,\emap{c_1}{\chantypO{\tV,\tVR}}.$$
Theorem~\ref{thm:costed-bisim-compositionality} thus states that in order to prove \eqref{eq:11cs}, it suffices to abstract away from this common code and prove Theorem~\ref{thm:relat-eff} 

\begin{thm}[Relative Efficiency]\label{thm:relat-eff}
  $\bigl(\envExt,\envFrn\bigr) \vDash \sysS{\pBEee{c_1}} \;\piCostBisAmm{0}\; \sysS{\pBbb{c_1}}$
\end{thm}

\begin{proof}  We prove $\envExt,\envFrn \vDash \sysS{\pBEee{c_1}} \;\piCostBisAmm{0}\; \sysS{\pBbb{c_1}}$  through the family of relations \relR\ defined below, which includes the required quadruple
$\langle (\envExt,\envFrn), 0, \bigl(\sysS{\pBEee{c_1}}\bigr),  \bigl(\sysS{\pBbb{c_1}}\bigr)\rangle$.
  \begin{equation*}
    \relR \deftxt\left\{
      \begin{array}{@{\langle\,}l@{,\;}l@{,\;}l@{,\;}l@{\,\rangle\;}|@{\quad}l}
        \bigl(\env,\envv\bigr) & n & \bigl(\sys{\sV'}{\pBEee{c}}\bigr) & \bigl(\sys{\sVV'}{\pBbb{c}}\bigr)  \\
        \bigl(\env,\envv\bigr) & n & \bigl(\sys{\sV'}{\pBEeee{c}{v}{c'}}\bigr) & \bigl(\sys{\sVV'}{\pBbbb{v}{c'}}\bigr) & \bigl(\envExt,\envFrn\bigr) \struct \env\\
        \bigl(\env,\envv\bigr) & n & \bigl(\sys{\sV''}{\pBEeeee{v}{c'}}\bigr) & \bigl(\sys{\sVV'}{\pBbbb{v}{c'}}\bigr) & n \geq 0,\; \sV' \subseteq \sVV'\\
         \bigl(\env,\envv\bigr) & n & \bigl(\sys{\sV''}{\pBE \piParal \piOutA{d}{c'}}\bigr) & \bigl(\sys{\sVV'}{\pB \piParal \piOutA{d}{c'}}\bigr) & c \not\in \sV'',\; \sV'' \subset \sVV'',\;  c \in \sVV''  \\
        \bigl(\env,\envv\bigr) & n & \bigl(\sys{\sV''}{\pBEe \piParal \piOutA{d}{c'}}\bigr) & \bigl(\sys{\sVV'}{\pBb \piParal \piOutA{d}{c'}}\bigr)
      \end{array}
    \right\}
  \end{equation*}
Note that, in the quadruples of $\relR$ our observer environment  is not limited to derived environments $\env$ obtained from restructurings of  $\envExt,\envFrn$, but may include also additional entries, denoted by the environment \envv; these originate from observer channel allocations and uses through the transition rules \rtit{lAllE} and \rtit{lStr} from \figref{fig:LTS}.
\relR\ observes the transfer property of Definition~\ref{def:amortized-typed-bisim}.  We here go over some  key transitions:
\begin{itemize}
\item Consider a tuple from the first clause of the relation, for some $\env,\envv, n$ and $c$ \ie 
$$ \bigl(\env,\envv\bigr) \vDash \bigl(\sys{\sV'}{\pBEee{c}}\bigr) \;\relR^n\; \bigl(\sys{\sVV'}{\pBbb{c}}\bigr) $$
We recall from the macros introduced in Section~\ref{sec:proofs-relat-effic}  that
\begin{align*}
  \pBEee{c} &= \piIn{c}{(y,z)}{\piFree{c}{\piOut{\ctit{out}}{y}{
        \bigl(\pBE \piParal \piOutA{d}{z}\bigr)}}}\\
  \pBbb{c} &=\piIn{c}{(y,z)}{ \piOut{\ctit{out}}{y}{ \bigl(\pB \piParal
      \piOutA{d}{z}\bigr)}}
\end{align*}
Whenever $\bigl(\env,\envv\bigr)$ allows it, the left hand configuration can perform an input transitions
$$\conf{\bigl(\env,\envv\bigr)}{\sys{\sV'}{\pBEee{c}}} \piRedDecCostPad{\actin{c}{(v,c')}}{0} \conf{\bigl(\env',\envv'\bigr)}{\sys{\sV'}{\pBEeee{c}{v}{c'}}} $$
where $\env = \env',\envmap{c}{\chantypO{\tV,\tVR}}$ and $\envv = \envv', \envmap{v}{\tV},\envmap{c'}{\tVR}$.  This can be matched by the transition
$$\conf{\bigl(\env,\envv\bigr)}{\sys{\sVV'}{\pBbb{c}}} \piRedDecCostPad{\actin{c}{(v,c')}}{0} \conf{\bigl(\env',\envv'\bigr)}{\sys{\sVV'}{\pBbbb{v}{c'}}} $$
where we have $ \bigl(\env',\envv'\bigr) \vDash \bigl(\sys{\sV'}{\pBEeee{c}{v}{c'}}\bigr) \;\relR^n\; \bigl(\sys{\sVV'}{\pBbbb{v}{c'}}\bigr) $  from the second clause of \relR.  The matching move for an input action from the right-hand configuration is dual to this. Matching moves for \actenv, \actall\ and \actfree{c} actions are analogous.
\item Consider a tuple from the first clause of the relation, for some $\env,\envv, n, c, v$ and $c'$ \ie 
$$ \bigl(\env,\envv\bigr) \vDash \bigl(\sys{\sV'}{\pBEeee{c}{v}{c'}}\bigr) \;\relR^n\; \bigl(\sys{\sVV'}{\pBbbb{v}{c'}}\bigr) $$
Since $\pBEeee{c}{v}{c'}= \piFree{c}{\piOut{\ctit{out}}{v}{ \bigl(\pBE \piParal \piOutA{d}{c'}\bigr)}}$, a possible transition by the left-hand configuration is the deallocation of channel $c$:
$$\conf{\bigl(\env,\envv\bigr)}{\sys{\sV'}{\pBEeee{c}{v}{c'}}} \piRedDecCostPad{\acttau}{-1} \conf{\bigl(\env,\envv\bigr)}{\sys{\sV''}{\pBEeeee{v}{c'}}} $$
where $\sV'=\sV'',c$.  In this case, the matching move is the empty (weak) transition, since we have
$ \bigl(\env,\envv\bigr) \vDash \bigl(\sys{\sV''}{\pBEeeee{v}{c'}}\bigr) \;\relR^{n+1}\; \bigl(\sys{\sVV'}{\pBbbb{v}{c'}}\bigr) $ by the third clause of \relR. Dually, if $\bigl(\env,\envv\bigr)$ allows it, the right hand configuration may perform an output action
$$\conf{\bigl(\env,\envv\bigr)}{\sys{\sVV'}{\pBbbb{v}{c'}}} \piRedDecCostPad{\actout{\ctit{out}}{v}}{0} \conf{\bigl(\env,\envv,\envmap{v}{\tV}\bigr)}{\sys{\sVV'}{\pB\piParal\piOutA{d}{c'}}}$$
This can be matched by the weak output action
$$\conf{\bigl(\env,\envv\bigr)}{\sys{\sV'}{\pBEeee{c}{v}{c'}}} \piRedWDecCostPad{\actout{\ctit{out}}{v}}{-1} \conf{\bigl(\env,\envv,\envmap{v}{\tV}\bigr)}{\sys{\sV''}{\pBE\piParal\piOutA{d}{c'}}}$$
where $\sV'=\sV'',c$; by the fourth clause of \relR, we know that this a matching move because $ \bigl(\env,\envv,\envmap{v}{\tV}\bigr) \vDash \bigl(\sys{\sV''}{\pBE\piParal\piOutA{d}{c'}}\bigr) \;\relR^{n+1}\; \bigl(\sys{\sVV'}{\pB\piParal\piOutA{d}{c'}}\bigr)$. \qedhere
\end{itemize}
\end{proof}




 \section{Related Work}
 \label{sec:RelatedWork}

 \emph{A note on terminology:} From a logical perspective, a \emph{linear}
 assumption is one that cannot be weakened nor contracted, while an
 \emph{affine} assumption cannot be contracted but can be weakened. This leads
 to a reading of linear as ``used exactly once'' and of affine as ``used at most
 once''. However, in the presence of divergence or deadlock,  
 most linear type systems do not in fact guarantee
 that a linear resource will be used exactly once. In the discussion below, we
 will classify such type systems as affine instead.

 Linear logic was introduced by Girard \cite{girard:linearlogic}; its use as a
 type system was pioneered by Wadler \cite{wadler:use}. Uniqueness typing was
 introduced by Barendsen and Smetsers \cite{barendsen:functional}; the relation
 to linear logic has since been discussed in a number of papers (see
 \cite{hage:usageanalysis}).  

 Although there are many substructural (linear or affine) type systems for
 process calculi \cite[and
 others]{Acciai:typeabstraction,Acciai:responsiveness,Amadio:receptive,Igarashi:generic,Kobayashi:hybrid,Nobuko:dependent},
 some specifically for resources \cite{Kobayashi:resourceusage}, the literature
 on \emph{behaviour} of processes typed under such type systems is much smaller.

 Kobayashi \emph{et al.} \cite{KobayashiPT:linearity} introduce an affine type
 system for the \pic. Their channels have a polarity (input, output, or
 input/output) as well as a multiplicity (unrestricted or affine), and an affine
 input/output can be split as an affine input and an affine output channel.
 Communication on an affine input/affine output channel is necessarily
 deterministic, like communication on an affine/unique-after-1 channel in our
 calculus; however, both processes lose the right to use the channel after the
 communication, limiting reuse. Although the paper gives a definition of
 reduction closed barbed congruence, no compositional proof methods are
 presented.  

 Yoshida \emph{et al} \cite{Yoshida07:linearity,Honda:processlogic} define a
 linear type system, which uses ``action types'' to rule out deadlock.  The use
 of action types means that the type system can provide some guarantees that we
 cannot;  this is however an orthogonal aspect of the type system and it would
 be interesting to see if similar techniques can be applied in our setting. Their
 type system does not have any type that corresponds to uniqueness; instead, the
 calculus is based on $\pi$I to control dynamic sharing of names syntactically,
 thereby limiting channel reuse.  The authors give compositional proof
 techniques for their behavioural equivalence, but give no complete
 characterization. 

 Teller \cite{teller:resourcespi} introduces a \pic variant with ``finalizers'',
 processes that run when a resource has been deallocated. The deallocation
 itself however is performed by a garbage collector. The calculus comes with a
 type system that provides bounds on the resources that are used, although the
 scope of channel reuse is limited in the absence of some sort of uniqueness
 information. Although the paper defines a bisimulation relation, this relation
 does not take advantage of type information, and no compositionality results or
 characterization is given. 

 Hoare and O'Hearn \cite{Hoare:seplogicCSP} give a trace semantics for a variant
 of CSP with point-to-point communication and explicit allocation and
 deallocation of channels, which relies on separation of permissions. However,
 they do not consider any behavioural theories.  Pym and Tofts
 \cite{Pym:calculus} similarly give a semantics for SCCS with a generic notion
 of resource, based on separation of permissions; they do however consider
 behaviour. They define a bisimulation relation, and show that it can be
 characterized by a modal logic.   These approaches do not use a type system but opt for an operational interpretation of permissions, where actions may block due to lack of permissions.   Nevertheless, our consistency requirements for configurations (\defref{def:configuration}) can be seen as separation criteria for permission environments.   A detailed comparison between this untyped approach and our typed approach would be worthwhile.


 Apart from the Clean programming language \cite{journals/mscs/BarendsenS96}, from where uniqueness types originated, static analysis relating to uniqueness has recently been applied to  (more mainstream) Object-Ori\-en\-ted programming languages \cite{Gordon:unique:12} as well.  In such cases, it would be interesting to investigate whether the techniques developed in this work can be applied to a behavioural setting  such as that in \cite{JeffreyR:JavaJr:05}.      

 Our unique-after-$i$ type is related to fractional permissions, introduced in
 \cite{boyland:03fractions} and used in settings such as separation logic for
 shared-state concurrency \cite{Bornat:05Separation}. A detailed survey of this
 field is however beyond the scope of this paper. 


 The use of substitutions in our LTS (\defref{def:renaming}) is reminiscent of the name-bijections
 carried around in spi-calculus bisimulations \cite{boreale:crypto}. In the
 spi-calculus however this substitution is carried through the bisimulation, and
 must remain a bijection throughout. 
 Since processes may lose the permission
 to use channels in our calculus, this approach is too restrictive for us. 

 Finally, amortisation for coinductive reasoning was originally developed by Keihn \etal, \cite{Kiehn05} and L\"uttgen \etal \cite{LuttgenV06}.  It is investigated further by Hennessy in \cite{hennessy:buysell}, whereby a correspondence with (an adaptation of) reduction-barbed congruences is established. However, neither work considers aspects of resource misuse nor the corresponding use of typed analysis in their behavioural and coinductive equivalences.



\section{Conclusion}
\label{sec:conclusion}

We have presented a compositional behavioural theory for  \picr, a \pic variant with mechanisms for explicit resource management; a preliminary version of the work appeared in \cite{DevFraHen09}.  The theory allows us to compare the efficiency of concurrent channel-passing programs \wrt their resource usage.  We integrate the theory with a substructural type system so as to limit our comparisons to safe programs.  In particular, we interpret the type assertions of the type system as permissions, and use this to model (explicit and implicit) permission transfer between the systems being compared and the observer during compositional reasoning.    Our contributions are as follows:
\begin{enumerate}
\item We define a costed semantic theory that orders systems of safe \picr programs, based on their costed extensional behaviour when deployed in the context of larger systems; Definition~\ref{def:cost-preorder}. Apart from cost, formulations relating to contextuality are different from those of typed congruences such as \cite{hennessy04behavioural}, because of the kind of type system used \ie substructural. 
\item We define a bisimulation-based proof technique that allows us to order \picr programs coinductively, without the need to universally quantify over the possible contexts that these programs may be deployed in; Definition~\ref{def:amortized-typed-bisim}.  As far as we are aware, the combination of actions-in-context and costed semantics, used in unison with 
  implicit and explicit transfer of permissions so as to limit the efficiency analysis to safe programs, is new. 
\item We prove a number of properties for our bisimulation preorder of Definition~\ref{def:amortized-typed-bisim}, facilitating the proof constructions for related programs.  Whereas Corollary~\ref{cor:preorder} follows \cite{Kiehn05,hennessy:buysell},  Theorem~\ref{thm:costed-bisim-compositionality} extends the property of compositionality for amortised bisimulations to a typed setting. Lemma~\ref{lem:symmetry-bound}, together with the concept of bounded amortisation, appears to be novel altogether.
\item We prove that the bisimulation preorder of Definition~\ref{def:amortized-typed-bisim} is a sound and complete proof technique for the costed behavioural preorder of Definition~\ref{def:cost-preorder}; Theorem~\ref{thm:soundness}  and Theorem~\ref{thm:completeness}. In order to obtain completeness, the LTS definitions employ non-standard mechanisms for explicit renaming of channel names not known to the context. Also, the concept of (typed) action definability \cite{hennessy04behavioural,Hennessy07} is different because it needs to take into consideration cost and typeability \wrt a substructural type system; the latter aspect also complicated the respective Extrusion Lemma --- see Lemma~\ref{lem:extrusion}.   
\item We demonstrate the utility of the semantic theory and its respective proof technique by applying them to reason about the client-server systems outlined in the Introduction and a case study, discussed in Section~\ref{sec:case-study}.
\end{enumerate}

\subsection*{Future Work}
\label{sec:future-work}
The extension of our framework to a higher-order and distributed setting seems worthwhile.  Also, the amalgamation of our uniqueness types with modalities for input and output \cite{PierceS96} would give scope for richer notions of subtyping involving covariance and contravariance, affecting the respective behavioural theory; it would be interesting to explore how our notions of permission transfer extend to such a setting.   It is also worth pursuing the applicability of the techniques developed in this work to nominal automata such as Variable Automata \cite{LATA10} and Finite-Memory Automata \cite{Kaminski1994329}. 

\subsection*{Acknowledgements}
\label{sec:acknowledgements}

We would like to thank the referees for their incisive comments.

\small
\bibliographystyle{plain}
\bibliography{bibliography}

\begin{thebibliography}{10}

\bibitem{Acciai:typeabstraction}
Lucia Acciai and Michele Boreale.
\newblock Type abstractions of name-passing processes.
\newblock In {\em FSEN'07}, pages 302--317, 2007.

\bibitem{Acciai:responsiveness}
Lucia Acciai and Michele Boreale.
\newblock Responsiveness in process calculi.
\newblock {\em Theor. Comput. Sci.}, 409(1):59--93, 2008.

\bibitem{Amadio:receptive}
Roberto~M. Amadio, G\'{e}rard Boudol, and C\'{e}dric Lhoussaine.
\newblock The receptive distributed $\pi$-calculus.
\newblock {\em ACM Trans. Program. Lang. Syst.}, 25(5):549--577, 2003.

\bibitem{Arun-Kumar:1992}
S.~Arun-Kumar and Matthew Hennessy.
\newblock An efficiency preorder for processes.
\newblock {\em Acta Inf.}, 29(9):737--760, December 1992.

\bibitem{barendsen:functional}
Erik Barendsen and Sjaak Smetsers.
\newblock Uniqueness typing for functional languages with graph rewriting
  semantics.
\newblock {\em MSCS}, 6:579--612, 1996.

\bibitem{journals/mscs/BarendsenS96}
Erik Barendsen and Sjaak Smetsers.
\newblock Uniqueness typing for functional languages with graph rewriting
  semantics.
\newblock {\em {M}athematical {S}tructures in {C}omputer {S}cience},
  6(6):579--612, 1996.

\bibitem{boreale:crypto}
Michele Boreale, Rocco~De Nicola, and Rosario Pugliese.
\newblock Proof techniques for cryptographic processes.
\newblock {\em SIAM J. Comput.}, 31(3):947--986, 2001.

\bibitem{Bornat:05Separation}
Richard Bornat, Cristiano Calcagno, Peter O'Hearn, and Matthew Parkinson.
\newblock Permission accounting in separation logic.
\newblock {\em SIGPLAN Not.}, 40(1):259--270, 2005.

\bibitem{boyland:03fractions}
John Boyland.
\newblock Checking interference with fractional permissions.
\newblock In R.~Cousot, editor, {\em Static Analysis: 10th International
  Symposium}, volume 2694 of {\em LNCS}, pages 55--72. Springer, 2003.

\bibitem{Bulka:1999:ECP:320041}
Dov Bulka and David Mayhew.
\newblock {\em Efficient C++: performance programming techniques}.
\newblock Addison-Wesley Longman Publishing Co., Inc., Boston, MA, USA, 2000.

\bibitem{DevFraHen09}
Edsko DeVries, Adrian Francalanza, and Matthew Hennessy.
\newblock Reasoning about explicit resource management (extended abstract).
\newblock In {\em PLACES}, pages 15--21. ETAPS, April 2011.
\newblock http://places11.di.fc.ul.pt/.

\bibitem{EFH:uniqueness:journal:12}
Edsko DeVries, Adrian Francalanza, and Matthew Hennessy.
\newblock Uniqueness typing for resource management in message-passing
  concurrency.
\newblock {\em Journal of Logic and Computation}, June 2012.

\bibitem{FraRatSas11}
Adrian Francalanza, Julian Rathke, and Vladimiro Sassone.
\newblock Permission-based separation logic for message-passing concurrency.
\newblock {\em Logical Methods in Computer Science}, 7(3), 2011.

\bibitem{girard:linearlogic}
Jean-Yves Girard.
\newblock Linear logic.
\newblock {\em Theoretical Computer Science}, 50(1):1--102, 1987.

\bibitem{Gordon:unique:12}
Colin~S. Gordon, Matthew~J. Parkinson, Jared Parsons, Aleks Bromfield, and Joe
  Duffy.
\newblock {Uniqueness and Reference Immutability for Safe Parallelism}.
\newblock In {\em {Proceedings of the 2012 ACM International Conference on
  Object Oriented Programming, Systems, Languages, and Applications
  (OOPSLA'12)}}, {Tucson, AZ, USA}, October 2012.

\bibitem{LATA10}
Orna Grumberg, Orna Kupferman, and Sarai Sheinvald.
\newblock Variable automata over infinite alphabets.
\newblock In Adrian-Horia Dediu, Henning Fernau, and Carlos Martín-Vide,
  editors, {\em Language and Automata Theory and Applications}, volume 6031 of
  {\em LNCS}, pages 561--572. Springer, 2010.

\bibitem{hage:usageanalysis}
Jurriaan Hage, Stefan Holdermans, and Arie Middelkoop.
\newblock A generic usage analysis with subeffect qualifiers.
\newblock In {\em Proceedings of the 12th ACM SIGPLAN International Conference
  on Functional Programming (ICFP)}, pages 235--246. ACM, 2007.

\bibitem{harrington:uniquenesslogic}
Dana Harrington.
\newblock Uniqueness logic.
\newblock {\em Theoretical Computer Science}, 354(1):24--41, 2006.

\bibitem{Hennessy07}
Matthew Hennessy.
\newblock {\em A Distributed Picalculus}.
\newblock Cambridge University Proess, Cambridge, UK., 2008.

\bibitem{hennessy:buysell}
Matthew Hennessy.
\newblock A calculus for costed computations.
\newblock {\em Logical Methods in Computer Science}, 7(1), 2011.

\bibitem{hennessy04behavioural}
Matthew Hennessy and Julian Rathke.
\newblock Typed behavioural equivalences for processes in the presence of
  subtyping.
\newblock {\em Mathematical Structures in Computer Science}, 14:651--684, 2004.

\bibitem{Hoare:seplogicCSP}
Tony Hoare and Peter O'Hearn.
\newblock Separation logic semantics for communicating processes.
\newblock {\em ENTCS}, 212:3--25, 2008.

\bibitem{Honda:processlogic}
Kohei Honda.
\newblock From process logic to program logic.
\newblock In {\em ICFP '04}, pages 163--174, 2004.

\bibitem{HondaTokoro92:AsyncSemantics}
Kohei Honda and Mario Tokoro.
\newblock On asynchronous communication semantics.
\newblock In Mario Tokoro, Oscar Nierstrasz, and Peter Wegner, editors, {\em
  Proceedings of the {ECOOP}'91 Workshop on Object-Based Concurrent Computing},
  volume 612 of {\em LNCS}, pages 21--51. Springer-Verlag, 1992.

\bibitem{Igarashi:generic}
Atsushi Igarashi and Naoki Kobayashi.
\newblock A generic type system for the pi-calculus.
\newblock {\em Theor. Comput. Sci.}, 311(1-3):121--163, 2004.

\bibitem{JeffreyR:JavaJr:05}
Alan Jeffrey and Julian Rathke.
\newblock {Java Jr}: {F}ully abstract trace semantics for a core java language.
\newblock In Shmuel Sagiv, editor, {\em ESOP}, volume 3444 of {\em LNCS}, pages
  423--438. Springer, 2005.

\bibitem{JonesGC96}
Richard Jones.
\newblock {\em Garbage Collection: Algorithms for Automatic Dynamic Memory
  Management}.
\newblock John Wiley and Sons, July 1996.
\newblock With a chapter on Distributed Garbage Collection by Rafael Lins.
  Reprinted 1997 (twice), 1999, 2000.

\bibitem{GC2011}
Richard Jones, Anthony Hosking, and Eliot Moss.
\newblock {\em The Garbage Collection Handbook: The Art of Automatic Memory
  Management}.
\newblock Applied Algorithms and Data Structures. Chapman and Hall/CRC, 1
  edition, 2011.

\bibitem{Kaminski1994329}
Michael Kaminski and Nissim Francez.
\newblock Finite-memory {A}utomata.
\newblock {\em {T}heoretical {C}omputer {S}cience}, 134(2):329 -- 363, 1994.

\bibitem{Kiehn05}
Astrid Kiehn and S.~Arun-Kumar.
\newblock Amortised bisimulations.
\newblock In {\em FORTE 2005}, volume 3731 of {\em LNCS}, pages 320--334, 2005.

\bibitem{KobayashiPT:linearity}
Naoki Kobayashi, Benjamin~C. Pierce, and David~N. Turner.
\newblock Linearity and the pi-calculus.
\newblock {\em ACM Trans. Program. Lang. Syst.}, 21(5):914--947, 1999.

\bibitem{Kobayashi:hybrid}
Naoki Kobayashi and Davide Sangiorgi.
\newblock A hybrid type system for lock-freedom of mobile processes.
\newblock {\em ACM Trans. Program. Lang. Syst.}, 32(5):1--49, 2010.

\bibitem{Kobayashi:resourceusage}
Naoki Kobayashi, Kohei Suenaga, and Lucian Wischik.
\newblock Resource usage analysis for the pi-calculus.
\newblock {\em Logical Methods in Computer Science}, 2(3), 2006.

\bibitem{LuttgenV06}
Gerald L{\"u}ttgen and Walter Vogler.
\newblock Bisimulation on speed: A unified approach.
\newblock {\em Theor. Comput. Sci.}, 360(1-3):209--227, 2006.

\bibitem{Milner99}
R.~Milner.
\newblock {\em Communicating and mobile systems: the {$\pi$}-calculus}.
\newblock Cambridge Univ., 1999.

\bibitem{Pierce:2002:TPL}
Benjamin~C. Pierce.
\newblock {\em Types and programming languages}.
\newblock MIT Press, Cambridge, MA, USA, 2002.

\bibitem{Pierce:2004:ATT}
Benjamin~C. Pierce.
\newblock {\em Advanced Topics in Types and Programming Languages}.
\newblock The MIT Press, Cambridge, MA, USA, 2004.

\bibitem{PierceS96}
Benjamin~C. Pierce and Davide Sangiorgi.
\newblock Typing and subtyping for mobile processes.
\newblock {\em Mathematical Structures in Computer Science}, 6(5):409--453,
  1996.

\bibitem{Pym:calculus}
David Pym and Chris Tofts.
\newblock A calculus and logic of resources and processes.
\newblock {\em Form. Asp. Comput.}, 18(4):495--517, 2006.

\bibitem{AliasTypes:SmithWM00}
Frederick Smith, David Walker, and J.~Gregory Morrisett.
\newblock Alias types.
\newblock In {\em ESOP}, volume 1782 of {\em LNCS}, pages 366--381. Springer,
  2000.

\bibitem{teller:resourcespi}
David Teller.
\newblock Recollecting resources in the pi-calculus.
\newblock In {\em Proceedings of IFIP TCS 2004}, pages 605--618. Kluwer
  Academic Publishing, 2004.

\bibitem{TerauchiAiken08}
T.~Terauchi and A.~Aiken.
\newblock A capability calculus for concurrency and determinism.
\newblock {\em TOPLAS}, 30(5):1--30, 2008.

\bibitem{wadler:use}
Philip Wadler.
\newblock Is there a use for linear logic?
\newblock In {\em PEPM}, pages 255--273, 1991.

\bibitem{Nobuko:dependent}
Nobuko Yoshida.
\newblock Channel dependent types for higher-order mobile processes.
\newblock {\em SIGPLAN Not.}, 39(1):147--160, 2004.

\bibitem{Yoshida07:linearity}
Nobuko Yoshida, Kohei Honda, and Martin Berger.
\newblock Linearity and bisimulation.
\newblock {\em Journal of Logic and Algebraic Programming}, 72(2):207 -- 238,
  2007.

\end{thebibliography}



\end{document}